%% file: paper.tex
\documentclass[sigconf, pageno]{acmart}

\AtBeginDocument{%
  }

\usepackage{xfrac}
\usepackage{flushend}
\usepackage{amsmath}
\usepackage{mathtools}
\usepackage{balance}
\usepackage{lipsum}  
\usepackage{cleveref}
\usepackage{graphicx}
\usepackage{float}
\usepackage{enumitem}
\usepackage{stfloats}

\usepackage{algorithm}
\usepackage[noend]{algpseudocode}
\usepackage[thinc]{esdiff}
\usepackage{graphicx}
\usepackage[english]{babel}
\usepackage[latin1]{inputenc}
\usepackage[T1]{fontenc}
\usepackage{amsmath}
\usepackage{hyperref}
\usepackage{color}
\usepackage[thinc]{esdiff}

\usepackage{subcaption}
\usepackage{flushend}
\usepackage{amsfonts}
\usepackage{amstext}
\usepackage{xspace}
\usepackage{nicefrac}
\usepackage{graphicx}
\usepackage{url}
\usepackage{graphics}
\usepackage{colordvi}

\copyrightyear{2024}
\acmYear{2024}
\setcopyright{rightsretained}
\acmConference[E-Energy '24]{The 15th ACM International Conference on Future and Sustainable Energy Systems}{June 4--7, 2024}{Singapore, Singapore} 
\acmBooktitle{The 15th ACM International Conference on Future and Sustainable Energy Systems (E-Energy '24), June 4--7, 2024, Singapore, Singapore}
\acmDOI{10.1145/3632775.3661942} 
\acmISBN{979-8-4007-0480-2/24/06}



\makeatletter
\newcommand{\algmargin}{\the\ALG@thistlm}
\makeatother
\newlength{\whilewidth}
\algdef{SE}[parWHILE]{parWhile}{EndparWhile}[1]
{\parbox[t]{\dimexpr\linewidth-\algmargin}{%
		\hangindent\whilewidth\strut\algorithmicwhile\ #1\ \algorithmicdo\strut}}{\algorithmicend\ \algorithmicwhile}%
\algdef{SE}[SUBALG]{Indent}{EndIndent}{}{\algorithmicend\ }%
\algtext*{Indent}
\algtext*{EndIndent}
\algnewcommand{\parState}[1]{\State%
	\parbox[t]{\dimexpr\linewidth-\algmargin}{\strut #1\strut}}
\makeatother
\DeclareMathAlphabet{\mathmybb}{U}{bbold}{m}{n}

\newtheorem{theorem}{Theorem}[section]

\newtheorem{lemma}[theorem]{Lemma}
\newtheorem*{lemma*}{Lemma}

\newtheorem{definition}[theorem]{Definition}

\newcommand{\cali}{\mathcal{I}}

\newcommand{\OCS}{\ensuremath{\mathsf{OCS}}\xspace}

\newcommand{\OCSU}{\ensuremath{\mathsf{OCSU}}\xspace}
\newcommand{\OCSUmin}{\ensuremath{\mathsf{OCSU}}\xspace}
\newcommand{\RORO}{\ensuremath{\texttt{RORO}}\xspace}
\newcommand{\ROROcmax}{\ensuremath{\texttt{RORO}_{\text{cmax}}}\xspace}
\newcommand{\ROROcmin}{\ensuremath{\texttt{RORO}_{\text{cmin}}}\xspace}
\newcommand{\ROROpred}{\ensuremath{\texttt{RORO}_{\text{pred}}}\xspace}
\newcommand{\cmax}{\ensuremath{c_{\text{max}}}\xspace}
\newcommand{\cmin}{\ensuremath{c_{\text{min}}}\xspace}
\newcommand{\tmpalgg}{\ensuremath{\hat{\texttt{ALG}}_2}\xspace}
\newcommand{\combalg}{\ensuremath{\texttt{RORO}_{\text{robust}}}\xspace}
\newcommand{\aug}{\ensuremath{\texttt{LACS}}\xspace}
\newcommand{\augd}{\ensuremath{\texttt{D-LACS}}\xspace}
\newcommand{\OPT}{\ensuremath{\texttt{OPT}}\xspace}
\newcommand{\lacs}{\ensuremath{\texttt{LACS}}\xspace}

\newcommand{\OWT}{\ensuremath{\texttt{OWT}_{\text{pred}}}\xspace}
\newcommand{\threshold}{\texttt{Single Threshold}\xspace}
\newcommand{\agnostic}{\texttt{Carbon Agnostic}\xspace}
\newcommand{\carbonscaler}{\texttt{CarbonScaler}\xspace}

\newcommand{\sref}[2]{\hyperref[#2]{#1 \ref{#2}}}

\DeclareMathOperator*{\argmin}{arg\,min}

\usepackage{etoolbox}
\makeatletter
\patchcmd{\hyper@makecurrent}{%
    \ifx\Hy@param\Hy@chapterstring
        \let\Hy@param\Hy@chapapp
    \fi
}{%
    \iftoggle{inappendix}{
        \@checkappendixparam{chapter}%
        \@checkappendixparam{section}%
        \@checkappendixparam{subsection}%
        \@checkappendixparam{subsubsection}%
        \@checkappendixparam{paragraph}%
        \@checkappendixparam{subparagraph}%
    }{}%
}{}{\errmessage{failed to patch}}

\newcommand*{\@checkappendixparam}[1]{%
    \def\@checkappendixparamtmp{#1}%
    \ifx\Hy@param\@checkappendixparamtmp
        \let\Hy@param\Hy@appendixstring
    \fi
}
\makeatletter

\newtoggle{inappendix}
\togglefalse{inappendix}

\apptocmd{\appendix}{\toggletrue{inappendix}}{}{\errmessage{failed to patch}}

\begin{document}


\author{Roozbeh Bostandoost}
\orcid{0000-0003-0959-6763}
\affiliation{%
  \institution{University of Massachusetts Amherst}
  \country{USA}
  }

\author{Adam Lechowicz}
\orcid{0000-0002-7774-9939}
\affiliation{%
  \institution{University of Massachusetts Amherst}
  \country{USA}
  }

\author{Walid A. Hanafy}
\orcid{0000-0001-5765-8194} 
\affiliation{%
  \institution{University of Massachusetts Amherst}
  \country{USA}
  }

\author{Noman Bashir}
\orcid{0000-0001-9304-910X}
\affiliation{%
  \institution{Massachusetts Institute of Technology}
  \country{USA}
  }

\author{Prashant Shenoy}
\orcid{0000-0002-5435-1901}
\affiliation{%
  \institution{University of Massachusetts Amherst}
  \country{USA}
  }

\author{Mohammad Hajiesmaili}
\orcid{0000-0001-9278-2254}
\affiliation{%
  \institution{University of Massachusetts Amherst}
  \country{USA}
  }

\title{\texttt{LACS}: Learning-Augmented Algorithms for Carbon-Aware Resource Scaling with Uncertain Demand}

\renewcommand{\shortauthors}{Bostandoost et al.}

\begin{abstract}
\input{0-abstract}
\end{abstract}

\maketitle

\section{Introduction} \label{sec:intro}
\input{1-introduction}

\section{Problem Statement} \label{sec:prob}
\input{2-problem}

\section{Algorithm Descriptions} \label{sec:alg_and_res}
\input{alg_and_results}

\section{Theoretical Results} \label{sec:analysis}
\input{analysis}

\section{Experimental Results} \label{sec:exp}
\input{5-exp}

\vspace{-0.1cm}
\section{Related Work} \label{sec:relwork}
\input{6-relwork}

\section{Conclusion} \label{sec:conclusion}
\input{conclusion}

\begin{acks}
Mohammad Hajiesmaili acknowledges this work is supported by the U.S. National Science Foundation (NSF) under grant numbers CAREER-2045641, CPS-2136199, CNS-2106299, CNS-2102963, and NGSDI-2105494. Prashant Shenoy's research is supported by NSF grants 2213636, 2105494, 2021693, 2020888, DOE grant DE-EE0010143, and VMware.

This material is based upon work supported by the U.S. Department of Energy, Office of Science, Office of Advanced Scientific Computing Research, Department of Energy Computational Science Graduate Fellowship under Award Number DE-SC0024386.
\end{acks}

\section*{Disclaimers}
This report was prepared as an account of work sponsored by an agency of the United States Government. Neither the United States Government nor any agency thereof, nor any of their employees, makes any warranty, express or implied, or assumes any legal liability or responsibility for the accuracy, completeness, or usefulness of any information, apparatus, product, or process disclosed, or represents that its use would not infringe privately owned rights. Reference herein to any specific commercial product, process, or service by trade name, trademark, manufacturer, or otherwise does not necessarily constitute or imply its endorsement, recommendation, or favoring by the United States Government or any agency thereof. The views and opinions of authors expressed herein do not necessarily state or reflect those of the United States Government or any agency thereof.

\clearpage
\bibliographystyle{ACM-Reference-Format}
\bibliography{paper}
\clearpage
\onecolumn

\appendix
\section*{Appendix}
\input{Z-appendix}

\end{document}

%% file: 0-abstract.tex
Motivated by an imperative to reduce the carbon emissions of cloud data centers,
this paper studies the online carbon-aware resource scaling problem with unknown job lengths (\OCSU) and applies it to carbon-aware resource scaling for executing computing workloads.
The task is to dynamically scale resources (e.g., the number of servers) assigned to a job of unknown length such that it is completed before a deadline, with the objective of reducing the carbon emissions of executing the workload.
The total carbon emissions of executing a job originate from the emissions of running the job and excess carbon emitted while switching between different scales (e.g., due to checkpoint and resume). 
Prior work on carbon-aware resource scaling has assumed accurate job length information, while other approaches have ignored switching losses and require carbon intensity forecasts. 
These assumptions prohibit the practical deployment of prior work for online carbon-aware execution of scalable computing workload. 

We propose \lacs, a theoretically robust, learning-augmented algorithm that solves \OCSU. To achieve improved practical average-case performance, \lacs integrates machine-learned predictions of job length. To achieve solid theoretical performance, \lacs extends the recent theoretical advances on online conversion with switching costs to handle a scenario where the job length is unknown. 
Our experimental evaluations demonstrate that, on average, the carbon footprint of \lacs lies within 1.2\% of the online baseline that assumes perfect job length information and within 16\% of the offline baseline that, in addition to the job length, also requires accurate carbon intensity forecasts. Furthermore, \lacs achieves a 32\% reduction in carbon footprint compared to the deadline-aware carbon-agnostic execution of the job.   

%% file: 1-introduction.tex

The exponential growth in computing demand and the resulting energy consumption has enhanced focus on its climate and sustainability implications~\cite{NSF:2022:DESC, Bashir:2022, Monserrate:2022, Strubell:2020}. 
The focus has been magnified since the widespread adoption of generative artificial intelligence tools, e.g., ChatGPT~\cite{Luccioni:2023}. 
Key stakeholders, including policymakers and end users, are trying to create direct incentives, through caps or taxes~\cite{EU:2024:CBAM} on carbon emissions, and indirect incentives, through social pressure, to curb the climate impact of this unprecedented demand. 
In response, researchers and other stakeholders in computing are trying to reduce the carbon footprint during its various lifecycle stages, including manufacturing~\cite{Tannu:2023, Gupta:2022:Chasing}, operations~\cite{Wiesner:21, Souza:23, Hanafy:23:CarbonScaler, Arora:2023}, and end-of-life~\cite{Gupta:2022:Act, Switzer:2023}. 
While computing's carbon footprint at all stages are important to address, this paper focuses on the operational carbon footprint arising from using electricity to run computing workloads, as it contributes significantly to computing's total carbon footprint~\cite{Gupta:2022:Chasing}.

Beyond improving the algorithmic efficiency of computing workloads and the energy efficiency of its hardware, computing's operational footprint can be reduced by enhancing the carbon efficiency of grid-supplied electricity (kilowatt-hours of energy produced per unit of carbon emissions)~\cite{Bashir:2022}. 
One approach is to use low-carbon energy sources, such as solar, wind, and nuclear, for electricity generation. However, as pathways to 100\% renewable energy adoption remain challenging and costly~\cite{Arent:2022, Cole:2021}, this approach may not entirely eliminate electricity's carbon emissions for all the locations in the near future~\cite{Holland:2022}. 
A complementary approach is to improve the \emph{effective} carbon efficiency of the energy \emph{used} for computing by running flexible computing workloads when and where low-carbon energy is available.
Prior work has proposed leveraging computing workloads' spatiotemporal flexibility~\cite{Wiesner:21, Lechowicz:23, sukprasert2023quantifying, Kim:2023, Dodge:2022} and resource elasticity~\cite{Hanafy:23:CarbonScaler, Lechowicz:24, Hanafy:23:War, Souza:23} to reduce their carbon footprint. 

In this work, we focus on exploiting resource elasticity for carbon-aware execution of scalable computing workloads with unknown job lengths and future carbon intensity. The carbon-aware resource scaling work requires determining the scale factor, i.e., the number of cores or servers, at each time step before the job deadline while considering the scalability properties of the job. Initial work on carbon-aware resource scaling by \citet{Hanafy:23:CarbonScaler} leverages carbon intensity forecasts and develops an offline optimal approach to determine the best scale factor for a job at each time step before the deadline. The authors ignore switching overhead and assume the deadline is provided at job submission time. \citet{Lechowicz:24} introduce and study an online class of problems motivated by carbon-aware resource scaling and electric vehicle (EV) charging applications. The proposed algorithms can be used to determine the optimal scale factor without requiring carbon intensity forecasts while considering the energy inefficiencies in resource scaling through a convex cost function, which is revealed online. However, a key drawback for both studies is that they assume each job's length (i.e., the amount of work to be done) is known.

Estimating the duration of a job remains a challenging problem in cluster and cloud computing. The unavailability of data on job attributes, lack of diversity in the available data, variations in the characteristics of the jobs submitted to the cluster over time, and skewed distribution of users submitting the jobs means that job length predictions, even when using machine learning techniques, remain highly inaccurate and cannot be used for scheduling purposes~\cite{Kuchnik:2019}.
\citet{Ambati:2021} showed that job length estimates at job submission time can have more than 140\% mean absolute percentage error. The inaccuracy can be further amplified as the properties of the job or the hardware it runs on change across different runs. As a result, practical algorithms need to work without assuming that accurate job length is available. In this work, we assume that the job lengths are unknown and only the lower- and upper-bounds on job length are available. This is a reasonable assumption as classifying a job as short or long tends to be highly accurate~\cite{Ambati:2021, Zrigui:2022}. 


The existing theoretical literature that studies similar problems crucially does not consider uncertainty in the job length~\cite{ElYaniv:01, Lorenz:08, SunZeynali:20, Lechowicz:23, Lechowicz:24} (i.e., the job length is precisely known to the algorithm).  In perhaps the closest setting to carbon-aware resource scaling, \citet{Lechowicz:24} presents the online conversion with switching costs (\OCS) problem.  Uncertainty about the job length breaks many of the assumptions in \OCS; under a deadline and without precise information on job length, the algorithm will either run too little of the job at a low carbon intensity or run too much of the job at a high carbon intensity.
Without even approximate knowledge of job length (i.e., lower- and upper-bounds on job lengths) and under a deadline, the existing algorithms, e.g., for \OCS, may not complete the job by the deadline and thus fail to provide worst-case guarantees.

It is worth noting that providing worst-case guarantees is an important consideration for algorithms solving this problem. In-production resource managers such as Borg~\cite{Verma:2015} and Resource Central~\cite{Coretz:2017}, prefer deploying techniques with safety guarantees~\cite{Rzadca:2020}.  Techniques which fully rely on machine learning (ML) perform well in the average-case, but can result in extremely poor outcomes in the worst-case (e.g., when presented with out of distribution data), making them undesirable for production deployment~\cite{Kuchnik:2019}.


\noindent\textbf{Contributions.}
This paper proposes \lacs, a learning-augmented algorithm for online carbon-aware resource scaling with unknown job lengths (\OCSU) that uses ML predictions of job lengths, which are potentially inaccurate, for resource scaling.  We then analyze the theoretical performance of \lacs using the framework of competitive analysis~\cite{Borodin:2005}, and its learning-augmented variants~\cite{Lykouris:18, Purohit:18}. We also evaluate the practical performance of \lacs using real-world data traces on an extensive set of experimental scenarios.


\vspace{-0.08cm}
\begin{enumerate}[leftmargin=*]
    \item \textit{Theoretical analysis of \kern0.1em \lacs: Bounded robustness and consistency.} The theoretical analysis of \lacs leverages and advances the emerging framework of robustness-consistency for learning-augmented online algorithms~\cite{Lykouris:18, Purohit:18} and the recent competitive results for \OCS~\cite{Lechowicz:24}, which tackle a simplified \OCSU with known job lengths. 
    We take multiple algorithmic steps to achieve bounded competitive guarantees for \lacs. First, we consider two extreme scenarios of the ``Ramp-On Ramp-Off'' (\texttt{RORO}) framework proposed by~\citet{Lechowicz:24} to design \textit{robust} baseline algorithms of assuming actual job lengths equal to the given lower and upper bound values of job lengths. 
    Then, using these two extremes as the baseline, we introduce an additional layer by integrating a job length predictor, and finally, \lacs leverages these robust baseline algorithms, combined with potentially inaccurate predictions of the job length, to achieve improved consistency in the average case (i.e., when predictions are of high quality), while retaining worst-case guarantees (i.e., robustness given by the baselines).

    \item \textit{Extensive trace-driven experiments.} We then evaluate the performance of \lacs against state-of-the-art methods in carbon-aware scheduling and online scheduling literature using three years of real carbon intensity traces from Electricity Maps~\cite{electricity-map}, using an extensive range of experimental scenarios. 
    In a set of representative experiments, we demonstrate that, on average, the carbon footprint of \lacs lies within 1.2\% of the online baseline that assumes perfect job length information and within 16\% of the offline baseline that also requires accurate carbon intensity forecasts. \lacs achieves a 32\% reduction in carbon footprint compared to the deadline-aware carbon-agnostic job execution. 
\end{enumerate}



%% file: 2-problem.tex

In the following, we will consider a server cluster that is used to run batch jobs. 
We introduce the online carbon-aware resource scaling problem with unknown job lengths (\OCSUmin), where the goal is to complete a job with total length $c$ (where $c$ is unknown) while minimizing its overall carbon emissions to complete the job.  
At each time step $t \in [T]$, a convex cost function $g_t(\cdot)$ arrives, which is a combination of both the time-varying carbon intensity (e.g., of the electricity grid) and the job's time-varying scaling profile (e.g., how parallelizable the job is).  
In response, the algorithm must choose the amount of server resources $x_t$ that will be given to the job in time step $t$, where  ${x_t \in [0,d_t]}$, which produces carbon emissions given by $g_t(x_t)$.  Here, $d_t$ is the maximum amount of job that can be scheduled for the job at time $t$ (rate constraint).  Intuitively, we assume that $g_t(0) = 0$ for any cost function (i.e., completing none of the job does not emit any carbon), and $g_t(x) \geq 0$ for any valid $x > 0$.  

Whenever the allocation decision changes in adjacent time steps, it incurs extra carbon emissions caused by switching denoted by $\beta \lvert x_t - x_{t-1} \rvert$.  For the model, we let $x_0 = 0$ and $x^{T+1} = 0$, which require the algorithm to incur \textit{some} switching carbon emissions to ``turn on'' and ``turn off'', respectively.  The parameter $\beta$ can be interpreted as a linear coefficient that charges the algorithm proportionally to the amount of scaling between consecutive time steps, based on, e.g., the carbon emitted due to overhead of changing the resource allocation or checkpointing/resuming the job.\footnote{
We note that while \OCSU assumes the emissions to ``turn on'' is equivalent to the cost to ``turn off'' (i.e., the switching emissions are symmetric), it can be extended to cases where the switching emissions are time-varying or asymmetric by letting $\beta \lvert x_t - x_{t-1} \rvert$
be an upper bound on the actual switching emissions, as discussed in~\cite{Lechowicz:24}}.
In summary, the offline version of \OCSUmin is formalized as:
\begin{align}
\OCSU: \min_{\{x_t\}_{ t \in [T] }} & \kern-1.5em \underbrace{ \ \sum\nolimits_{t=1}^T g_t( x_t ) }_{\text{Execution carbon emissions}} + \underbrace{ \sum\nolimits_{t=1}^{T+1}\beta \lvert x_t - x_{t-1} \rvert }_{\text{Switching carbon emissions}} \kern-0.5em , \label{align:obj} \\ 
\text{s.t., }  &\underbrace{\sum\nolimits_{t=1}^T x_t \geq c,}_{\text{Job completion constraint}} x_t \in [0, d_t], \ \forall t \in [T]. \label{align:deadline}
\end{align}

In this paper, we focus on designing algorithms for the online version of this problem, where the algorithm must choose an irrevocable $x_t$ at each time step without knowledge of future cost functions or the total job length $c$.  
Each cost function $g_t(\cdot)$ is revealed online at the start of time step $t$, and the actual job length $c$ is revealed when the constraint in \autoref{align:deadline} is satisfied (i.e., the job has been finished). We note that \OCSUmin builds on the existing formulation of \textit{online conversion with switching costs} ($\mathsf{OCS}$), introduced by~\citet{Lechowicz:24}.  The minimization variant of $\mathsf{OCS}$ is a special case of \OCSUmin where the online algorithm has perfect knowledge of the actual job length $c$. The core notations are summarized in \autoref{table:notation_definition}.

\begin{table}[t]
    \centering
    \footnotesize
    \captionof{table}{Summary of Notations}
    \vspace{-0.3cm}
    \label{table:notation_definition}
    \begin{tabular}{|l |l |}
    \hline
    \textbf{Notation} & \textbf{Definition} \\ \hline
    $c$, $c_{\min}$, $c_{\max}$  & Actual job length, minimum job length, maximum job length  \\ \hline    
    $\hat{c}$  & Job length prediction \\ \hline
    $\beta$  & Switching emissions coefficient    \\ \hline
    $\text{CI}_t$  & Carbon intensity (e.g., in gCO$_2$eq./kWh) at time $t$ \\ \hline
    $E$  & Energy used (e.g. in kWh) by one unit of threads/cores/servers \\ \hline
    $T$  & Deadline: Maximum time (e.g., 24 hours) the job is allowed to run  \\
         & after being submitted \\ \hline
    $g_t(\cdot)$  & Convex cost function that arrives at time $t$  \\ \hline
    $x_t$  & Amount of the job that is scheduled to be done at time $t$  \\ \hline
    $w^{(t)}$  & Total amount of the job that has been completed up to time $t$    \\ \hline
    $r_t$  & Maximum resource available at time $t$\\ \hline  
    $d_t$  & Maximum amount of the job that can be scheduled at time $t$    \\ \hline    
    \end{tabular}
    \vspace{-0.4cm}
\end{table}

\noindent\textbf{Details of the cost function.}
To formalize the definition of the cost functions $g_t(\cdot)$, let $h_t(\cdot)$ denote the scaling profile of the job at time $t$. For a given $s$, as the number of allocated resources (e.g., cores or servers), $h_t(s)$ provides the throughput (e.g., amount of work done) $x$ at time $t$.  Naturally, $h_t(\cdot)$ is a concave function since adding resources has diminishing returns even in highly parallel computing workloads. Note that $h_t^{-1}(x)$ (i.e., the inverse of the scaling function) maps a throughput $x$ to the necessary amount of allocated resources $s$ at time $t$.  
There is a \textit{carbon emissions} associated with the resource allocation amount of $s$. Let $E$ denote the energy used (e.g., in kWh) by one unit of resource, and let $\text{CI}_t$ denote the carbon intensity (e.g., in gCO$_2$eq./kWh) at time $t$.  Then, the cost function $g_t(\cdot)$ of \OCSU is defined as:
\[
g_t( x ) = \text{CI}_t \times E \times h_t^{-1} ( x ), \  \ t \in [T].
\]
Since the scaling profile $h_t(\cdot)$ is known, the primary \textit{unknown} quantity which makes $g_t(\cdot)$ an online input is the unknown carbon intensity $\text{CI}_t$. 
Furthermore, considering that the available resources constrain the quantity of job that can be scheduled in each time slot, we introduce $r_t$ to represent the maximum resources available at a given time. Consequently, the maximum amount of the job that can be scheduled at time $t$, denoted as $d_t$ (maximum rate), equals $h_t(r_t)$.

\smallskip

\noindent\textbf{Assumptions.}
We make the following assumptions in the paper. 


$\bullet$ \noindent\textit{Assumption 1.}
Although the job length $c$ is unknown, we assume that the value of $c$ is bounded between a minimum and a maximum length $c_{\min}$ and $c_{\max}$, i.e., $c \in [c_{\min}, c_{\max}]$.  Without loss of generality, we assume $c_{\min} = 1$, which further gives that ${\cmax > \cmin = 1}$.

$\bullet$ \noindent\textit{Assumption 2.}
We assume that the derivatives of the cost function are bounded, i.e., $L \leq d g_t / d x_t \leq U \ \forall t \in [T]$ on the interval $x_t \in [0,d_t]$, where $L$ and $U$ are known positive constants.  This is a necessary assumption for any competitive algorithm, as shown in~\cite{ElYaniv:01, Zhou:08, SunZeynali:20}; otherwise, no online algorithm can achieve a bounded competitive ratio. 

$\bullet$ \noindent\textit{Assumption 3.}
The switching emissions coefficient $\beta$ is known to the algorithm, and is bounded within an interval ($\beta \in [0, \nicefrac{(U-L)}{2})$), as in~\cite{Lechowicz:24}.  If $\beta$ exceeds $\nicefrac{(U-L)}{2}$, any competitive algorithm should only consider the excess emissions due to switching because the overhead is very large; hence the decision-making becomes trivial.    


$\bullet$ \noindent\textit{Assumption 4.}
\OCSUmin requires the algorithm to complete the entire job before the sequence ends at ``deadline'' $T$.
If the scheduler has completed $w^{(j)}$ amount of the job at time $j$, a \emph{compulsory execution} begins whenever $(T - j - 1) < (c - w^{(j)})$ (i.e., when the remaining time steps are barely enough to complete the job).  During this compulsory execution, a carbon-agnostic algorithm takes over and runs the job with the maximum available resources in the remaining time steps.  Although $c$ is unknown to the algorithm, for modeling purposes, we assume that the algorithm will begin this compulsory execution when the remaining steps are sufficient to fulfill the worst-case job length, which is given by $(\cmax - w^{(j)})$.

$\bullet$ \noindent\textit{Assumption 5.}
In an application such as carbon-aware resource scaling, the deadline $T$ is typically known in advance. Our algorithms, however, do not require this assumption to be true.  If $T$ is unknown, we assume that the algorithm is given a signal to indicate that the deadline is coming up and that compulsory execution should begin to finish a job with worst-case size \cmax.

$\bullet$ \noindent\textit{Assumption 6.}
We assume that the job execution time horizon has sufficient \textit{slackness}, i.e., the compulsory execution does not make up a large fraction of the sequence -- otherwise, the problem is trivial.  Formally, we have that the earliest time step $j'$ at which the compulsory execution begins (i.e., the first time step such that $(T - j' - 1) < \cmax$) is $j' \gg 1$, which implies that $T$ is sized appropriately for the job. This assumption is reasonable in practice, since if $T$ is small or $c$ is large, the job's \textit{temporal flexibility} will be low, so even a solution with perfect knowledge of future carbon intensity values will be unable to take advantage of time-varying carbon intensity to reduce emissions.

\noindent\textbf{Competitive analysis.}
We tackle \OCSU from the perspective of competitive analysis, where the objective is to design an online algorithm that maintains a small \textit{competitive ratio}~\cite{Borodin:92}, defined as:
\begin{definition}[Competitive Ratio] \label{def:cr}
We denote $\OPT(\mathcal{I})$ as the offline optimum on the input $\mathcal{I}$, and $\texttt{ALG}(\mathcal{I})$ represents the profit obtained by an online algorithm (\texttt{ALG}) on that input. 
Formally, letting $\Omega$ denote the set of all possible inputs, we say that \texttt{ALG} is $\eta$-competitive if the following holds:
$\texttt{CR} = \max_{\mathcal{I} \in \Omega}\nicefrac{\texttt{ALG}(\mathcal{I})}{\OPT(\mathcal{I})} = \eta.$
Observe that \texttt{CR} is greater than or equal to one. The smaller it is, the closer the algorithm is to the optimal solution.
\end{definition}

\noindent\textbf{Learning-augmented competitive algorithms.}
In the nascent literature on learning-augmented algorithms~\cite{Lykouris:18, Purohit:18}, algorithms are evaluated through the metrics of \textit{consistency} and \textit{robustness}.  Intuitively, these quantities measure how close a learning-augmented algorithm's solution is to that of the offline optimal solution when the prediction is accurate (consistency) and how far an algorithm's solution can be from the optimal solution in the worst case when the prediction is erroneous (robustness). 
\begin{definition}[Consistency and Robustness] \label{def:const-rob}
Formally, an algorithm is $b$-consistent if it is $b$-competitive with respect to an accurate prediction and $r$-robust if it is $r$-competitive regardless of the quality of the prediction.
\end{definition}

%% file: alg_and_results.tex
\begin{algorithm}[!t]
    \small
	\caption{Online ramp-on, ramp-off (\RORO) algorithm \cite{Lechowicz:24}}
	\label{alg:roro}
	\begin{algorithmic}[1]
		\State \textbf{input:} pseudo-cost threshold $\phi(w)$
        \State \textbf{initialization:} initial decision $x_0 = 0$, initial progress $w^{(0)} = 0$;
		\While{cost function $g_t(\cdot)$ is revealed and $w^{(t-1)} < c$}
		\State solve \textbf{pseudo-cost minimization problem} to obtain decision $x_t$, 
        \begin{align}
            \quad {x}_t  &= \ \ \argmin_{\mathclap{ \quad x \in [0, \min (1- w^{(t-1)}, d_t) ]}} \ \  g_t( x ) + \beta \lvert x - x_{t-1} \rvert - \int_{w^{(t-1)}}^{w^{(t-1)} + x}\phi(u) du.
        \end{align}
        \State update the progress $w^{(t)} = w^{(t-1)} + x_t$;
        \EndWhile
	\end{algorithmic}
\end{algorithm}

In this section, we introduce \lacs, a \texttt{\textbf{L}}earning-\texttt{\textbf{A}}ugmented \texttt{\textbf{C}}arbon-aware Resource \texttt{\textbf{S}}caling algorithm that solves \OCSU. To achieve the best of both worlds on satisfactory practical performance and theoretical worst-case guarantees,
\lacs integrates \textit{predictions of the job length} into its operation by combining the decisions of an algorithm that assumes the prediction is correct with the decisions of two competitive baselines. By combining these strategies, \aug can improve its performance significantly when the predictions are accurate while maintaining worst-case competitive guarantees.  Below, we start by reviewing approaches from prior work that inform our design of the competitive baselines.





\subsection{Algorithmic Background} \label{subsec:alg_background}

The competitive baselines we consider in the next section build on prior work, specifically the ``ramp-on, ramp-off'' (\RORO) framework proposed by~\cite{Lechowicz:24} that achieves the optimal competitive ratio for \OCS, as a simplified version of \OCSU that assumes job length is known a priori to the online algorithm. 
In the \RORO framework~\cite{Lechowicz:24}, whenever an input arrives online, the algorithm solves a \textit{pseudo-cost minimization} problem to determine the amount of job to run at time $t$ (denoted by $x_t \in [0,d_t]$).  The progress $w^{(t)}$ denotes the fraction of the total job that has been completed up to time~$t$.  This pseudo-cost minimization design generalizes the concept of threshold-based algorithm design -- at each time step, when a cost function arrives, the pseudo-cost of a particular decision $x$ is defined as the actual carbon emissions of running $x$ amount of the job (including both the execution and the switching emissions), minus a threshold value which describes the exact amount which should be allocated to maintain a certain competitive ratio.

This pseudo-cost acts as an incentive to prevent the algorithm from ``waiting too long'' to run the job.  Intuitively, if an algorithm naively minimizes the cost function $g_t$ at each time step (resulting in decisions $x_t = 0$ for all $t \in [T]$), it will be required to complete the entire job during compulsory execution during a potentially bad period for carbon intensity.  The pseudo-cost minimization provides a framework that balances the extreme options of allocating ``too much'' early on and waiting indefinitely.  Whenever the carbon intensity is ``attractive enough,'' the pseudo-cost minimization finds the best decision that allocates just enough resources given the current carbon intensity to maintain competitiveness. In the setting where the job length is known, we summarize the \RORO algorithm in \autoref{alg:roro}.

To define this pseudo-cost minimization problem, the authors in \cite{Lechowicz:24} define a \textit{dynamic threshold function}, which essentially defines the highest carbon intensity deemed acceptable by \RORO at time $t$.  We note that in \OCS with known job lengths, $c$ is defined to be $1$ (without loss of generality). 
According to \cite[Definition 3.1]{Lechowicz:24}, the dynamic threshold for \OCS, for a job with length $c$, and for any progress $w \in [0, c]$ is defined as:
\begin{align}
    \phi_{\texttt{OCS}}(w) = U-\beta + \left( \frac{U}{\alpha}-U+2\beta \right) \exp \left( \frac{w}{c\alpha} \right), \label{eq:ocs-thresh}
\end{align}
where $\alpha$ is the competitive ratio defined as:
\begin{align}
    \alpha &= \left[W \left[ \left( \frac{2\beta}{U} + \frac{L}{U} - 1 \right) \exp \left(\frac{2\beta}{U} - 1 \right) \right] - \frac{2\beta}{U} + 1 \right]^{-1}\kern-1em . \label{eq:ocs-alpha}
\end{align}
In the above equation, $W(\cdot)$ is the Lambert $W$ function, defined as the inverse of $f(y)=ye^y$~\cite{Corless:96LambertW}.  Note that it is well-known that $W(x) \thicksim \ln (x)$~\cite{HoorfarHassani:08}.
Given this definition of $\alpha$, note that $\phi_{\OCS}(\cdot)$ is monotonically decreasing on the interval $w \in [0,c]$.

\subsection{\aug: A Learning-augmented Algorithm for Carbon-aware Resource Scaling}\label{sec:augmentation}

In this section, we describe the design of \aug, which uses predictions of the actual job length to significantly improve average-case performance (consistency) without losing worst-case guarantees (robustness).
We first introduce two baseline competitive algorithms before describing how \lacs leverages predictions to improve average-case performance without losing competitive guarantees.


\noindent\textbf{Competitive baseline algorithms.}
Here we present two adaptations of the \RORO framework detailed in \autoref{subsec:alg_background}, denoted by  \ROROcmax and \ROROcmin.
Since \OCSU introduces job length uncertainty, each of these adaptations considers an opposing extreme case for the job length.  We describe each variant in turn below.

\noindent\textit{\emph{\ROROcmax}: \RORO assuming maximum job lengths.}
\ROROcmax takes an optimistic approach by assuming every job has the maximum length \cmax. This aims to prepare for potentially long jobs by gathering enough resources.
\ROROcmax's assumption of worst-case job sizes makes it less conservative in terms of carbon intensities where it is willing to run the job. This gives it the flexibility to prepare for scenarios where long jobs (e.g., with length \cmax) do arise, although it risks ``overspending'' for shorter jobs.
We see the impact of this assumption in \ROROcmax's threshold function \autoref{eq:alg1-thresh}. Compared to alternatives like \RORO, which knows the exact job length, \ROROcmax's threshold reduces at a slower exponential rate as job progress increases -- intuitively, this is because \ROROcmax plans for a \textit{longer} job, which scales up the threshold function along the axis of $w$. 
Though this strategy may run the job when carbon intensities are ``too high,'' particularly for jobs that are much shorter than \cmax, it can be advantageous when job lengths do approach the maximum. In such situations, \ROROcmax may result in a more favorable outcome, avoiding the last-minute compulsory execution.

\begin{definition}\label{def:alg1-thresh}
The threshold function $\phi_1$ used by \ROROcmax for any progress $w \in [0, \cmax]$ is defined as:
    \begin{align}
        \phi_{1}(w) = U-\beta + \left( \frac{U}{\alpha}-U+2\beta \right) \exp \left(\frac{w}{\cmax\alpha} \right), \label{eq:alg1-thresh} 
    \end{align}
where $\alpha$ is defined in \autoref{eq:ocs-alpha}. 
\end{definition}

This approach captures one of two extreme cases that inform our algorithm design.  Next, we will ``flip'' these assumptions to capture the other extreme case in the \ROROcmin algorithm. 



\noindent\textit{\emph{\ROROcmin}: \RORO assuming minimum job lengths.}
\ROROcmin takes a pessimistic approach by assuming each job is as short as \cmin. This approach is efficient for handling shorter jobs since \ROROcmin is more conservative in choosing which carbon intensities are good enough to run the job. The threshold function decreases faster than \ROROcmax, which assumes the maximum job length.

However, when \ROROcmin encounters a longer job, its conservative nature can become a hindrance. \ROROcmin is, by design, reluctant to allocate resources liberally due to its lower threshold, potentially missing the chance to make significant progress on lengthy jobs early on. This may necessitate costly compulsory execution at the end of the time period.

To mitigate this, we scale the threshold by the ratio $\nicefrac{\cmax}{\cmin}$. This adjustment still allows the threshold to remain more cautious than \ROROcmax, but avoids the worst-case scenario for jobs that may turn out to be longer than \cmin. 
This remains economical for the assumed short jobs while also being flexible enough to accommodate the resource needs of unexpectedly longer jobs without resorting to compulsory executions at the end of the time period.

\begin{definition}\label{def:alg2-thresh}
The threshold function $\phi_2$ used by \ROROcmin for any progress $w \in [0, \cmax]$ is defined as:
    \begin{align}
        \phi_{2}(w) = U-\beta + \left( \frac{U}{\alpha'}-U+2\beta \right) \exp \left( \frac{w}{\cmax\alpha'} \right), \label{eq:alg2-thresh} 
    \end{align}
where $\alpha'$ is defined as follows:
{\small\begin{align}
    \alpha' &= \left[ \frac{\cmax}{\cmin}  W \left[ \frac{\cmin}{\cmax} \left( \frac{2\beta}{U} + \frac{L}{U} - 1 \right) \exp \left( \frac{\cmin}{\cmax} \left(\frac{2\beta}{U} - 1 \right) \right) \right] - \frac{2\beta}{U} + 1 \right]^{-1}\kern-1em . \label{eq:alphaprime}
\end{align}}
\end{definition}

These two approaches comprise our worst-case algorithm design.  While the competitive bounds of each algorithm differ, the empirical performance of each intuitively depends on the actual observed job lengths.  In practice, we may often have a relatively accurate prediction about a given job's length.  In the next section, we consider how this type of job length prediction can be incorporated into an algorithm design without losing worst-case guarantees.


\smallskip
\noindent\textbf{Learning-augmented algorithm design.}
Here we formalize our learning-augmented algorithm, referred to as \aug and outlined in \autoref{alg:aug}. This algorithm integrates insights from two robust algorithms, \ROROcmax and \ROROcmin, alongside predictions from an algorithm named \ROROpred.
\ROROpred is essentially a \RORO algorithm that uses the predicted job length $\hat{c}$ rather than the actual job length $c$. Although \ROROpred operates on predictions, it guarantees that the job is completed before the deadline  (\autoref{align:deadline}) by beginning a compulsory execution when the remaining time steps are enough to complete a job with length \cmax.

The combination of \ROROpred with competitive baselines \ROROcmax and \ROROcmin is designed to enhance average performance. Since \ROROcmax is tailored for longer jobs and \ROROcmin is more effective for shorter ones, we introduce an intermediate algorithm called \combalg which leverages strengths of both competitive baselines.  Let \(k \in [0,1]\) denote a decision factor, which dictates the proportion of the solution to derive from \ROROcmax ($\{x_{1t}\}_{\forall t \in [T]}$) or \ROROcmin ($\{x_{2t}\}_{\forall t \in [T]}$). Then \combalg constructs a solution ($\{\Tilde{x}_t\}_{\forall t \in [T]}$), where each $\Tilde{x}_t$ is defined as $\Tilde{x}_t = kx_{1t} + (1-k)x_{2t}, \forall t \in [T]$.

By unifying the two competitive baseline algorithms, we simplify the expression of \aug, which integrates these competitive decisions $\Tilde{x}_t$ with the decisions of \ROROpred as follows:  We set an \textit{augmentation factor} \(\lambda \in [0,1]\), which determines the influence of each algorithm on the final decision ($\lambda$ from \ROROpred, $(1-\lambda)$ from \combalg).  Intuitively, a larger value of $\lambda$ implies that \aug is closer to the prediction.
The result is a solution that benefits from the predictive strength of \ROROpred while maintaining the robustness provided by the combined solutions of \ROROcmax and \ROROcmin. 

In the following, we formalize our instantiation of \aug (and \combalg as a subroutine) for \OCSUmin, which includes definitions of $k, \lambda$, and the parameters of $\epsilon$ and $\gamma$ that they depend on.  

\begin{definition} \label{def:def-aug-con-rob-params}
   Let $\epsilon \in [0, |\alpha_1 - \alpha_2|]$, where $\alpha_1$ and $\alpha_2$ are the robust competitive ratios defined in~\autoref{eq:alg1-alpha} and \autoref{th:alg2-alpha}. 
   
   We set $k = 1 - \frac{\epsilon}{\alpha_1 - \alpha_2}$ (which is bounded in $[0,1]$) to form the solution ($\{\Tilde{x}_t\}_{\forall t \in [T]}$) obtained by \combalg.

   Let $\gamma \in [0, \alpha_1 - \texttt{sign}(\alpha_1 - \alpha_2)\epsilon - \alpha]$, where $\alpha$ is the robust competitive ratio defined in \autoref{eq:ocs-alpha}, and \texttt{sign}(x) is the \texttt{sign} function. 
   
   \aug sets a augmentation factor of $\lambda = 1 - \frac{\gamma}{\alpha_1 - \texttt{sign}(\alpha_1 - \alpha_2)\epsilon - \alpha}$ which is bounded in $[0,1]$.
\end{definition}

In the next section, we provide the consistency and robustness bounds for \aug.  Intuitively, the primary desiderata for \aug is to be able to nearly match the performance of an online algorithm which knows the exact job length (e.g., \RORO) when the job length predictions are correct, while preserving worst-case performance bounds in line with that of \ROROcmax and \ROROcmin.

\begin{algorithm}[!t]
    \small
	\caption{\aug: A \textit{learning-augmented} algorithm for \OCSUmin}
	\label{alg:aug}
	\begin{algorithmic}[1]
		\State \textbf{input:} The predicted solution $\{\hat{x}_t\}_{\forall t \in [T]}$ given by \ROROpred, competitive solutions $\{x_{1t}\}_{\forall t \in [T]}$ and $\{x_{2t}\}_{\forall t \in [T]}$ given by \ROROcmax and \ROROcmin, decision factor $k$, augmentation factor $\lambda$.
		\While{cost function $g_t(\cdot)$ is revealed and $w^{(t-1)} < c$}
            \State obtain robust decisions $x_{1t}$ and $x_{2t}$;
            \State $\Tilde{x}_t = kx_{1t} + (1-k)x_{2t}$; 
            \State obtain prediction decision $\hat{x}_t$;
            \State set the online decision as $x_t = \lambda \hat{x}_t + (1-\lambda) \Tilde{x}_t$;
        \State update the progress $w^{(t)} = w^{(t-1)} + x_t$;
        \EndWhile
	\end{algorithmic}
\end{algorithm}

%% file: analysis.tex



In this section, we state our main theoretical results.  We start with the competitive results for the competitive baseline \ROROcmax and \ROROcmin algorithms before stating the consistency and robustness results for \aug.  We discuss the results and their significance here, while deferring their full proofs to \autoref{apx:analysis}.

\noindent\textbf{Competitive analysis for \ROROcmax.} \label{subsec:alg1-comp-amalysis}
Recall that \ROROcmax assumes each job has the maximum length \cmax.  In the following theorem, we state the competitive result for \ROROcmax and discuss its significance.  The full proof of \autoref{th:alg1-alpha} is in \autoref{apx:comp-alg1}.

\begin{theorem}\label{th:alg1-alpha}
    \ROROcmax for \OCSUmin is $\alpha_1$-competitive when the threshold function is given by $\phi_1(w)$ from \sref{Definition}{def:alg1-thresh}.
    \begin{align}
        \alpha_1 = \frac{U}{\alpha L} + \frac{2 \beta}{L}. \label{eq:alg1-alpha}
    \end{align} 
\end{theorem}

Intuitively, compared to the competitive bound $\alpha$ shown for \OCS when job lengths are exactly known, $\alpha_1$ is worse.  This captures an edge case where the actual job length is, e.g., \cmin, while \ROROcmax's design assumes the job has length \cmax.  As we discuss in the full proof, this occurs because \ROROcmax's scaled threshold design allows it to complete a job with length \cmin by using the \textit{worst \cmin fraction} of the threshold, while \RORO (where the job length is known) uses the entire domain of the monotonically decreasing threshold function to complete the job.

\noindent\textbf{Competitive analysis for \ROROcmin.} \label{sec:comp-alg2}
Contrary to the assumption of \ROROcmax, the \ROROcmin algorithm is derived to prepare for a job with length \cmin, while acknowledging that the actual job length may be \cmax.  In the following theorem, we state the competitive result for \ROROcmin and discuss both its significance and relation to the existing \RORO algorithm with known job lengths.

\begin{theorem}\label{th:alg2-alpha}
    \ROROcmin for \OCSUmin is $\alpha'$-competitive when the threshold function is given by $\phi_2(w)$ from \sref{Definition}{def:alg2-thresh}.  We henceforth use $\alpha_2 = \alpha'$ to denote the competitive ratio of \ROROcmin.
\end{theorem}

As we show in the full proof, in \autoref{apx:comp-alg2}, $\alpha_2 = \alpha' \geq \alpha$, further implying that $ \phi_2(w) \leq \phi_1(w) $ for any $w \in [0, \cmax]$. 
This supports the notion that \ROROcmin is indeed more conservative than \ROROcmax in terms of the carbon intensities for which it is willing to run the job. We note that the functions  $\phi_1(w)$ and $ \phi_2(w) $ are equivalent when \cmin = \cmax, indicating that both algorithms make the same decisions when all jobs have the same length.

\begin{figure}[t]
    \centering
    \includegraphics[width=\linewidth]{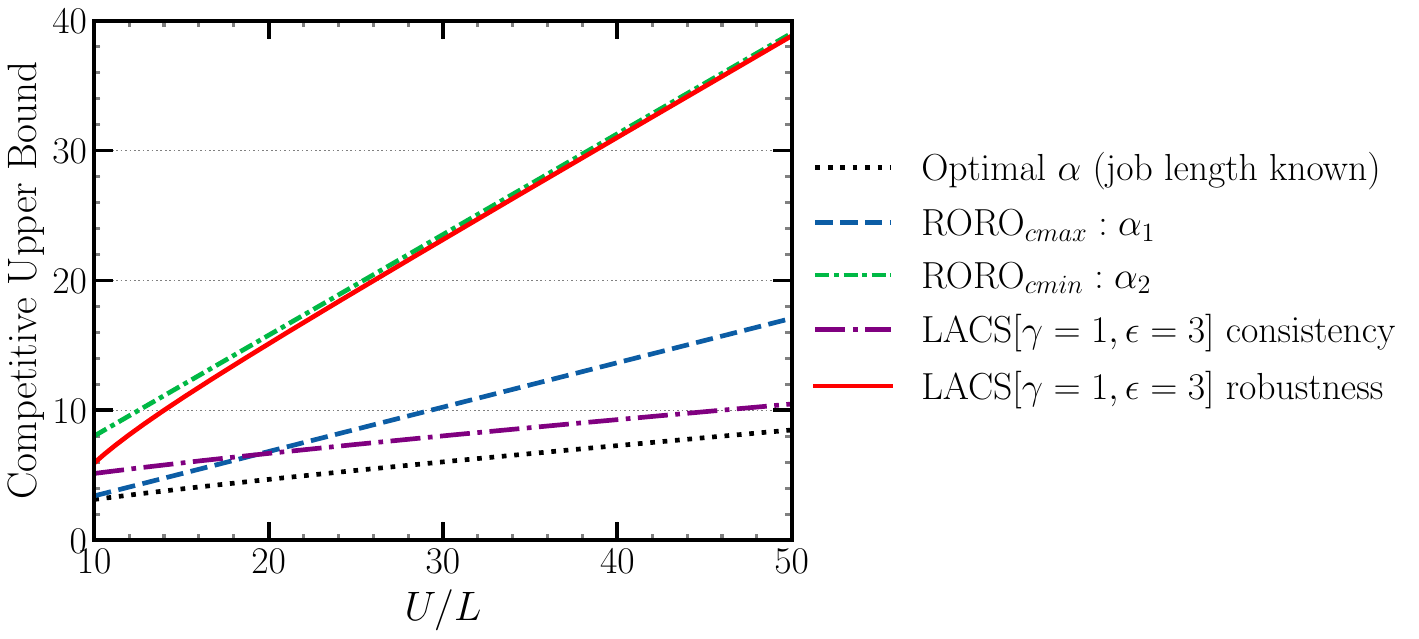}
    \vspace{-0.7cm}
    \caption{Competitive upper bounds for different algorithms (\RORO with the full knowledge of job length, \ROROcmin, \ROROcmax, and $\texttt{LACS}[\gamma = 1, \epsilon = 3]$)}
    \vspace{-0.5cm}
    \label{fig:alpha}
\end{figure}

\noindent\textbf{Consistency and robustness of \aug.}
Recall that \aug (summarized in \autoref{alg:aug}) is our learning-augmented algorithm that plays a convex combination of the solutions obtained by \combalg and \ROROpred.  Letting $\alpha^{\max}_{\ROROpred}$ denote the worst-case competitive ratio of \ROROpred (i.e., when the job length predictions are maximally incorrect), we obtain the following consistency and robustness bounds for \aug for any value of $\epsilon \in [0, \lvert \alpha_1 - \alpha_2 \rvert ]$ and any value of $\gamma \in [0, \alpha_1 - \texttt{sign}(\alpha_1 - \alpha_2)\epsilon - \alpha ]$:

\begin{theorem} \label{th:aug-con-rob}
    Given parameters $\epsilon$ and $\gamma$, \aug is $(\alpha + \gamma)$-consistent and $\bigg[ \left( 1-\frac{\gamma}{\alpha_1 - \texttt{sign}(\alpha_1 - \alpha_2)\epsilon - \alpha} \right) \alpha^{\max}_{\ROROpred} + \left( \frac{\gamma(\alpha_1 - \texttt{sign}(\alpha_1 - \alpha_2)\epsilon)}{\alpha_1 - \texttt{sign}(\alpha_1 - \alpha_2)\epsilon - \alpha} \right) \bigg]$-robust.
\end{theorem}

This result, proven fully in \autoref{apx:aug-con-rob}, implies that \aug can achieve a competitive ratio of $\alpha$ when the job length prediction is correct.  Since $\alpha$ is the best achievable competitive ratio for the original $\mathsf{OCS}$ problem, \aug with accurate predictions of the job length achieves the optimal consistency bound as $\gamma \to 0$.  Furthermore, the robustness bound implies that the worst-case competitive ratio when predictions are \textit{incorrect} remains bounded by a combination of $\alpha_1$ and $\alpha_2$; this is intuitive because the competitive bounds of \ROROcmax and \ROROcmin (respectively) are the worst-case results for algorithms which expect one extreme job length and must deal with the other extreme job length.

In \autoref{fig:alpha}, we plot the numerical values of $\alpha$, $\alpha_1$, $\alpha_2$, and the consistency-robustness results of \aug with $\gamma = 1, \epsilon = 3$ for several different values of $\nicefrac{U}{L}$.  $\beta$ is fixed to $\nicefrac{U}{10}$, and $\nicefrac{\cmax}{\cmin} = 4$.  Note that $\alpha$ (the best competitive ratio for standard $\mathsf{OCS}$) grows sublinearly in $\nicefrac{U}{L}$, while the competitive bounds of \ROROcmax ($\alpha_1$) and \ROROcmin ($\alpha_2$) grow linearly in $\nicefrac{U}{L}$.  This reflects the inherent challenges of \OCSUmin and the impact of uncertain job lengths on the achievable competitive ratios.  Notably, for this setting of $\epsilon$ and $\gamma$, \aug is able to nearly match the optimal $\alpha$ for $\mathsf{OCS}$ when the predictions are correct (consistency), and is strictly upper bounded by $\alpha_2$ when the predictions are adversarially incorrect (robustness), meaning that it achieves the \textit{best of both worlds}.

%% file: 5-exp.tex
In this section, we experimentally evaluate the performance of \aug in reducing the carbon footprint of scalable computing workloads.


\vspace{-0.1cm}
\subsection{Experimental Setup}\label{sec:exp_setup}

We take a job-centric approach where a carbon-aware scheduler independently allocates (i.e., scales) computing resources to each job to reduce the total carbon footprint of its execution while respecting per-job deadlines. We next detail our experimental setup.



\noindent \textbf{Carbon intensity trace.}
We use carbon intensity data from Electricity Maps~\cite{electricity-map} for California ISO (CAISO). 
The carbon trace provides carbon intensity measured in grams of CO2 equivalent per kilowatt-hour (gCO2eq/kWh) at an hourly granularity and spans 2020 to 2023.
We picked 1314 time slots as the job arrival times (once every 20 hrs\footnote{We purposefully avoid an arrival every 24 hrs to avoid diurnal patterns.}) to assess the performance across the whole trace duration. 
To evaluate algorithms that require carbon intensity ($\text{CI}$) forecasts, we introduce a uniformly random error to carbon intensity data to account for forecast errors, denoted as $\text{CI}_{\text{err}}$, where $\text{err}$ is the mean of added percentage error.

\begin{table}[t]
\caption{Inverse of scaling profiles as a mapping between the completed part of the job ($x$) and amount of resources allocated ($s$).}
\vspace{-0.3cm}
\label{tab:ScalingProfiles}
\footnotesize
\begin{tabular}{|c|c|c|c|}
\hline
\textbf{Profile} & \textbf{Equation} & \textbf{Profile} & \textbf{Equation} \\ \hline
\texttt{P1} & $s=x$ & \texttt{P4} & $s = 0.5x^2+x$ \\ \hline
\texttt{P2} & $s=0.15x^2+x$ & \texttt{P5} & $s=0.75x^2+x$ \\ \hline
\texttt{P3} & $s=0.25x^2+x$ & \texttt{P6} & $s=x^2+x$ \\ \hline
\end{tabular}%
\vspace{-0.25cm}
\end{table}

\begin{table}[t]
    \centering
    \captionof{table}{Summary of characteristics for the baseline algorithms and two variants of the proposed algorithm}
    \vspace{-0.3cm}
    \label{table:baseline-algs}
    \resizebox{0.455\textwidth}{!}{%
        \begin{tabular}{|c|c|c|c|c|}
            \hline
            \textbf{Algorithm} & \textbf{Carbon-Aware} & \textbf{Switching-Aware} & \textbf{Job Length Input} \\ \hline
            \RORO \cite{Lechowicz:24} & Yes & Yes & Actual \\ \hline
            \OWT & Yes & No & Prediction  \\ \hline
            \threshold & Yes & No & N/A  \\ \hline
            \agnostic & No & No & N/A \\ \hline
            \carbonscaler\cite{Hanafy:23:CarbonScaler} & Yes & No & Prediction  \\ \hline
            \aug (this paper) & Yes & Yes & Prediction  \\ \hline
            \augd \footnotemark (this paper) & Yes & Yes & Prediction  \\ \hline
        \end{tabular}
    }
    \vspace{-0.5cm}
\end{table}
\footnotetext{\augd is the modified version of \aug, which considers discrete resource allocation based on the solution given by \aug.}

\noindent \textbf{Job characteristics.}
Each job arrives independently with a job length $c$ uniformly sampled within the range [$\cmin$, $\cmax$]. 
To evaluate the impact of job length prediction error, we model a predictor that yields a job length estimate within the range [$c - p\times c$, $c + p\times c$], where $p$ is the percentage error in job length predictions. We also assume that all jobs have a deadline ($T$) of 24 hrs and incur a fixed symmetric maximum switching overhead of $\beta.h(r)$ (e.g., for checkpoint and resume) when scaling from zero to the maximum resources $r$, where $h(\cdot)$ is the scaling profile.

\noindent \textbf{Resource scaling profiles.}
Carbon savings highly depend on the scalability of jobs~\cite{Hanafy:23:CarbonScaler}, where more scalable jobs can yield higher savings as they provide higher marginal throughput for each added resource. \autoref{tab:ScalingProfiles} depicts the mapping functions $s = h_t^{-1}(x)$ (the inverse of the scaling profiles), where $x$ represents the completed part of the job (i.e., progress made) and $s$ represents the amount of resource, e.g., servers, to obtain the given progress. 
Some of the utilized profiles represent common scalability profiles of real-world batch jobs. For instance, \texttt{P1} refers to embarrassingly parallel applications such as BLAST~\cite{Souza:23}, while \texttt{P2} and \texttt{P4} are a fitted version of the machine learning training workloads for ResNet18 and MobileNetV2, respectively, described in~\cite{Hanafy:23:CarbonScaler}. In contrast, \texttt{P3}, \texttt{P4}, and \texttt{P6} are synthesized profiles to represent moderate and non-scalable applications.


\noindent \textbf{Parameter settings.}
We evaluate \aug across a wide range of experimental scenarios that impact its performance, including a range of maximum job lengths (\cmax), varying coefficients for switching emissions ($\beta$), errors in job length predictions, and a range of learning augmentation factors ($\lambda$). 
We set the value of the decision factor ($k$) in \autoref{alg:aug} to 0.5, so both robust algorithms receive equal consideration. To impose practical constraints, we evaluate \aug across a range of carbon intensity forecast errors ($\text{CI}_{\text{err}}$) and resource constraints ($r$) since available resources are constrained.

\begin{figure*}[t]
    \centering
    \includegraphics[width=0.8\textwidth]{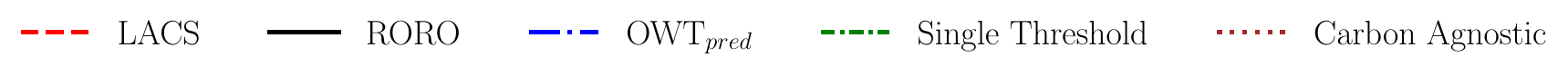}\\
    \begin{subfigure}[b]{0.245\textwidth}
        \centering
        \includegraphics[width=\textwidth]{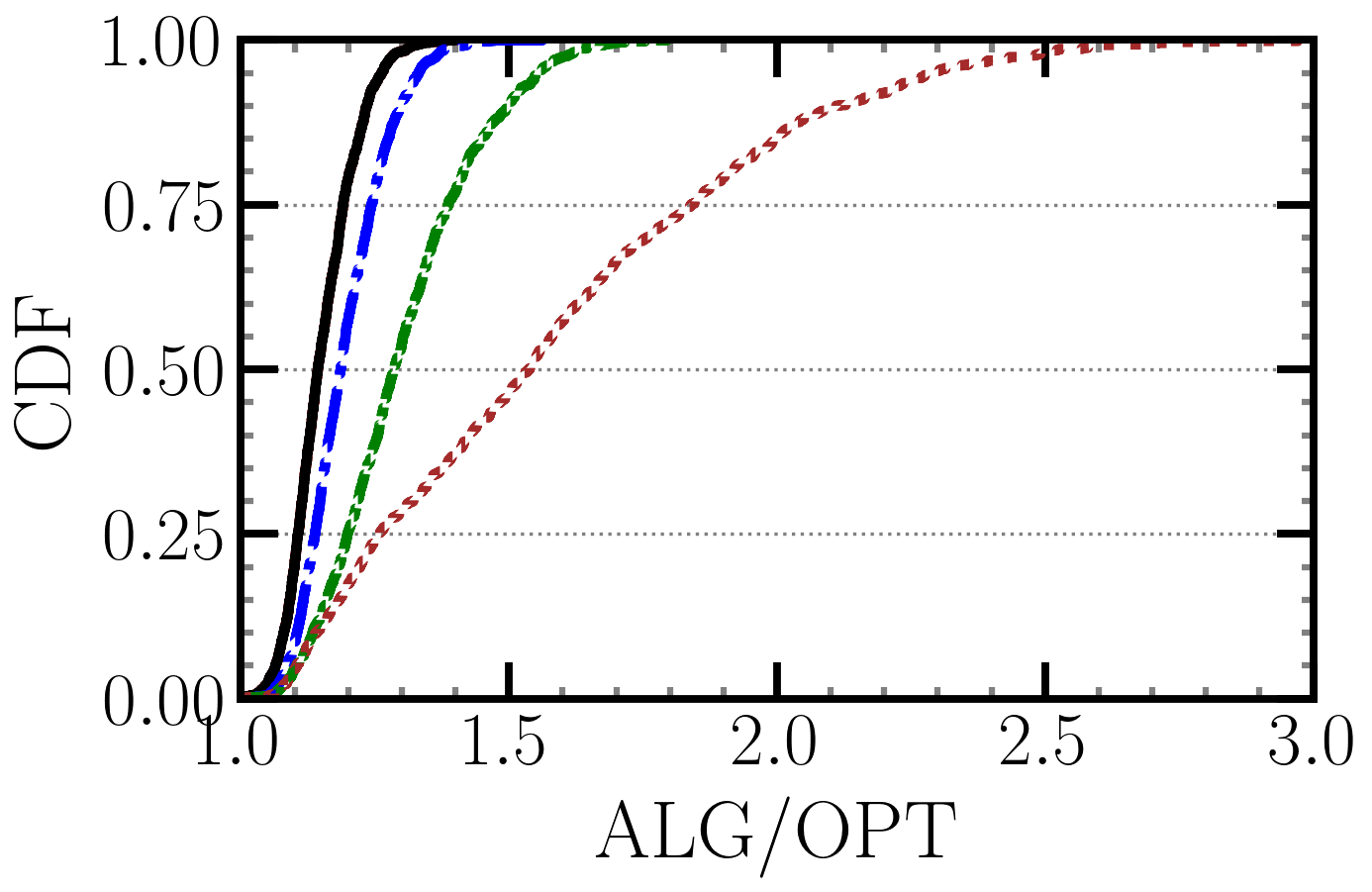}
        \vspace{-0.6cm}
        \caption{\cmax= 1}
        \label{fig:effect_cMax_1}
    \end{subfigure}
    \hfill
    \begin{subfigure}[b]{0.245\textwidth}
        \centering
        \includegraphics[width=\textwidth]{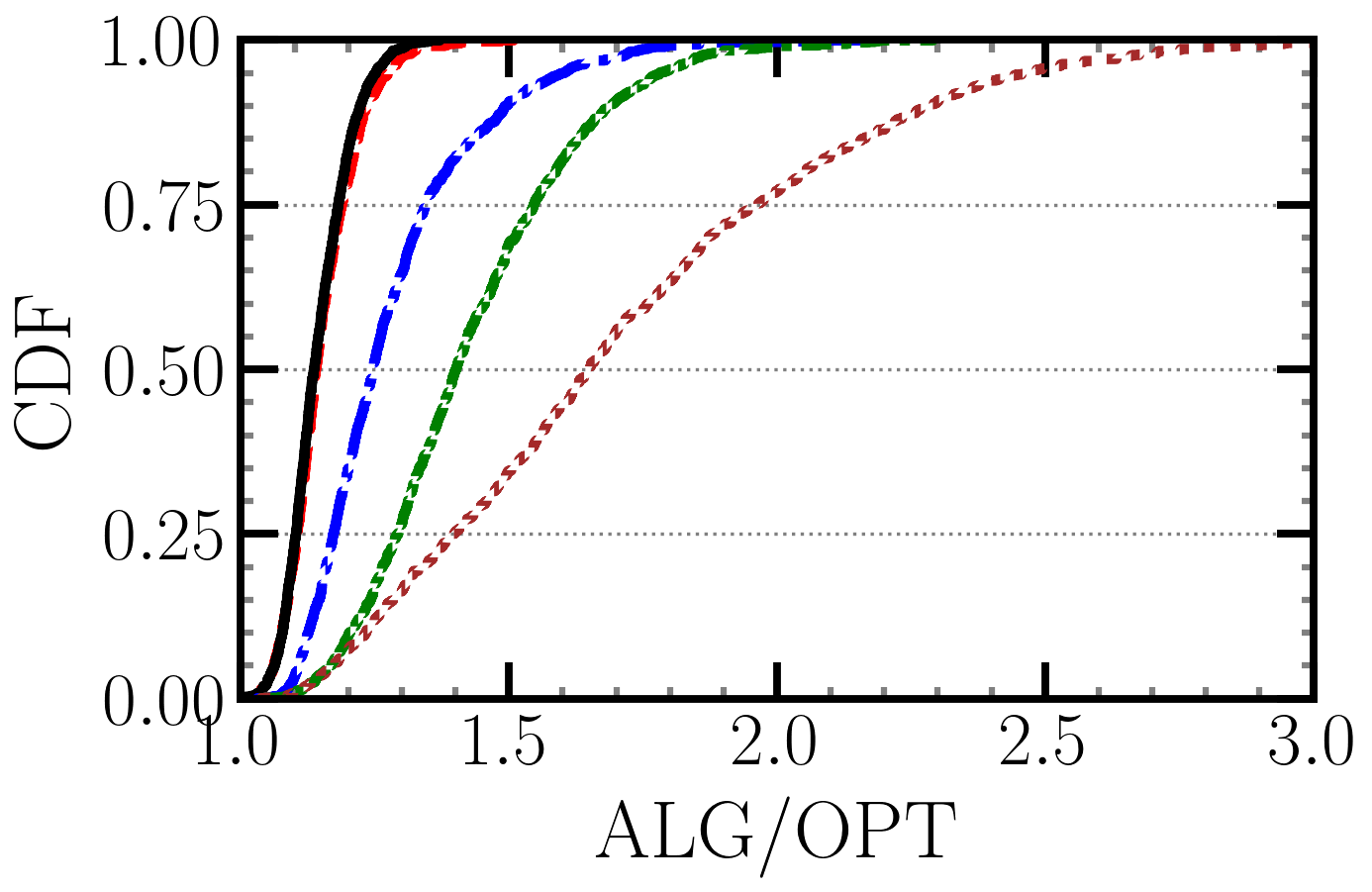}
        \vspace{-0.6cm}
        \caption{\cmax= 3}
        \label{fig:effect_cMax_3}
    \end{subfigure}
    \hfill
    \begin{subfigure}[b]{0.245\textwidth}
        \centering
        \includegraphics[width=\textwidth]{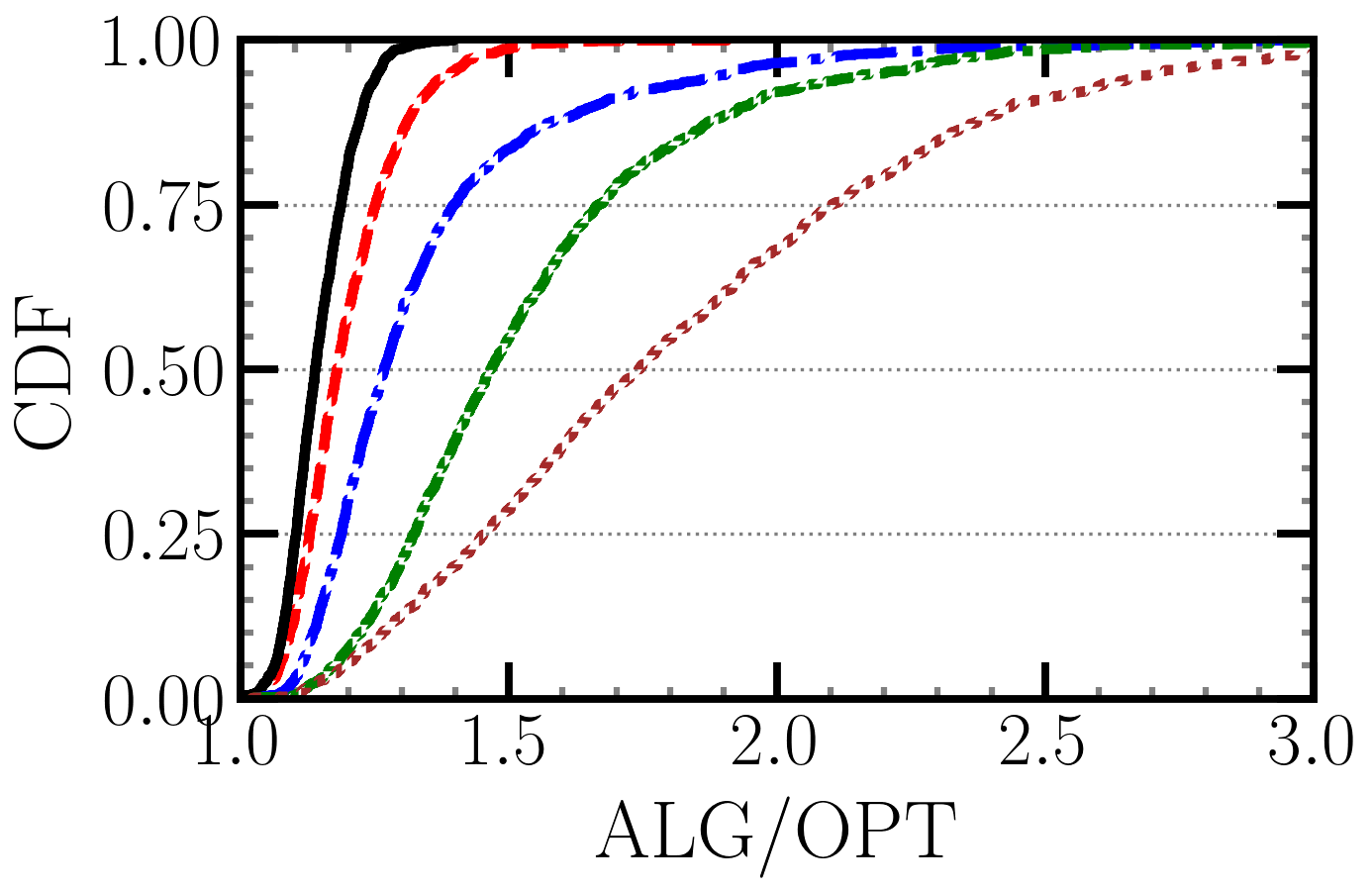}
        \vspace{-0.6cm}
        \caption{\cmax = 6}
        \label{fig:effect_cMax_6}
    \end{subfigure}
    \hfill
    \begin{subfigure}[b]{0.245\textwidth}
        \centering
        \includegraphics[width=\textwidth]{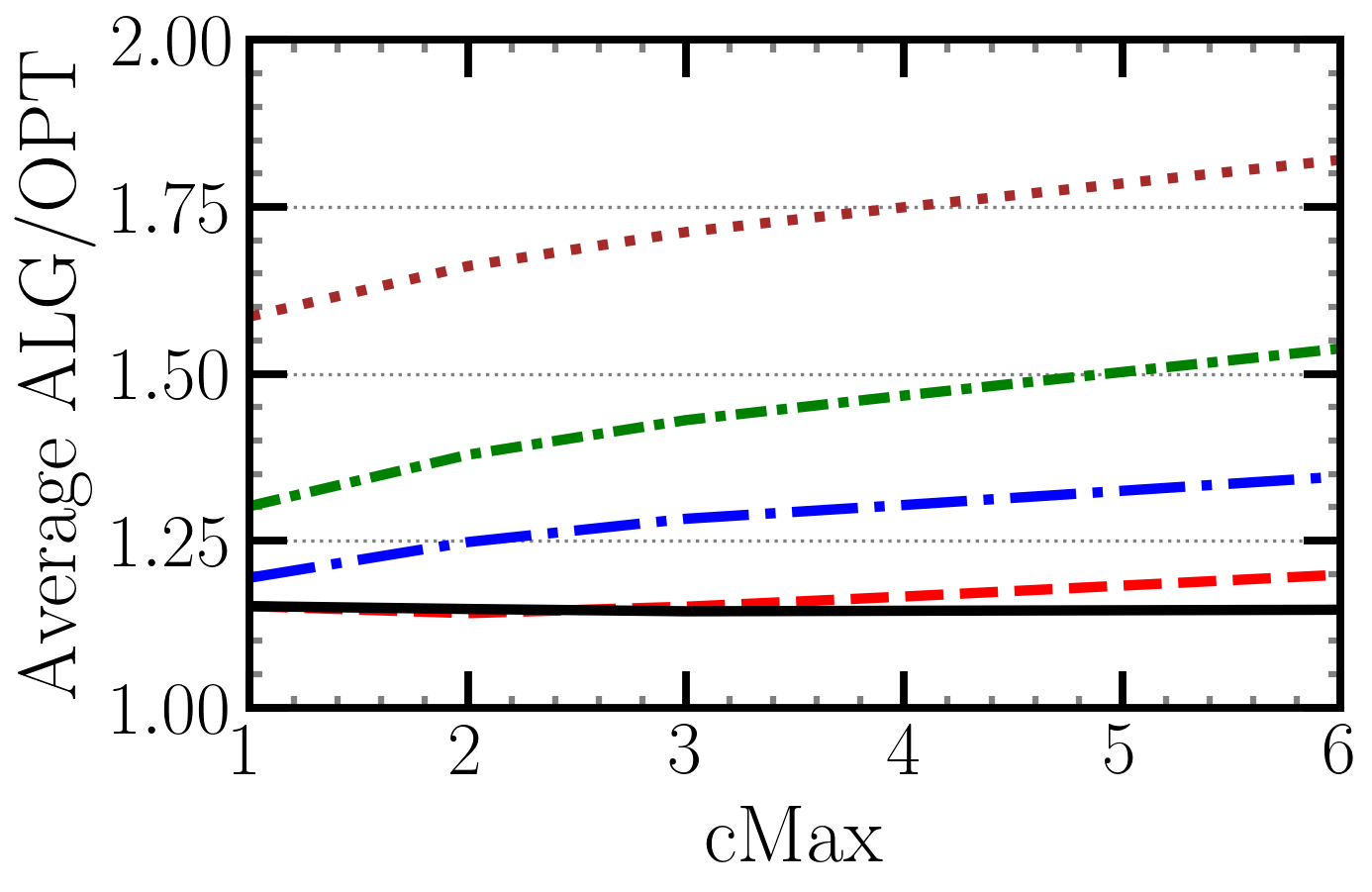}
        \vspace{-0.6cm}
        \caption{Effect of \cmax}
        \label{fig:effect_cMax_all}
    \end{subfigure}
    \vspace{-0.4cm}
    \caption{(a), (b), and (c) report cumulative distribution functions (CDFs) of empirical competitive ratios for evaluated algorithms under different \cmax values, where \cmin=1. (d) Shows the effect of \cmax on empirical competitive ratios. The scaling profile is \texttt{P1}, $\beta=20$, job prediction error is $20\%$, and $\lambda=0.5$. A CDF curve towards the top left corner indicates better performance.}
    \label{fig:effect_cMax}
    \vspace{-0.4cm}
\end{figure*}

\begin{figure*}[t]
    \centering
    \includegraphics[width=0.8\textwidth]{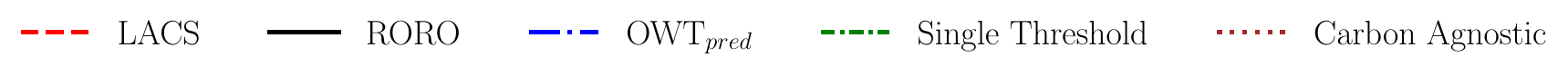}\\
    \begin{subfigure}[b]{0.24\textwidth}
        \centering
        \includegraphics[width=\textwidth]{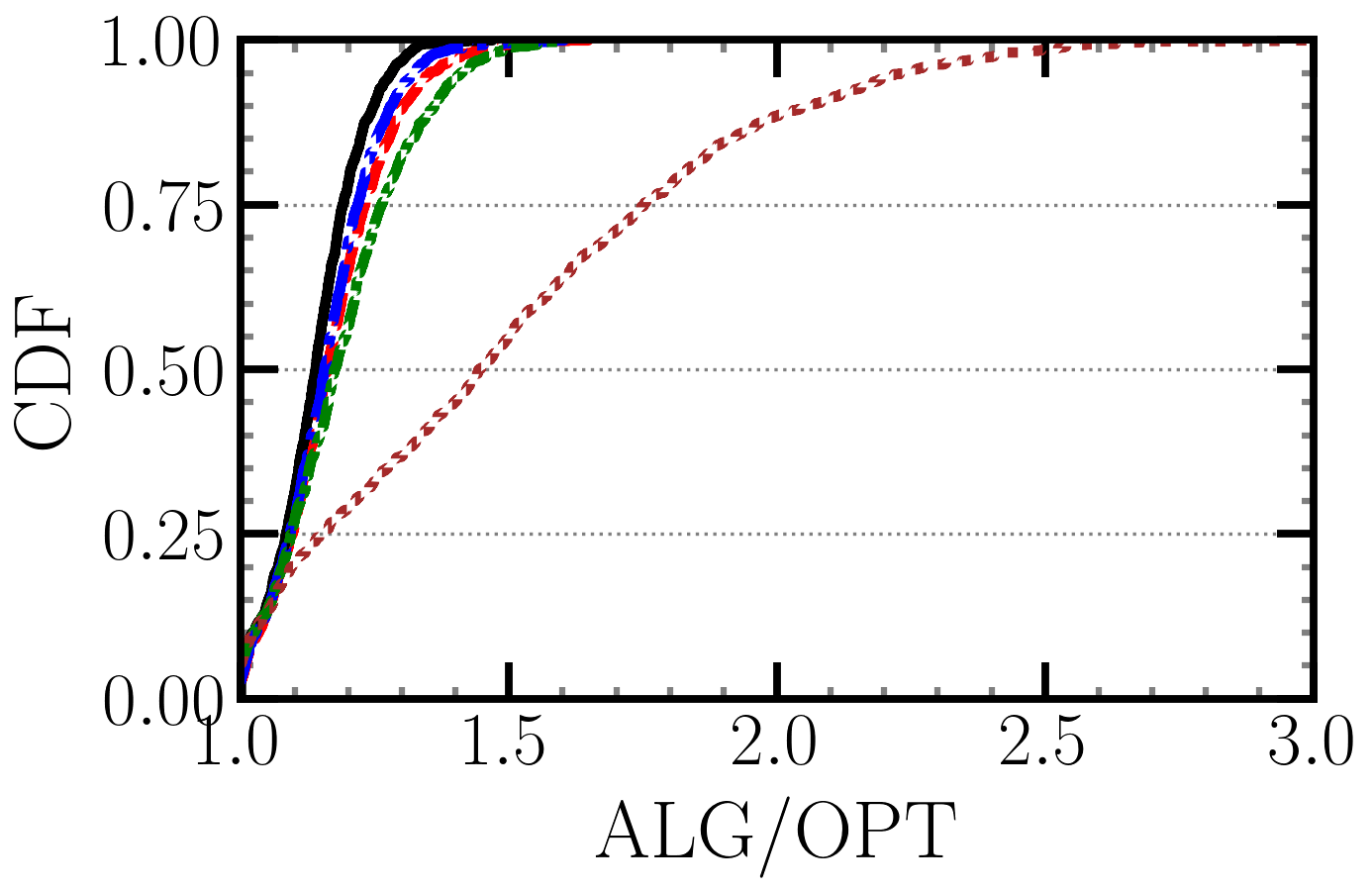}
        \vspace{-0.6cm}
        \caption{$\beta = 0$}
        \label{fig:effect_beta_0}
    \end{subfigure}
    \hfill
    \begin{subfigure}[b]{0.24\textwidth}
        \centering
        \includegraphics[width=\textwidth]{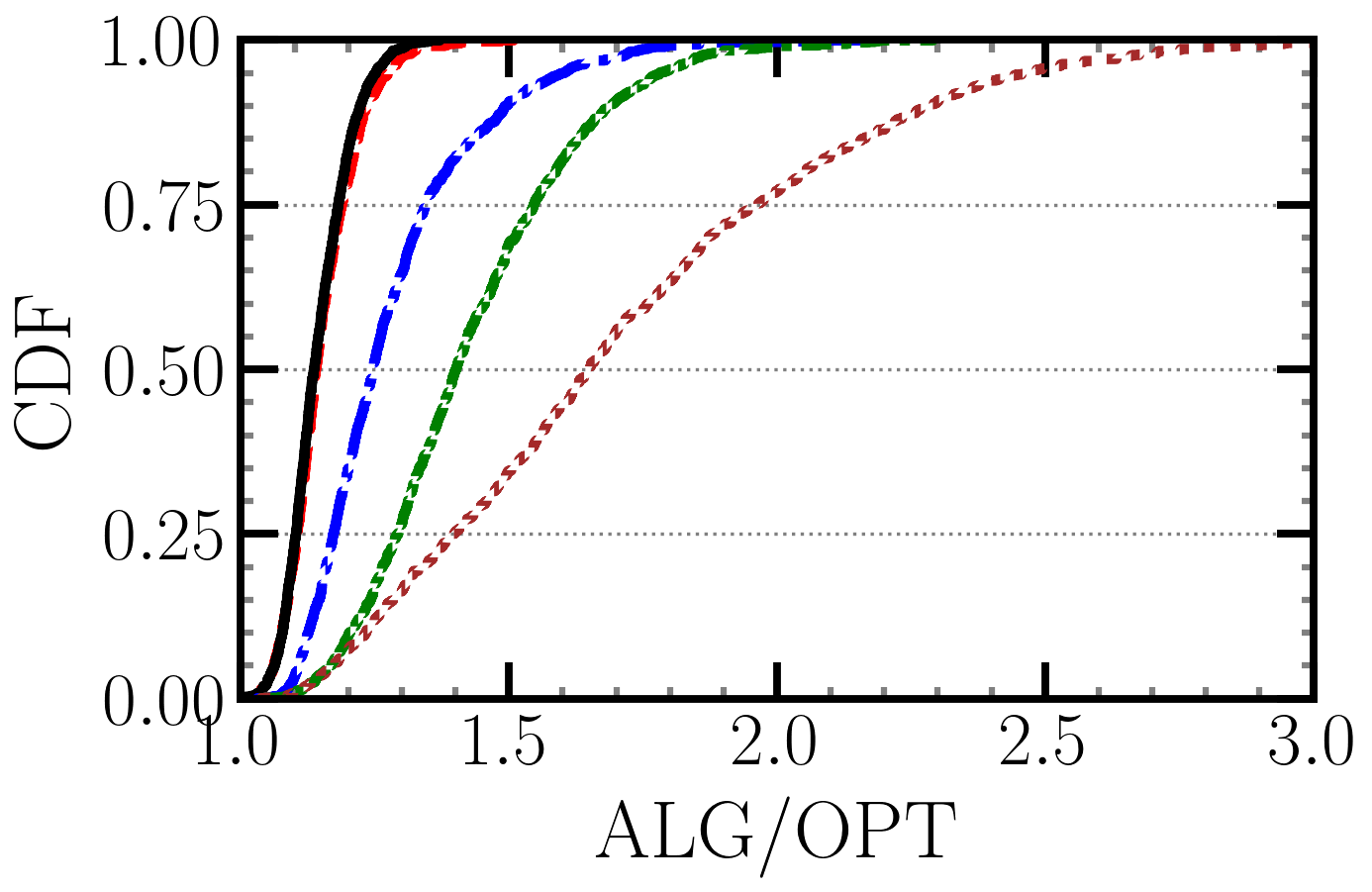}
        \vspace{-0.6cm}
        \caption{$\beta = 20$}
        \label{fig:effect_beta_20}
    \end{subfigure}
    \hfill
    \begin{subfigure}[b]{0.24\textwidth}
        \centering
        \includegraphics[width=\textwidth]{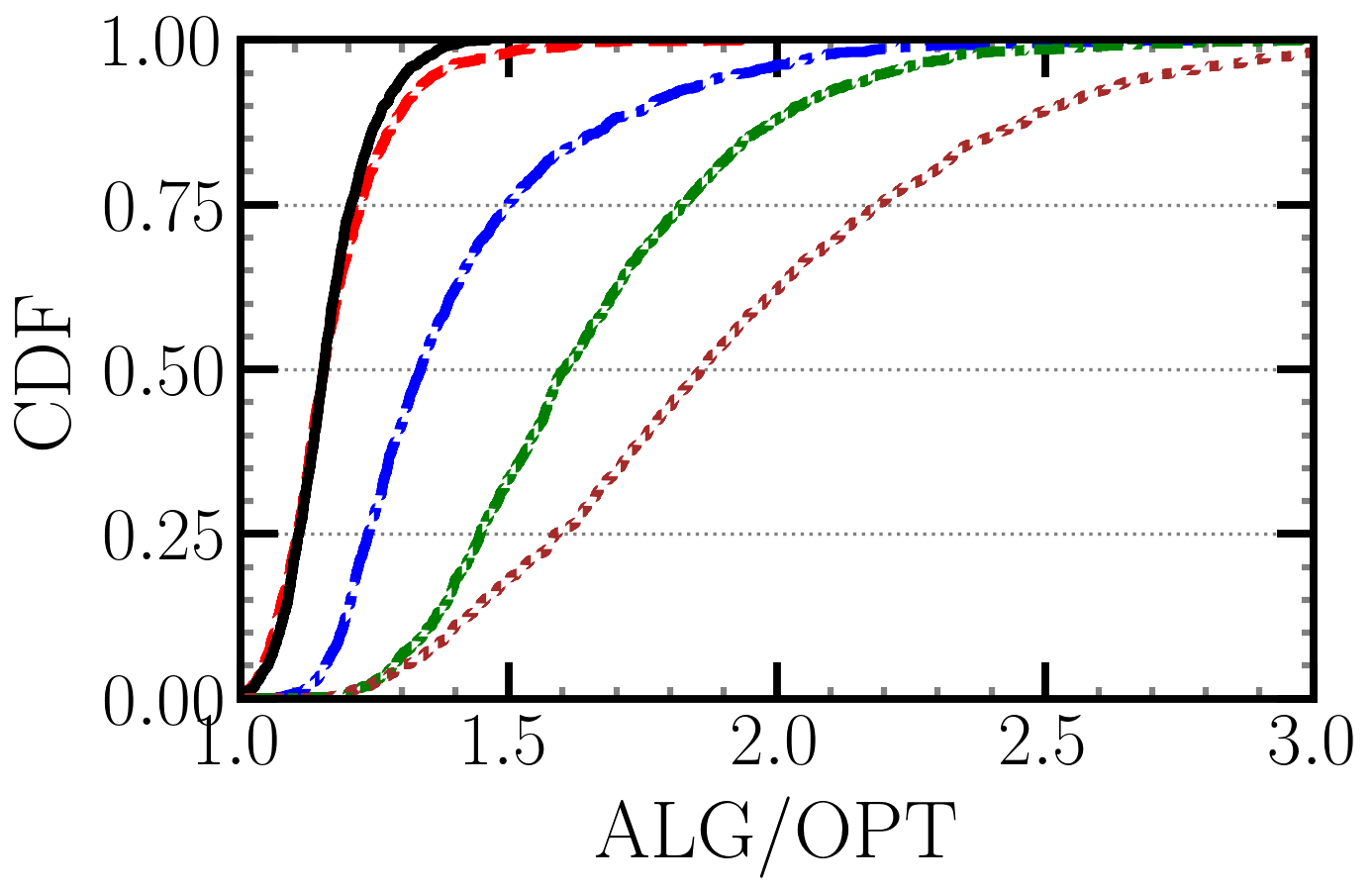}
        \vspace{-0.6cm}
        \caption{$\beta = 40$}
        \label{fig:effect_beta_40}
    \end{subfigure}
    \hfill
    \begin{subfigure}[b]{0.24\textwidth}
        \centering
        \includegraphics[width=\textwidth]{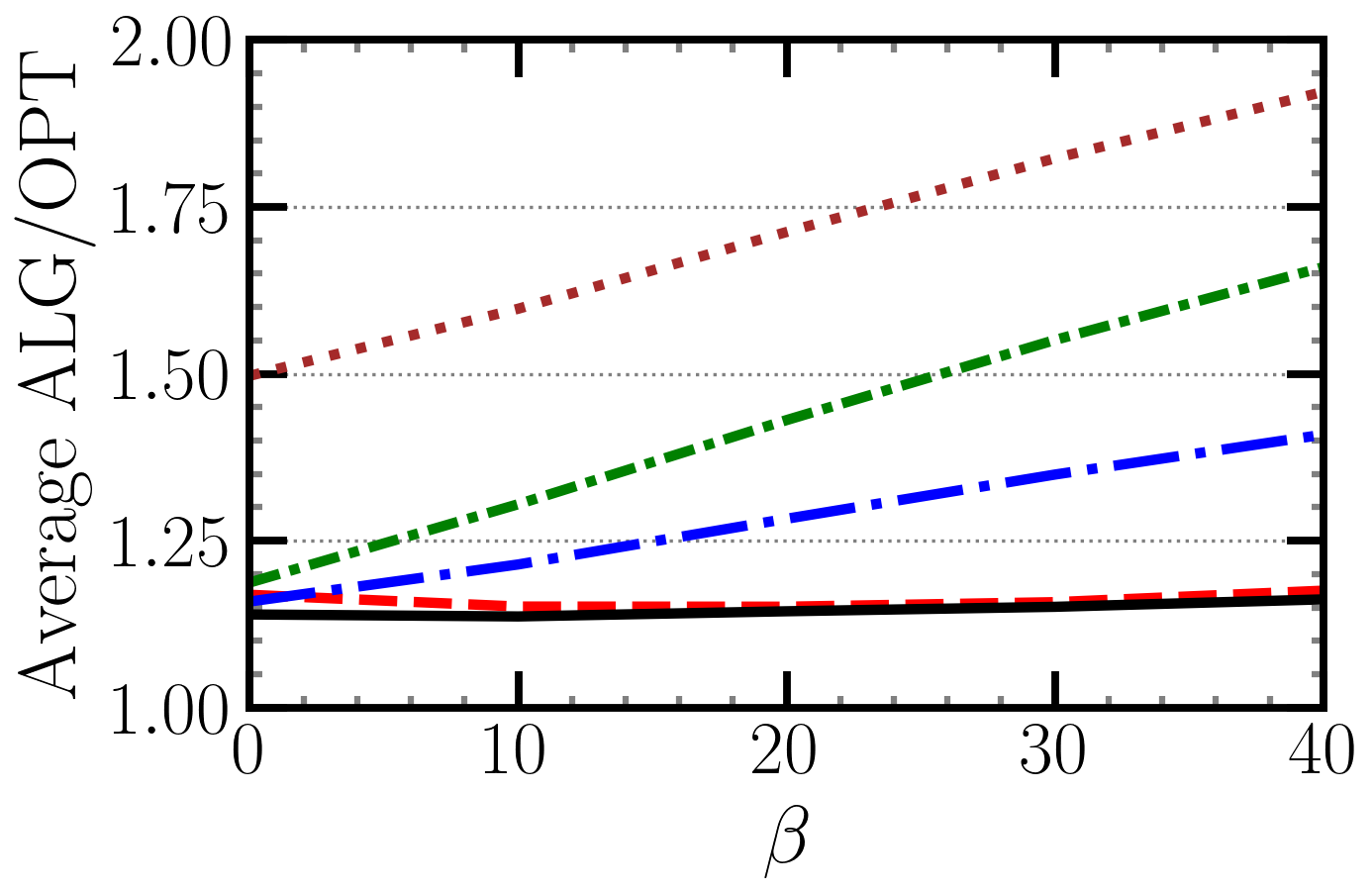}
        \vspace{-0.6cm}
        \caption{Effect of $\beta$}
        \label{fig:effect_beta_all}
    \end{subfigure}
    \vspace{-0.4cm}
    \caption{(a), (b), and (c) report cumulative distribution functions (CDFs) of empirical competitive ratios for selected algorithms under different $\beta$ values. (d) Shows the effect of $\beta$ on empirical competitive ratios. The scaling profile is \texttt{P1}, \cmax$=3$, job prediction error is $20$\%, and $\lambda$=$0.5$. A CDF curve towards the top left corner indicates better performance.}
    \label{fig:effect_beta}
\end{figure*}

\noindent \textbf{Baseline algorithms.}
We evaluate the two variants of our proposed algorithm, \aug and \augd, against multiple state-of-the-art algorithms, summarized in \autoref{table:baseline-algs} and detailed below.
\begin{enumerate}[leftmargin=*]
    \item \textbf{\RORO:} An online conversion with switching costs algorithm with knowledge of job length~\cite{Lechowicz:24}, described in \autoref{subsec:alg_background}. This \textit{online algorithm} knows the accurate job length, so it serves as an upper bound for \aug that has no access to the exact job length a priori. 
    
    \item \textbf{\OWT:} Threshold-based one-way trading with job length predictions. This algorithm adapts a threshold-based one-way trading algorithm~\cite{SunZeynali:20} that assumes perfect job length information.
    Our adaptation, \OWT, uses inaccurate job length predictions to decide the amount of job $x_t$ to be scheduled at time $t$ based on a threshold function $\Phi$. It accounts for the carbon emissions of execution but ignores any switching emissions. When $\beta=0$, \ROROpred (described in \autoref{sec:augmentation}) reduces to \OWT.

    \item \textbf{\threshold:} The algorithm utilizes a static threshold set at $\sqrt{UL}$, a value initially introduced in \cite{ElYaniv:01}. Adapted for \OCSUmin, \threshold operates by running the job with the maximum resource available at each time step $t$, but only if the execution emissions are lower than $\sqrt{UL}$. This algorithm also does not consider switching emissions. 

    \item \textbf{\agnostic: } This algorithm presents a greedy approach that executes each job with maximum available resources upon submission. As the scheduler does not know the job length, executing the job with maximum resources reduces execution time and ensures completion before the deadline, if feasible. 

    \item \textbf{\carbonscaler:} The \carbonscaler~\cite{Hanafy:23:CarbonScaler} algorithm utilizes its knowledge of the job length, carbon intensity, and deadline to construct a carbon-aware resource scaling. To accommodate job length prediction inaccuracies and potential errors in carbon intensity forecasts, we feed the predicted job length and the erroneous carbon intensity forecasts to the algorithm.
\end{enumerate}

It must be noted that all algorithms, except \RORO, start the compulsory execution at the time slot $T - (\cmax - w^{(t)})/d_t$, where $T$ is the deadline, and $w^{(t)}$ is the progress at time $t$ to ensure that the job completes  before the deadline\footnote{\RORO starts the compulsory execution based the actual job length $c$}.

\noindent \textbf{Evaluation metric.}
We compare the performance of an algorithm against an offline optimal resource scaling schedule computed using a numerical solver~\cite{SciPy}, which allows us to report the empirical competitive ratio represented as $\nicefrac{\texttt{ALG}}{\texttt{OPT}}$ (lower value is better). We also report the reduction in carbon footprint with respect to the \agnostic execution of the job that uses maximum resources and aims to finish the job as soon as possible.


\vspace{-0.3cm}
\subsection{Effect of Maximum Job Length}
The ratio between the maximum and minimum job lengths (\cmax, \cmin) dictates the characteristics of jobs that may arrive at the scheduler; a higher ratio means the job lengths can be more diverse. 
For simplicity, we set \cmin$=1$ and analyze the effect of \cmax. 
\autoref{fig:effect_cMax} shows the performance of various algorithms under different \cmax values, where lower $\nicefrac{\texttt{ALG}}{\texttt{OPT}}$ is better. We evaluate the proposed algorithm against \RORO, \OWT, \threshold, and \agnostic by showing their performance compared to the offline optimal algorithm. We omit \carbonscaler in this section as it requires discrete resource allocation; we will consider it when we enforce discrete assignment.
We assume scaling profile \texttt{P1}, i.e., the job is embarrassingly parallel, switching emission coefficient $\beta=20$, job prediction error of $20\%$, and augmentation factor $\lambda=0.5$.

\autoref{fig:effect_cMax_1} shows the special case where $\cmax=\cmin$, i.e., all jobs are the same length and known to \aug,  \RORO, \OWT. As expected, \aug and \RORO are identical while \OWT is slightly behind as it does not consider switching emissions. Nonetheless, these algorithms achieve higher performance than \threshold and \agnostic. 
\autoref{fig:effect_cMax_3} and \autoref{fig:effect_cMax_6} show realistic cases where the \cmax is 3 and 6, respectively. As shown, \aug is the closest to \RORO in all settings. Its performance is within 15.1 and 19.8\% on average from the offline optimal and only 0.6\% and 4.6\% from \RORO, when \cmax is 3 and 6, respectively.
\autoref{fig:effect_cMax_all} summarizes the average competitive ratio across different \cmax values; increasing the \cmax decreases the performance of all techniques, as a higher \cmax yields more uncertainty on the actual job length and exacerbates switching emissions as we do not employ resource constraints. Nonetheless, the results indicate the superiority of \aug across different \cmax values, where 
\aug performs within 16.1\% of the offline optimal and only 1.3\% away from \RORO, resulting in a 31\% reduction in carbon emissions compared to the \agnostic policy. 
\vspace{-1em}

\begin{figure*}[t]
    \centering
    \begin{subfigure}[b]{0.247\textwidth}
        \centering
        \includegraphics[width=\textwidth]{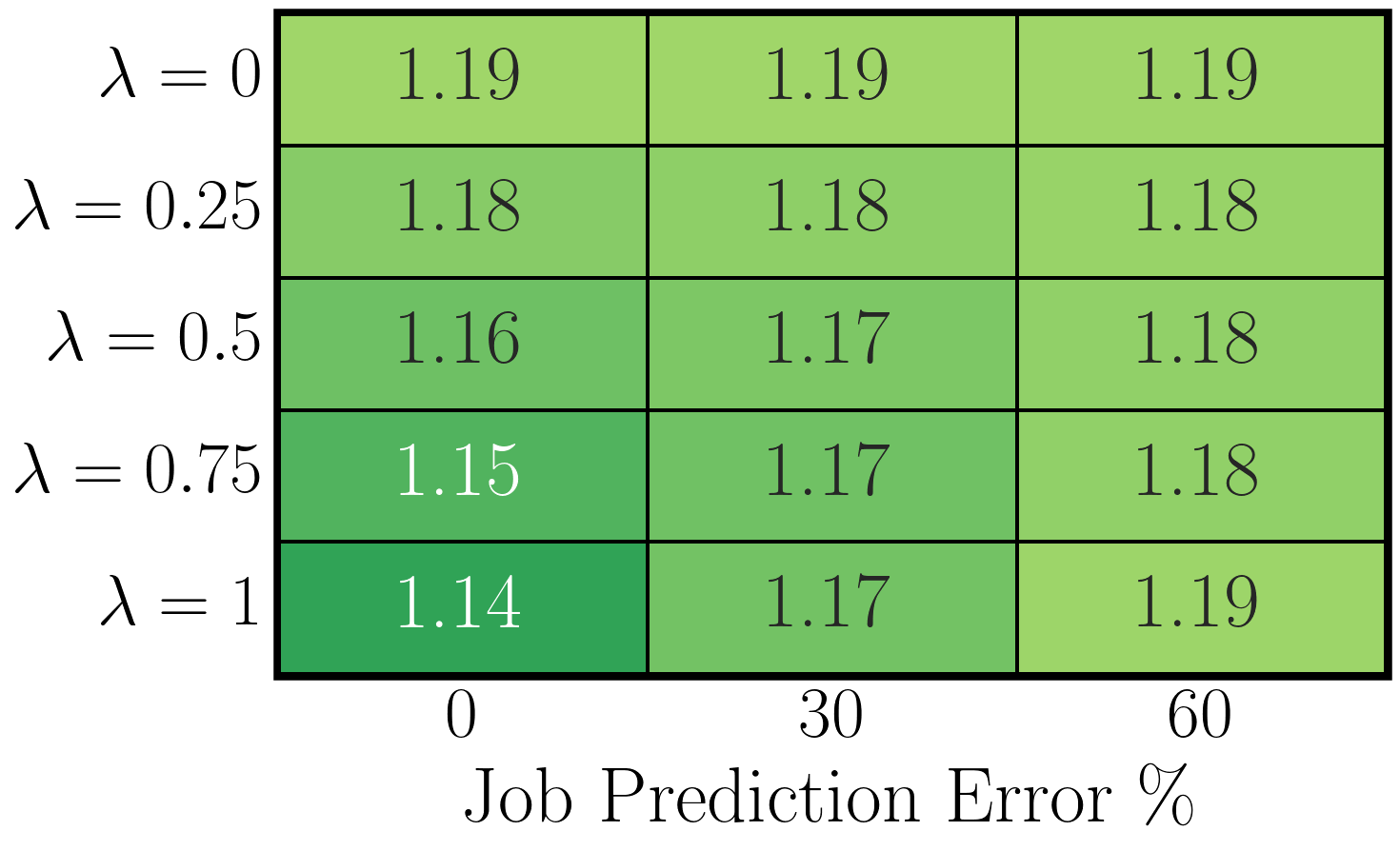}
        \vspace{-0.5cm}
        \caption{$\beta = 0, \cmax=3$}
        \label{fig:lambda_error_3_0}
    \end{subfigure}
    \hfill
    \begin{subfigure}[b]{0.2\textwidth}
        \centering
        \includegraphics[width=\textwidth]{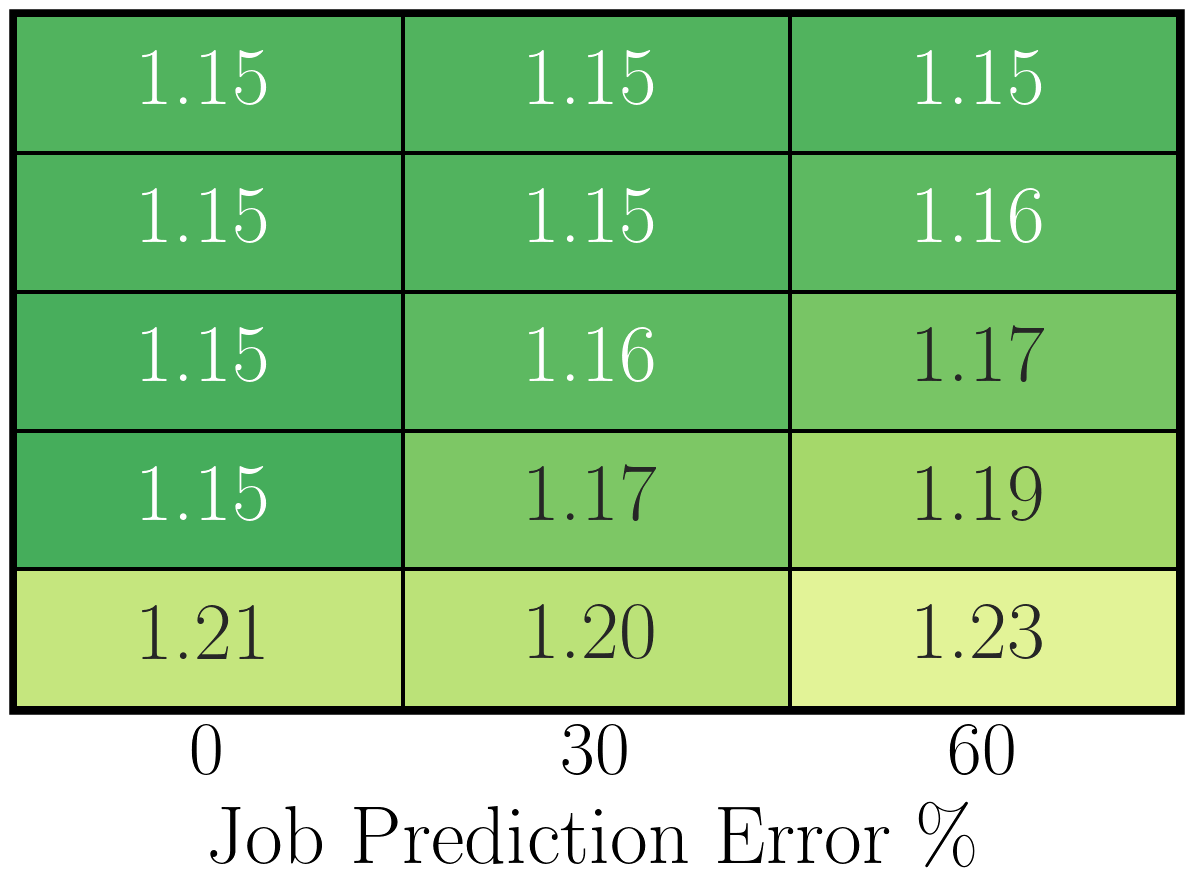}
        \vspace{-0.5cm}
        \caption{$\beta = 10, \cmax=3$}
        \label{fig:lambda_error_3_10}
    \end{subfigure}
    \hfill
    \begin{subfigure}[b]{0.2\textwidth}
        \centering
        \includegraphics[width=\textwidth]{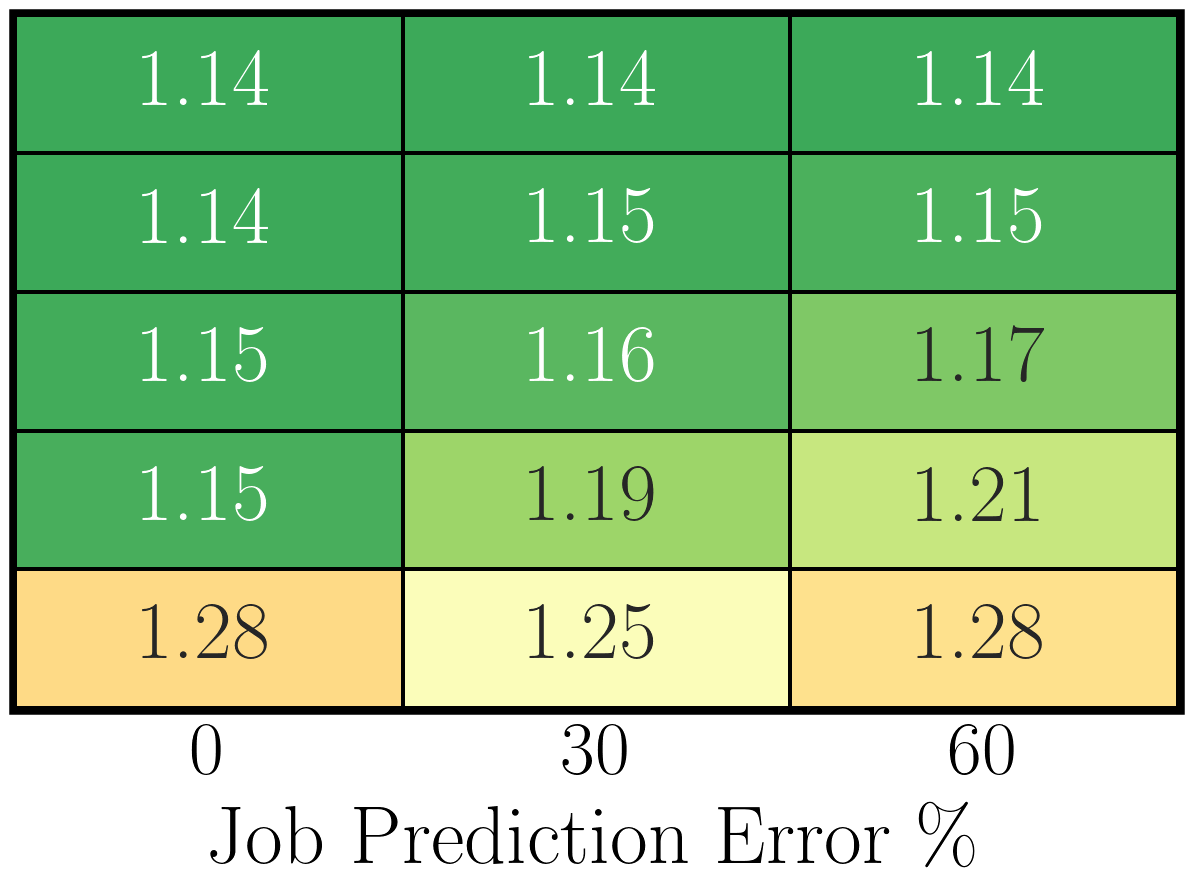}
        \vspace{-0.5cm}
        \caption{$\beta = 20, \cmax=3$}
        \label{fig:lambda_error_3_20}
    \end{subfigure}
    \hfill
    \begin{subfigure}[b]{0.232\textwidth}
        \centering
        \includegraphics[width=\textwidth]{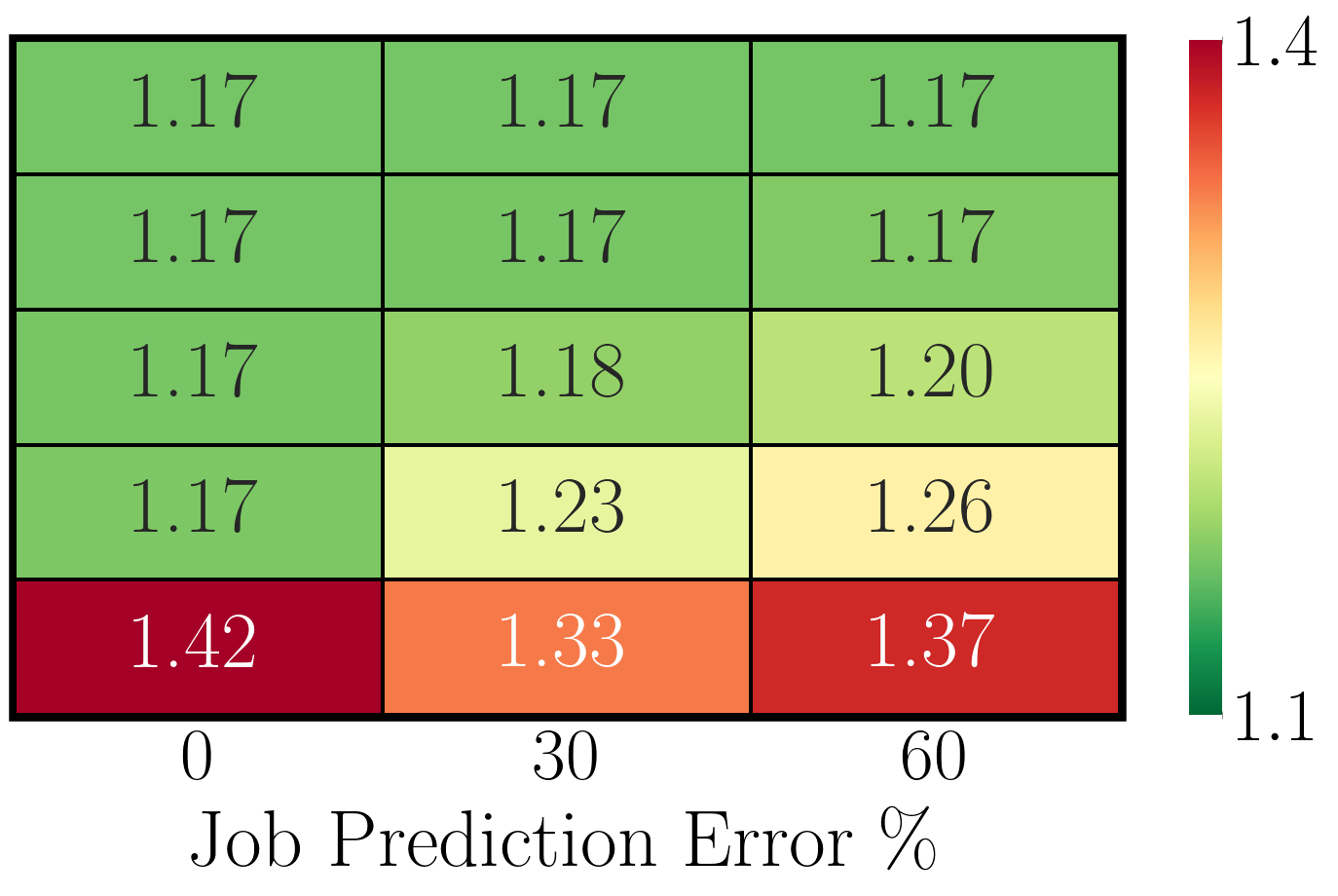}
        \vspace{-0.5cm}
        \caption{$\beta = 40, \cmax=3$}
        \label{fig:lambda_error_3_40}
    \end{subfigure}
    \\
    \begin{subfigure}[b]{0.247\textwidth}
        \centering
        \includegraphics[width=\textwidth]{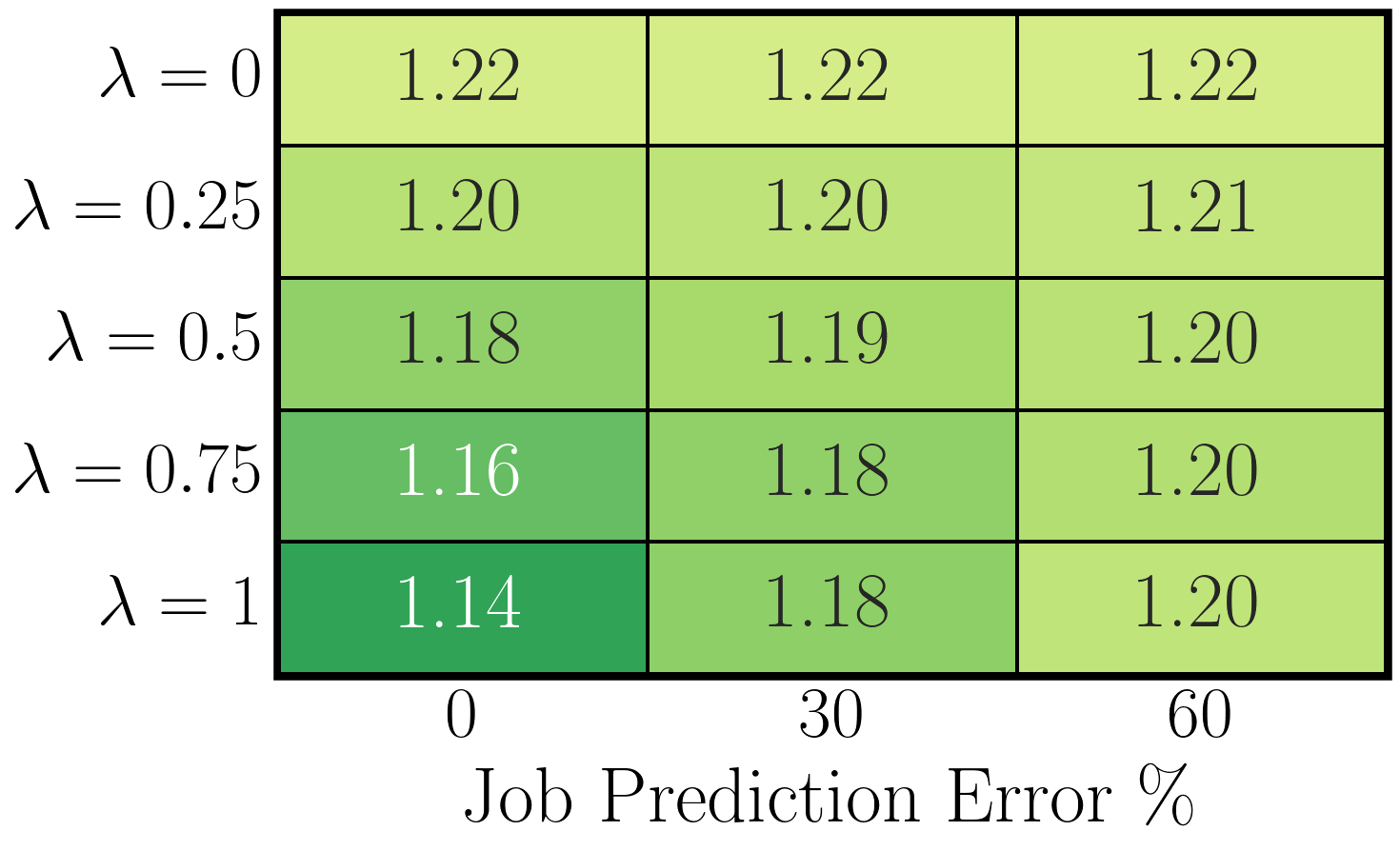}
        \vspace{-0.5cm}
        \caption{$\beta = 0, \cmax=6$}
        \label{fig:lambda_error_6_0}
    \end{subfigure}
    \hfill
    \begin{subfigure}[b]{0.2\textwidth}
        \centering
        \includegraphics[width=\textwidth]{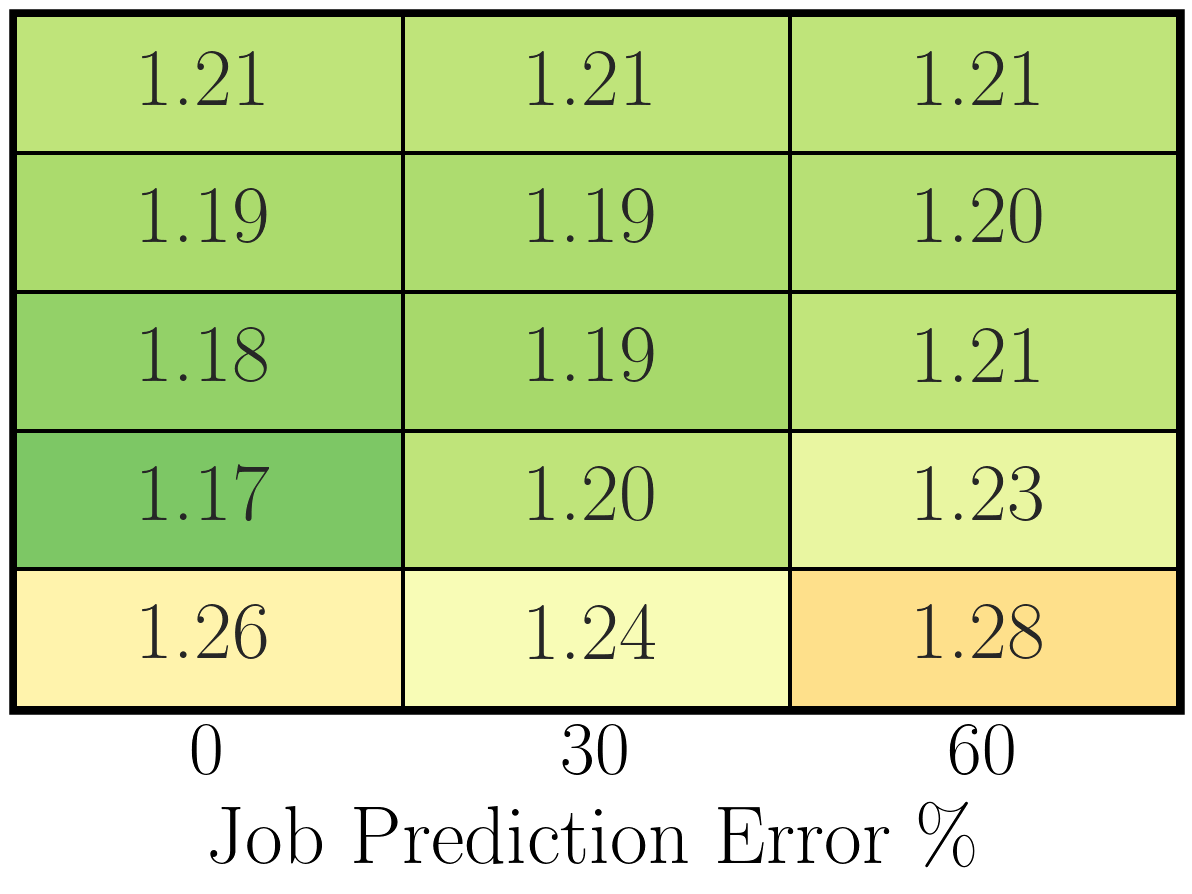}
        \vspace{-0.5cm}
        \caption{$\beta = 10, \cmax=6$}
        \label{fig:lambda_error_6_10}
    \end{subfigure}
    \hfill
    \begin{subfigure}[b]{0.2\textwidth}
        \centering
        \includegraphics[width=\textwidth]{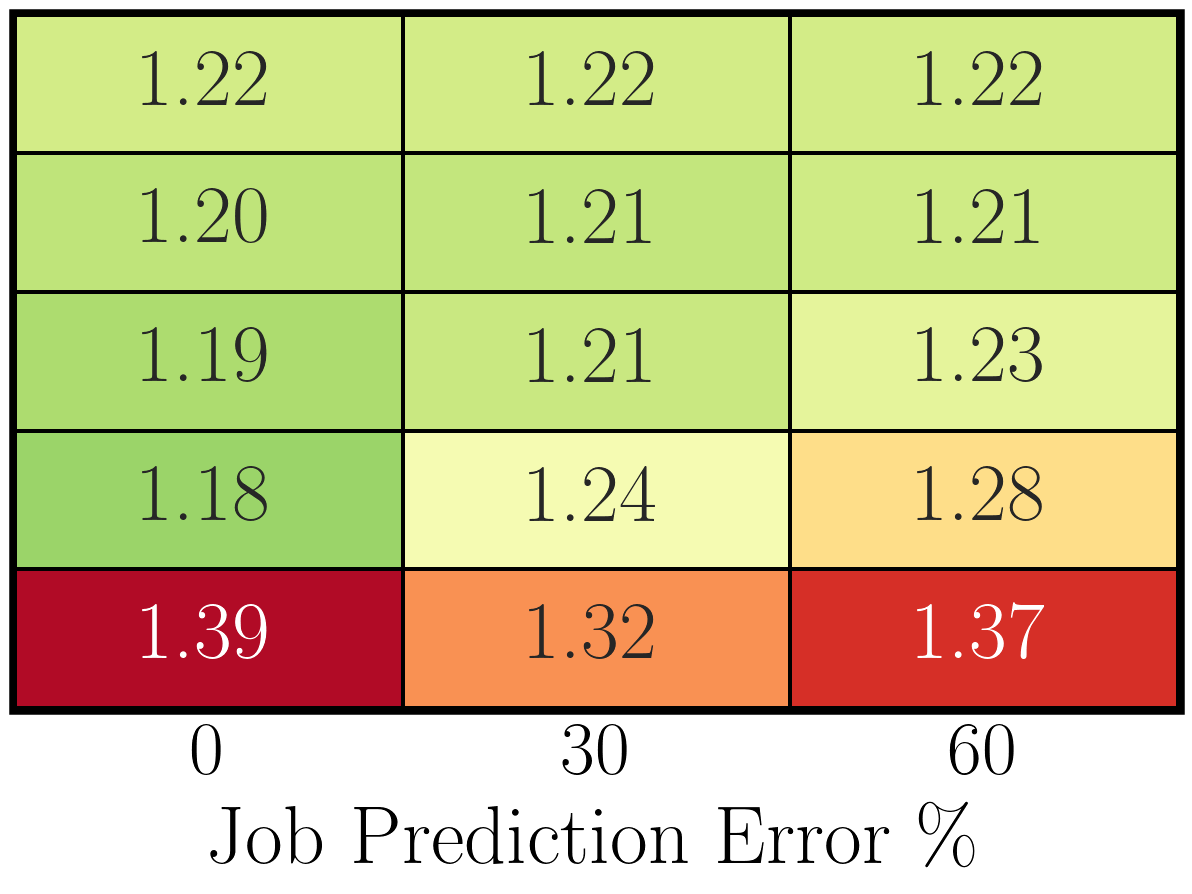}
        \vspace{-0.5cm}
        \caption{$\beta = 20, \cmax=6$}
        \label{fig:lambda_error_6_20}
    \end{subfigure}
    \hfill
    \begin{subfigure}[b]{0.232\textwidth}
        \centering
        \includegraphics[width=\textwidth]{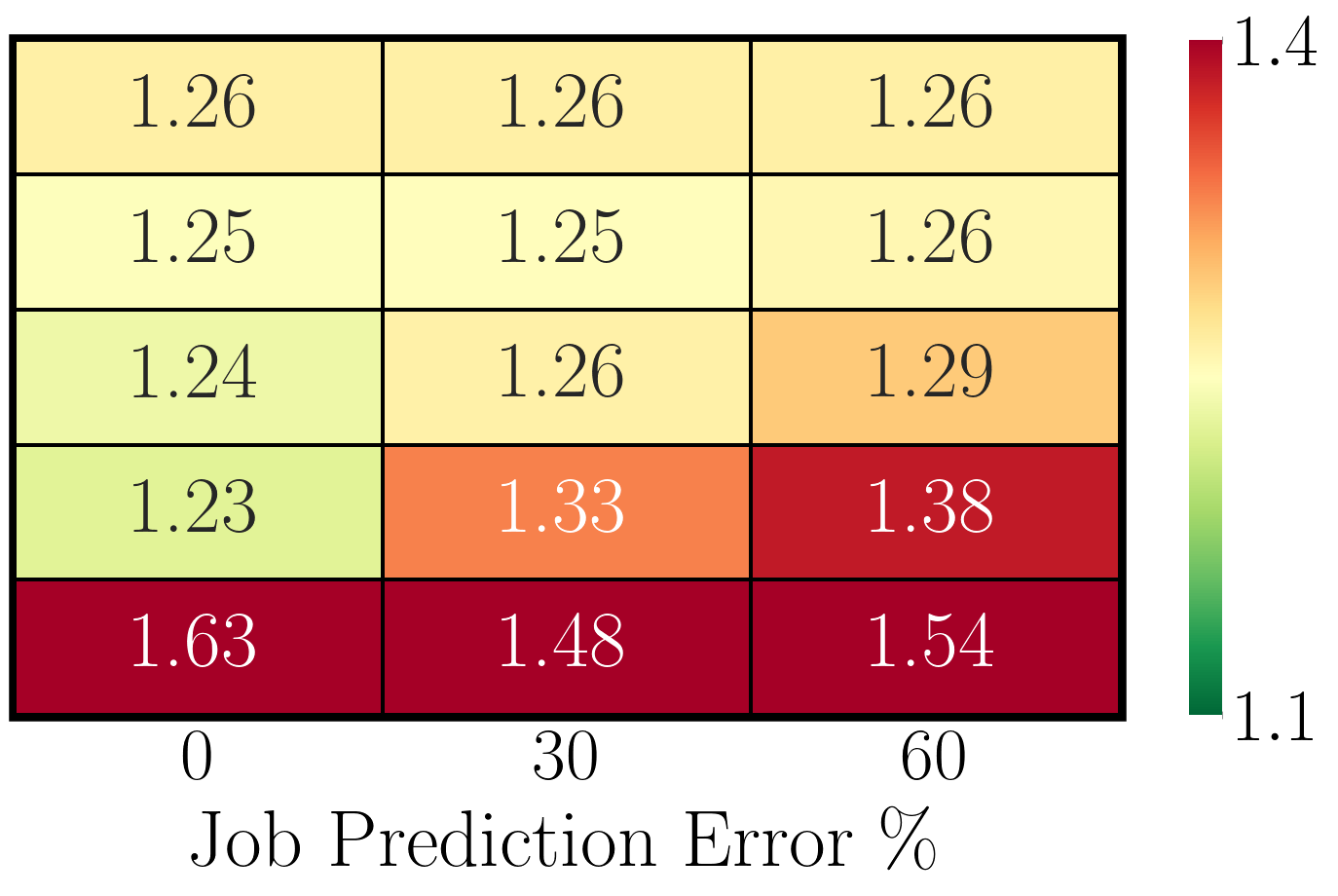}
        \vspace{-0.5cm}
        \caption{$\beta = 40, \cmax=6$}
        \label{fig:lambda_error_6_40}
    \end{subfigure}
    \vspace{-0.3cm}
    \caption{Average competitive ratios for \aug under different job prediction errors, augmentation factors $\lambda$, switching emissions $\beta$ and \cmax values for scaling profile \texttt{P1}.}
    \label{fig:lambda_error}
\end{figure*}

\subsection{Effect of Switching Emissions}
Dynamic resource adjustment leads to wasted time and energy, which result in extra emissions. For example, when applications employ suspend-resume scheduling, the switching emissions are due to the incurred energy required to checkpoint or restore the state before being able to resume processing. Since these switching emissions differ across systems and applications, e.g., due to memory size~\cite{Sharma:2016:Flint}, we evaluate the effect of the switching emissions coefficient $\beta$ on performance, where higher $\beta$ hinders the policy's ability to adapt to variations in carbon intensity. 
\autoref{fig:effect_beta} shows the performance of the algorithms against different values of $\beta$. We fix maximum job length (\cmax$=3$) while keeping the other parameters the same as \autoref{fig:effect_cMax}. 
We assigned $\beta$ as 0, 20, and 40 gCO2eq, representing 0.0, 7.3, and 14.6\% of California's average hourly carbon intensity (273 gCO2eq/kWh), respectively. 
For completeness, we added the case of $\beta=0$ where resource scaling does not incur any overheads. 
\autoref{fig:effect_beta_0} shows the case where switching is cost-free. 
We observe that \OWT  outperforms \aug, because \aug must also incorporate the robust decisions from \combalg (note that since $\beta=0$, \OWT is equivalent to \ROROpred). 
\autoref{fig:effect_beta_20} and \autoref{fig:effect_beta_40} explore realistic cases with non-zero switching emissions, where \aug nearly matches \RORO.  It lags behind on average and in the worst case by 0.6 and 10\%, respectively, when $\beta=20$ and by 1.1 and 32.3\%, respectively, when $\beta=40$. 
Finally, \autoref{fig:effect_beta_all} summarizes the average performance of all the algorithms, highlighting the effectiveness of our proposed approach. As expected, higher $\beta$ leads to higher $\nicefrac{\texttt{ALG}}{\texttt{OPT}}$ for the algorithms that do not incorporate switching emissions. However, increasing $\beta$ does not affect the relative performance of the algorithms; \aug achieves average performance within 1.2\% of \RORO and 16\% of the offline optimal, which constitutes 32\% carbon savings compared to \agnostic.

\vspace{-0.1cm}
\subsection{Effect of Job Length Prediction Error}
As explained in \autoref{sec:augmentation}, \aug augments robust algorithms with a prediction-based approach that leverages job length predictions. 
However, since typical batch job predictors are highly erroneous, we use an augmentation factor $\lambda$ that controls how much this job length prediction influences the final solution. \autoref{fig:lambda_error} shows the effect of various prediction accuracies, defined by prediction error~\%, and augmentation factors $\lambda$ on the empirical competitiveness of \aug under different maximum job lengths \cmax and switching emissions coefficients $\beta$. The scaling profile is \texttt{P1}.

The results show that as the prediction error increases, \aug should employ a lower augmentation factor for the predictions. 
For instance, at a job prediction error of 60\%, a mid-range augmentation factor (between 0.25 and 0.75) leads to better results than either extreme.  
Interestingly, the results show that higher prediction accuracy and augmentation factors do not always guarantee the highest performance. For instance, for 0\% error in job length predictions, fully augmenting with predictions does not yield the best $\nicefrac{\texttt{ALG}}{\texttt{OPT}}$. For example, when $\beta=10$, \cmax$=3$ (\autoref{fig:lambda_error_3_10}) and $\beta=40$, \cmax$=6$ (\autoref{fig:lambda_error_6_40}) setting $\lambda=0.75$ outperforms fully augmenting with predictions ($\lambda=1$) by 21\% and 33\%, respectively. 

This counter-intuitive result manifests because, during compulsory execution, \ROROpred must run with the maximum available resources, regardless of the job length prediction accuracy. This can lead to increased switching emissions compared to scenarios where robust algorithms contribute more to the decision-making in \aug, potentially reducing the time spent in the compulsory execution.  This may happen since the robust algorithm caters to the upper bound and sets a less conservative threshold. 
Higher values of $\beta$ and \cmax can intensify this phenomenon, as more switching emissions are incurred when scheduling with the maximum available resources in the compulsory execution zone. 
In summary, \aug balances the decisions from the robust algorithm and the job length predictor by using a moderate augmentation factor $\lambda$; the results indicate that aside from different values of \cmax, $\beta$, and job length prediction errors, an augmentation factor 
$\lambda=0.5$ can perform within 4\% of \RORO and within 18\% of the offline optimal, which translates to 20.8\% carbon savings over \agnostic.

\begin{figure*}[t]
    \centering    \includegraphics[width=0.45\textwidth]{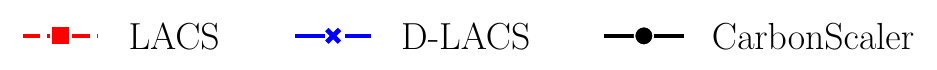}\\
    \vspace{-0.1cm}
    \begin{subfigure}[b]{0.24\textwidth}
        \centering
        \includegraphics[width=\textwidth]{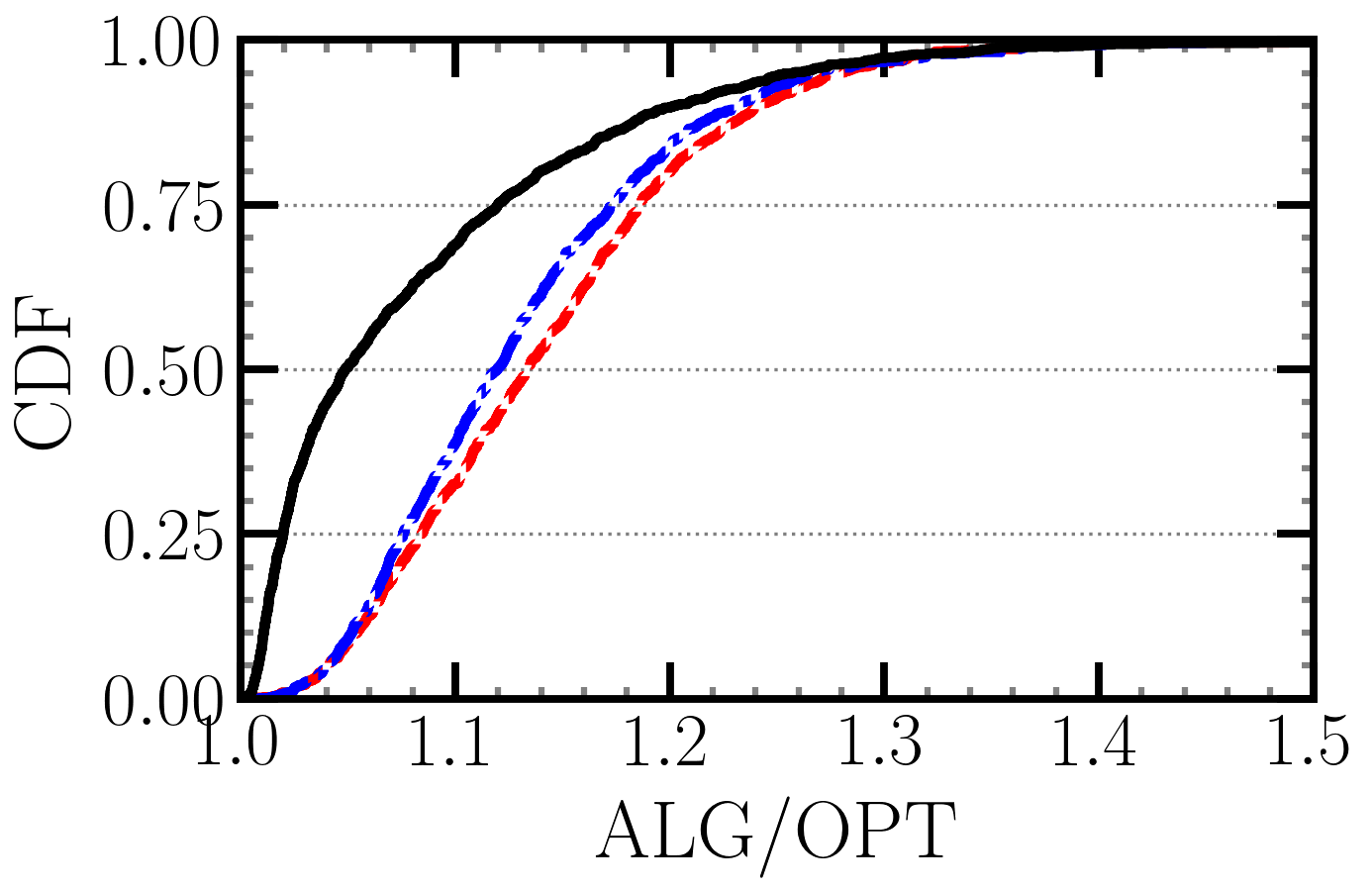}
        \vspace{-0.5cm}
        \caption{$\text{CI}_{\text{err}}=0$}
        \label{fig:carbon_err_0}
    \end{subfigure}
    \hfill
    \begin{subfigure}[b]{0.24\textwidth}
        \centering
        \includegraphics[width=\textwidth]{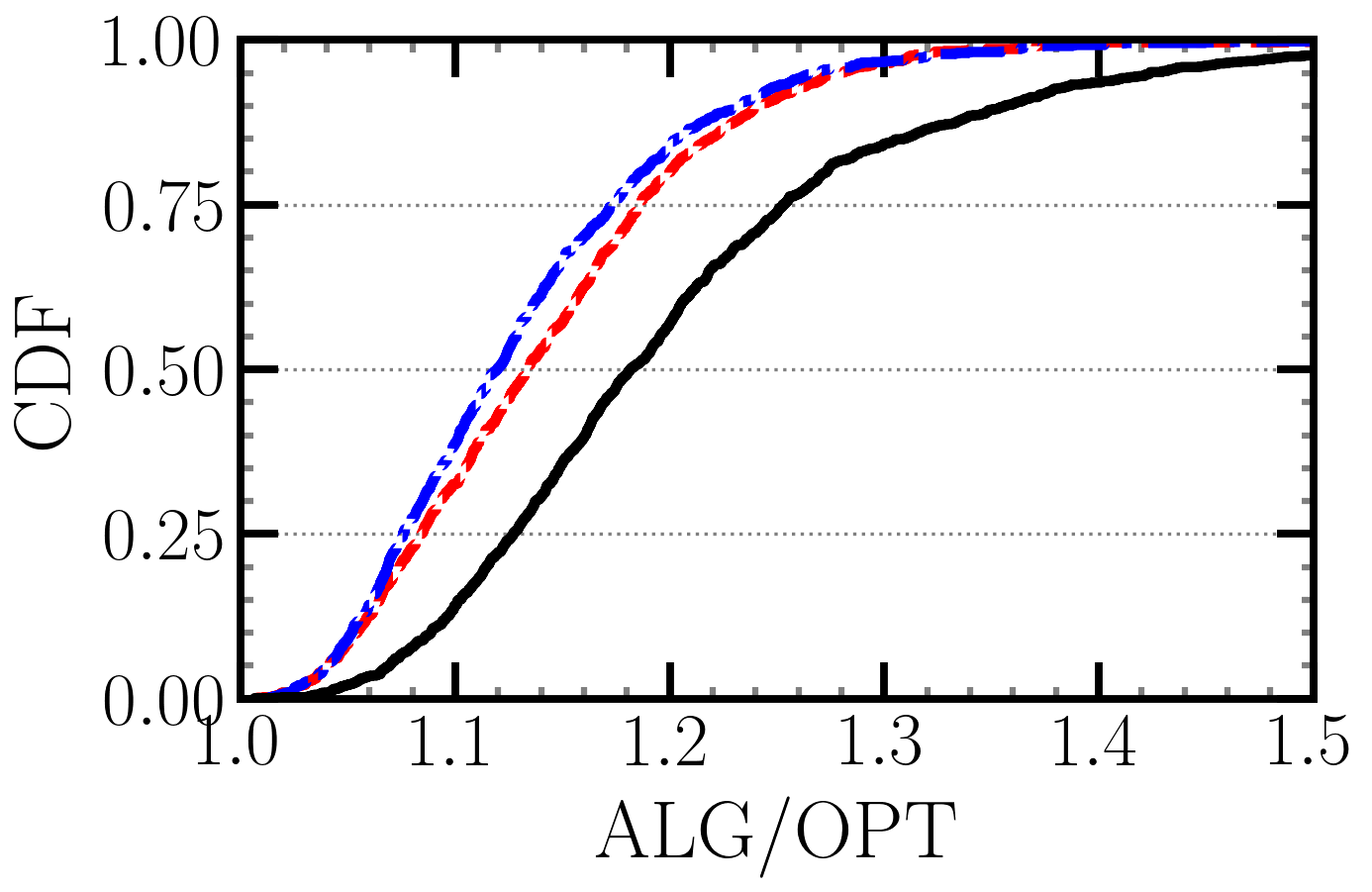}
        \vspace{-0.5cm}
        \caption{$\text{CI}_{\text{err}}=10\%$}
        \label{fig:carbon_err_10}
    \end{subfigure}
    \hfill
    \begin{subfigure}[b]{0.24\textwidth}
        \centering
        \includegraphics[width=\textwidth]{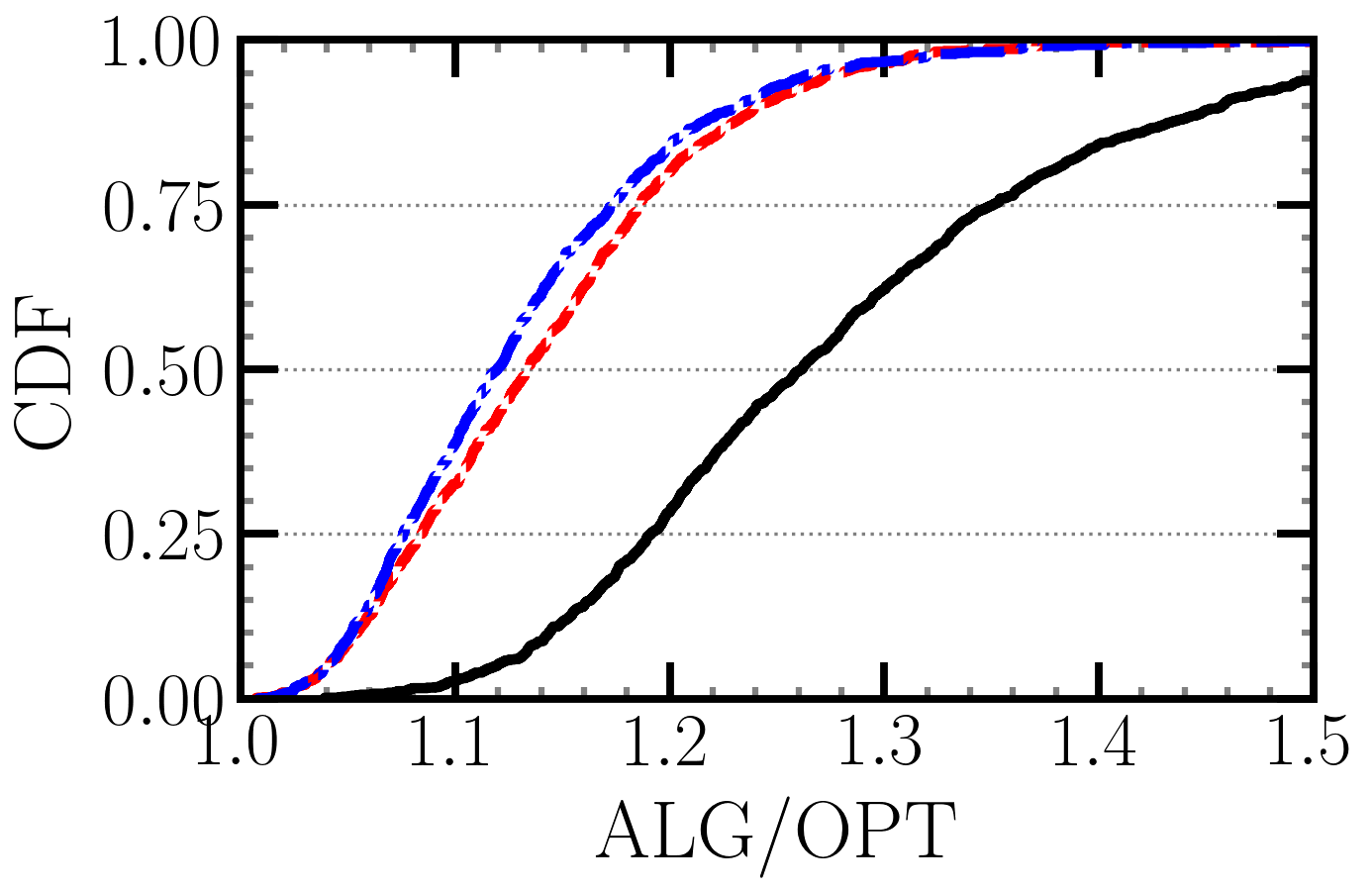}
        \vspace{-0.5cm}
        \caption{$\text{CI}_{\text{err}}=20\%$}
        \label{fig:carbon_err_20}
    \end{subfigure}
    \hfill
    \begin{subfigure}[b]{0.24\textwidth}
        \centering
        \includegraphics[width=\textwidth]{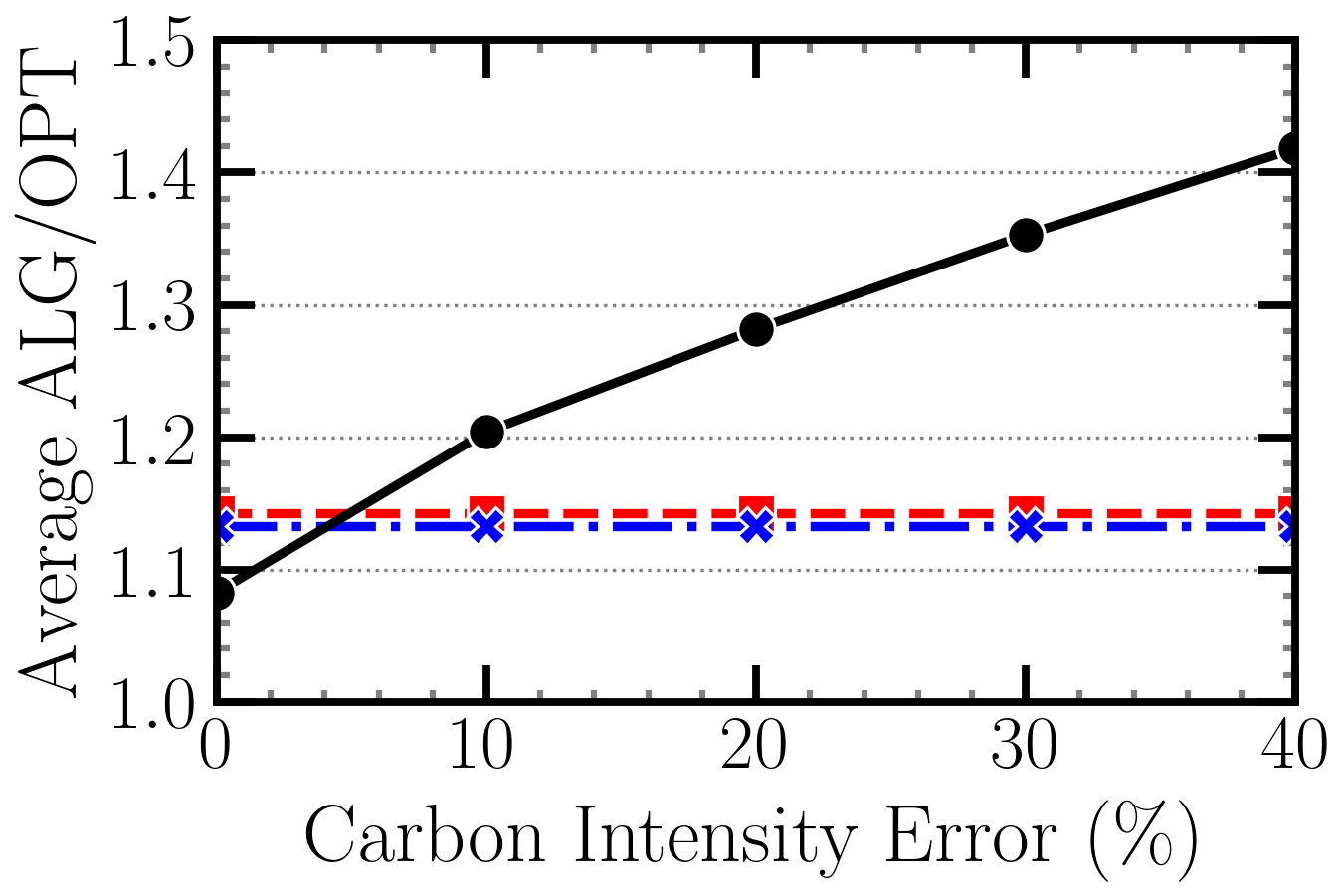}
        \vspace{-0.5cm}
    \caption{Effect of $\text{CI}_{\text{err}}$}
        \label{fig:carbon_err_all}
    \end{subfigure}
    \vspace{-0.35cm}
    \caption{(a), (b), and (c) report cumulative distribution functions (CDFs) of empirical competitive ratios for selected algorithms under different $\text{CI}_{\text{err}}$ values. (d) shows the effect of $\text{CI}_{\text{err}}$ on average $\nicefrac{\texttt{ALG}}{\texttt{OPT}}$. We assume scaling profile is \texttt{P2}, \cmax$=3$, $r=1/4$, $\beta=20$, job prediction error is $30\%$, and $\lambda=0.5$. A CDF curve towards the top left corner indicates better performance.}
    \label{fig:carbon_err}
\end{figure*}

\begin{figure}[t]
    \centering    \includegraphics[width=0.47\textwidth]{figures/sec2/integer_error/cdf_integer_labels.pdf}\\
    \vspace{-0.1cm}
     \begin{subfigure}[b]{0.236\textwidth}
        \centering
        \includegraphics[width=\textwidth]{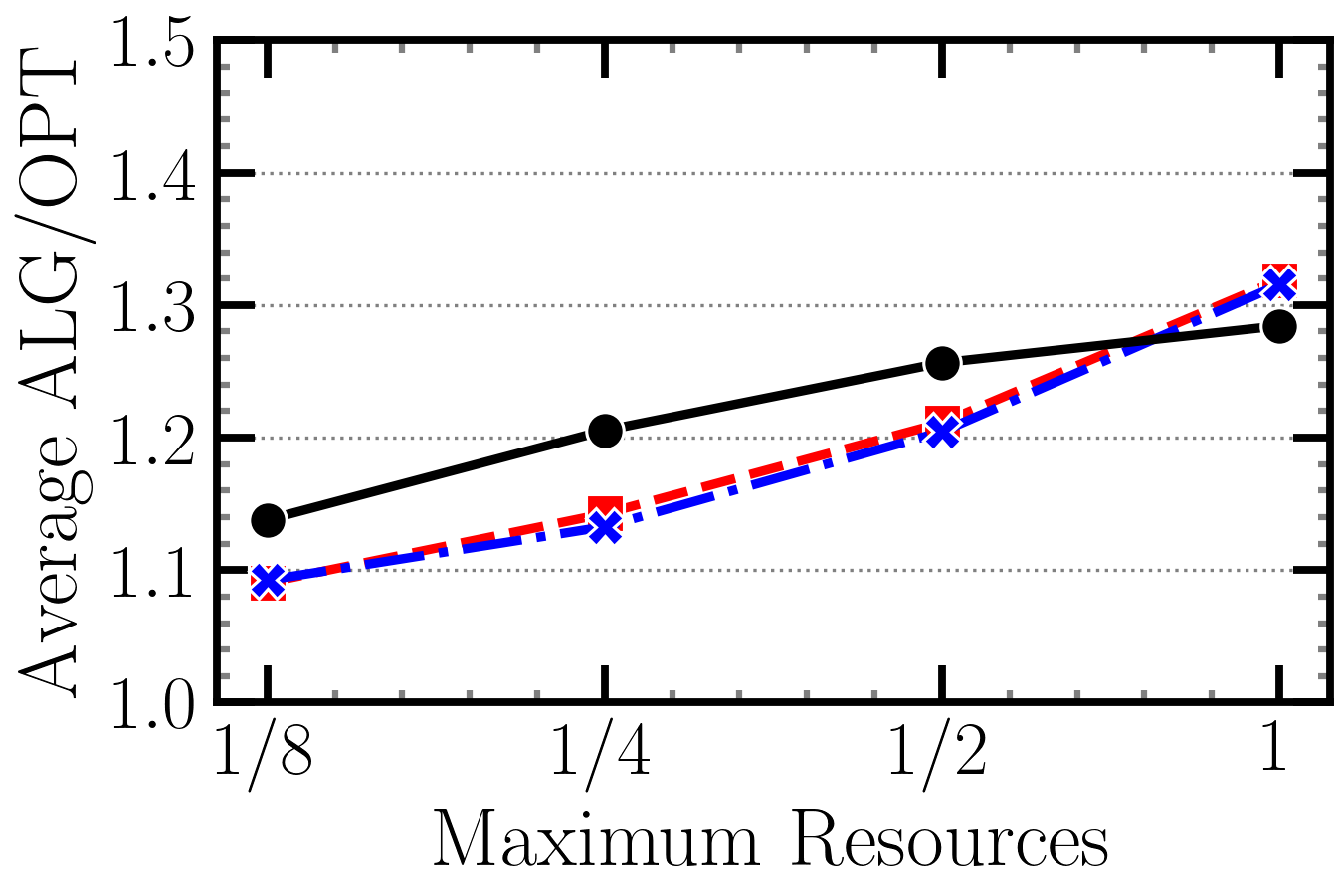}
        \vspace{-0.5cm}
        \caption{$\text{CI}_{\text{err}}=10\%$}
    \label{fig:rate_10}
    \end{subfigure}
    \begin{subfigure}[b]{0.236\textwidth}
        \centering
        \includegraphics[width=\textwidth]{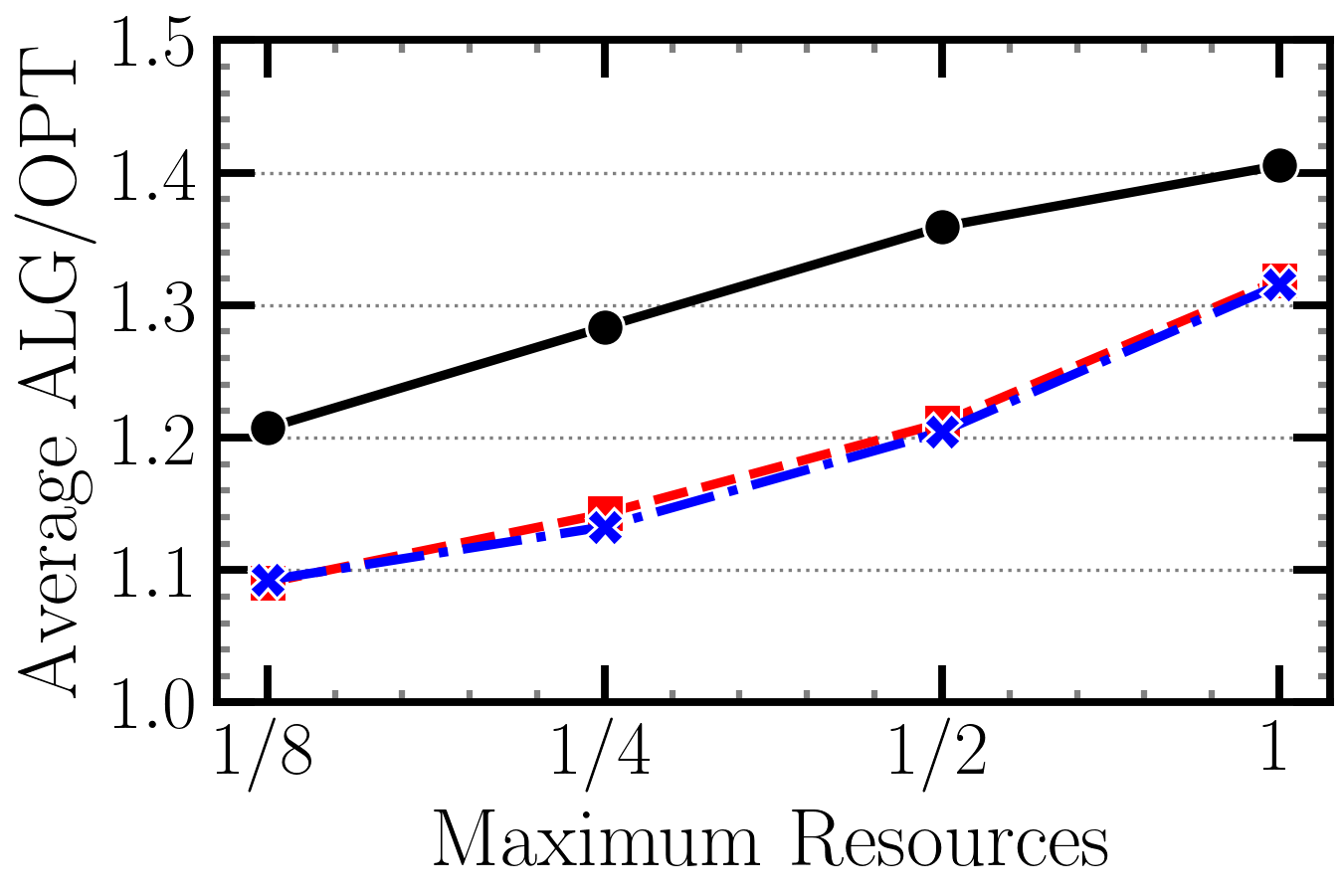}
        \vspace{-0.5cm}
        \caption{$\text{CI}_{\text{err}}=20\%$}
    \label{fig:rate_20}
    \end{subfigure}
    \vspace{-0.65cm}
    \caption{Average competitive ratios across policies under different maximum resource constraints $r$ and carbon intensity forecast error $\text{CI}_{\text{err}}$. We assume scaling profile is \texttt{P2}, $\cmax=3$, $\beta=20$, job prediction error is 30\% and $\lambda=0.5$. 
    }
    \vspace{-0.7cm}
    \label{fig:rates}
\end{figure}

\subsection{Real-world Considerations}

In the previous experiments, we assumed that resource allocation is continuous and resources can be acquired in any quantity. In practice, however, resources such as cores or servers must be allocated in discrete quantities and have physical and performance constraints. To accommodate such requirements, we discretize the scaling profile into fixed-size units, which map to different amounts of the job based on the scalability of the profile. Then, to map the scheduling decisions to discrete allocations, we round the decisions of \aug to the nearest discrete quantity, a policy we denote \texttt{Discrete-}\aug (\augd).  To create this discrete allocation, we dissect each resource unit into 8 equally sized segments.  Our evaluation shows that the discretization process has a negligible impact on our results, as observed across all the results in \autoref{fig:carbon_err} and \autoref{fig:rates}.  

Additionally, the maximum resources available to a job may be constrained for various reasons, including resource contention and cost considerations. 
To accommodate such constraints and analyze their effect, we bound the resources allocated to the job in a time slot to a maximum number $r$. 
We assume that the total number of resource units is 32, where resource $r$ takes values of $1/8,1/4, 1/2$, and $1$, denoting 4, 8, 16, and 32 resource units.
This ensures compatibility with the employed resource discretization.
In what follows, we evaluate the performance of these considerations and compare our proposed method to \carbonscaler\cite{Hanafy:23:CarbonScaler} that assumes both discrete resources and rate limits but relies on carbon intensity forecasts to make scheduling decisions.

\noindent
\textbf{Impact of carbon intensity forecast error.}
The requirement of fairly accurate carbon intensity forecasts significantly limits the practical deployment of \carbonscaler\cite{Hanafy:23:CarbonScaler}. 
We evaluate the effect of $\text{CI}_{\text{err}}$ on the performance of \carbonscaler and compare it with \aug that does not require $\text{CI}$ forecasts. 
In the experimental setup for \autoref{fig:carbon_err}, we choose profile \texttt{P2}\footnote{We do not use \carbonscaler with \texttt{P1} as it will scale to the maximum resources during expected low carbon intensity slots, exacerbating the effect of $\text{CI}_{\text{err}}$.}, impose a resource constraint ($r$=$1/4$), consider job length prediction error of 30\%, while keeping the rest of the setup the same as \autoref{fig:effect_beta}.
We evaluate the performance of \aug, \augd\footnote{We note that the optimal policy does not mandate discrete assignments but considers the rate constraints, making it a lower bound for our solutions}, and \carbonscaler under different $\text{CI}_{\text{err}}$ values.


\autoref{fig:carbon_err_0} shows the baseline case without any forecast errors, i.e., $\text{CI}_{\text{err}}=0$, and shows that although \carbonscaler outperforms \augd on average by 4.6\%, \augd outperforms \carbonscaler by 4.8\% in the worst-case.
\autoref{fig:carbon_err_10} and \autoref{fig:carbon_err_20} show more realistic scenarios where the forecasts are erroneous and demonstrate the sensitivity of \carbonscaler to carbon intensity errors. 
At 10\% error, which is equivalent to the average error rate of state-of-the-art carbon intensity forecasting models such as~\cite{Maji:22:CC}, the performance of \carbonscaler is strictly below \augd, where \augd outperforms \carbonscaler by 6.3 and 4.8\% when $\text{CI}_{\text{err}}=10$ and by 13.1 and 7.9\% when $\text{CI}_{\text{err}} =20$, on average and in the worst case, respectively. 
\autoref{fig:carbon_err_all} depicts the performance of \carbonscaler and shows the applicability of \augd in the real world with erroneous or unavailable carbon intensity forecasts.

\noindent
\textbf{Impact of resource constraints.}
\autoref{fig:rates} shows the performance of \aug, \augd and \carbonscaler as a function of resource constraints. 
We set the scaling profile (\texttt{P2}), $\cmax=3$, switching emissions coefficient $\beta=20$, the job prediction error to 30\%, and augmentation factor as $\lambda=0.5$. 
The results indicate that enforcing a lower resource constraint narrows the gap between algorithms and the offline optimal, where the average performance of \augd ranges from 9 to 31\% of the offline optimal, rending  21 to 37\% carbon savings compared to the \agnostic policy. This is reasonable as a lower resource constraint means a lower degree of freedom, forcing all approaches to run similarly. 
Nevertheless, \augd outperforms \carbonscaler in almost all cases. 
For example, when $r=1/8$,  \augd outperforms \carbonscaler by 3.6 and 10\% when  $\text{CI}_{\text{err}}=10$ and  $\text{CI}_{\text{err}}=20$, respectively. 
The results are consistent with previous experiments, where higher $\text{CI}_{\text{err}}$ yields a higher gap between \augd and \carbonscaler, reaching 13.3\% when $\text{CI}_{\text{err}}=20$ and $r=1/4$. 
In the absence of resource constraints, the gap between \carbonscaler and \augd shrinks as \carbonscaler can fully scale up when encountering a good carbon intensity, possibly avoiding compulsory execution.
\begin{figure}
    \centering
    \includegraphics[width=0.4\textwidth]{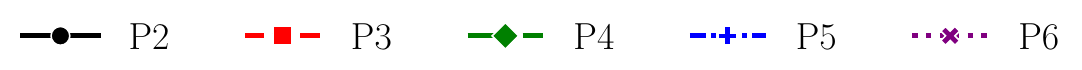}\\
     \begin{subfigure}[b]{0.236\textwidth}
        \centering
        \includegraphics[width=\textwidth]{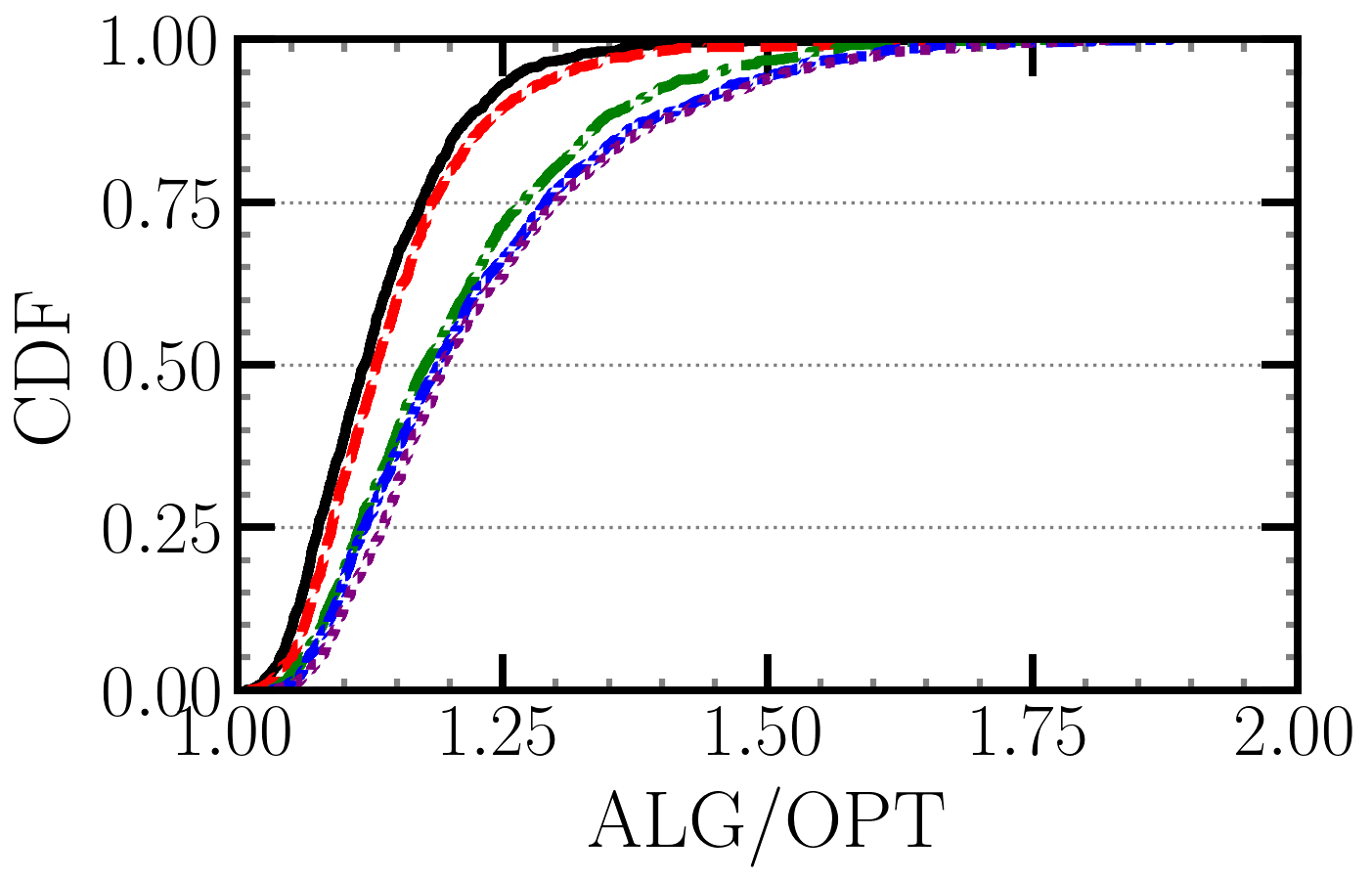}
        \vspace{-0.6cm}
        \caption{Maximum resources is 1/4}
    \label{fig:profiles_dlacs_8}
    \end{subfigure}
    \begin{subfigure}[b]{0.236\textwidth}
        \centering
        \includegraphics[width=\textwidth]{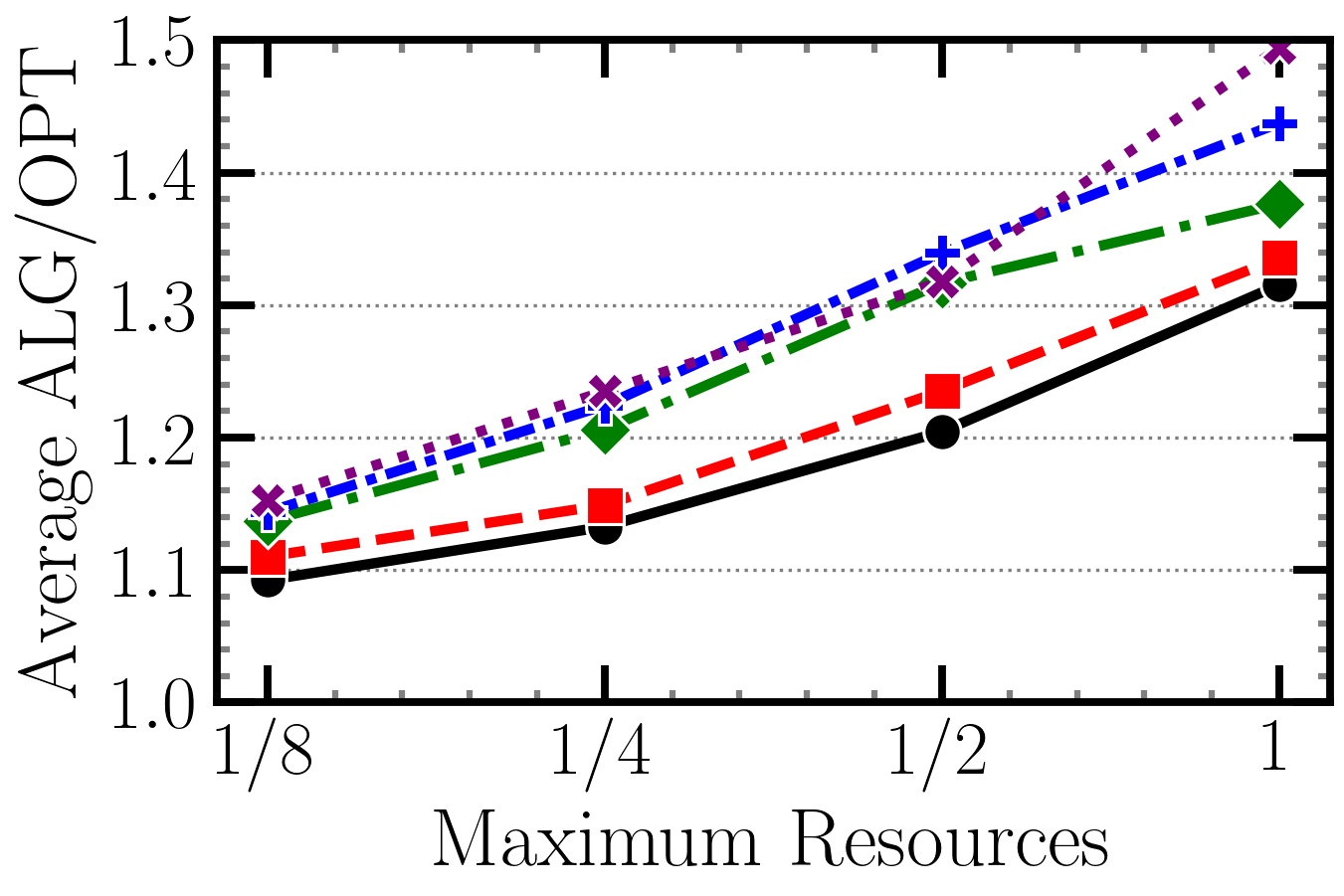}
        \vspace{-0.6cm}
        \caption{Effect of maximum resources}
    \label{fig:profiles_dlacs_all}
    \end{subfigure}
    \vspace{-0.65cm}
    \caption{Effect of profiles on Average competitive ratio across scaling profiles from \sref{Table}{tab:ScalingProfiles}. We assume \cmax$=3$, $r=1/4$, $\beta=20$ job prediction error is $30\%$, $\text{CI}_{\text{err}}=10$, and $\lambda=0.5$.}
    \label{fig:profiles_dlacs}
    \vspace{-0.6cm}
\end{figure}

\subsection{Effect of Scaling Profiles}
Many factors, including the network speed, and the ratio between workload's communication and computation, determine an application's scalability. To evaluate the effect of scaling profiles, we test various profiles (see \autoref{tab:ScalingProfiles}), which represent possible scaling profiles seen in the real world. 
\autoref{fig:profiles_dlacs} evaluates the performance of \augd for these different profiles. We set $\cmax=3$, switching emissions coefficient $\beta=20$, the job prediction error is 30\%, the carbon intensity forecast error is set to $\text{CI}_{\text{err}}=10\%$, and augmentation factor $\lambda=0.5$. \autoref{fig:profiles_dlacs_8} shows the performance of \augd across different profiles; \augd is more effective at higher scalability, but shows comparable competitiveness across profiles. 
As resources increase, jobs with worse scaling profiles will be more conservative in their decisions, forcing them to run at inefficient scales during the compulsory execution. \autoref{fig:profiles_dlacs_all}, depicts this behavior across scaling behavior and rates, showing that all scaling profiles experience less performance as the rate increases, which is consistent with earlier results.
The results indicate that across scaling profiles \augd performs between 9 and 15\% and 31 and 49\% of the offline optimal for $r=1/8$ and $r=1$, respectively.

%% file: 6-relwork.tex
\begin{table}[t]
\caption{Carbon-aware temporal shifting and scaling}
\vspace{-0.3cm}
\label{tab:related}
\resizebox{0.45\textwidth}{!}{%
\begin{tabular}{|c|c|c|c|c|c|}
\hline
Algorithm  & \begin{tabular}[c]{@{}c@{}}Unknown\\Job Length\end{tabular} & \begin{tabular}[c]{@{}c@{}}Forecast Not\\ Required\end{tabular} & Deadline & \begin{tabular}[c]{@{}c@{}}Switching\\ Cost \end{tabular} & \begin{tabular}[c]{@{}c@{}}Decision\\Space \end{tabular} \\ \hline \hline
WaitAwhile-Thr.~\cite{Wiesner:21} &    Yes    &    Yes     &   No    &   No    &     Shift           \\ \hline
WaitAwhile~\cite{Wiesner:21} &   No     &    No     &   Yes   &    No    &    Shift          \\ \hline
$k$-min search~\cite{Lorenz:08}  &  No  &    Yes     &   Yes  &   No   &    Shift    \\ \hline
Double Threshold~\cite{Lechowicz:23}  &  No  &    Yes     &   Yes  &   Yes   &    Shift    \\ \hline
Wait\&Scale~\cite{Souza:23}  & Yes   &    No     &  No   &  No   &    Scale    \\ \hline
CarbonScaler~\cite{Hanafy:23:CarbonScaler}  & No   &    No   & Yes   &   No   &    Scale    \\ \hline
$\mathsf{OWT}$~\cite{ElYaniv:01, SunZeynali:20}  &  No  &    Yes     &   Yes  &   No   &    Scale    \\ \hline
$\RORO$~\cite{Lechowicz:24}  & No   &    Yes   & Yes    &   Yes   &    Scale    \\ \hline\hline
\textbf{This work}  & Yes   &    Yes  &   Yes  & Yes &    Scale    \\ \hline
\end{tabular}
}
\vspace{-0.4cm}
\end{table}

To bridge the gap between the availability of low-carbon energy and demand, carbon-aware schedulers utilize the inherent flexibility of workloads to select an appropriate time and location to execute the workloads~\cite{sukprasert2023quantifying, Hanafy:23:CarbonScaler, Lechowicz:23, Lechowicz:24, Hanafy:23:War, Souza:23, Wiesner:21, Dodge:2022, Kim:2023, Thiede:2023:Containers, Zheng:2020:Curtailment, Souza:2023:Casper}. 
In this paper, we focus on a special case of carbon-aware temporal shifting of batch jobs, where schedulers decide on a scale factor at each time step ranging from 0 (i.e., suspending) to a user-defined max factor. 
In addition, the problem of carbon-aware scheduling has historically been closest to \textit{online search problems} such as $k$-min search~\cite{Lorenz:08, Lee:24}, one-way trading~\cite{ElYaniv:01, SunZeynali:20, mohr2014online, SunLee:21, Damaschke:07}, and   
online knapsack~\cite{Marchetti:95, Zhou:08, Bockenhauer:14, Cygan:16, Zeynali:21, Yang2021Competitive, sun2022online}.  These have seen applications in e.g., cloud pricing~\cite{zhang2017optimal}, EV charging~\cite{SunZeynali:20, Bostandoost:23}, and stock trading~\cite{Lorenz:08}, among others.
Recently, two studies have explicitly extended online search ideas towards carbon-aware problems.  In~\cite{Lechowicz:23}, the authors present online pause and resume, which extends the $k$-min search problem to incorporate switching costs.  Similarly, the authors in~\cite{Lechowicz:24} present the $\mathsf{OCS}$ problem discussed in this paper, which introduces a linear switching cost into the formulation of one-way trading.
In contrast to all of the above literature, the \OCSUmin problem we propose is the first online search-type setting where the job length (i.e., demand) is uncertain.  
We summarize state-of-the-art work in carbon-aware scheduling and applicable methods in online scheduling in  
\sref{Table}{tab:related}, highlighting different assumptions regarding job length, dependency on carbon intensity forecasts, and compliance with deadlines. 

%% file: conclusion.tex
This paper introduces \aug, a novel learning-augmented carbon-aware algorithm for resource scaling of computing workloads with uncertain job lengths. We analyzed the theoretical performance of \lacs using the framework of robustness-consistency of competitive online algorithms. Further, we evaluated the empirical performance of \lacs through extensive experimental, showing its superior performance as compared to an extensive set of baseline algorithms. 
\lacs, to be best of our knowledge, is the first algorithm with both theoretical guarantees and promising practical performance for carbon-aware resource scaling with unknown job lengths. 


%% file: Z-appendix.tex
\section{Deferred Proofs from \sref{Section}{sec:analysis}} \label{apx:analysis}

\subsection{Competitive Proofs for \ROROcmax} \label{apx:comp-alg1}
\input{Z-cmax}

\subsection{Competitive Proofs for \ROROcmin} \label{apx:comp-alg2}
\input{Z-cmin}

\subsection{Analyzing the impact of rate constraints} \label{apx:rate}

We note that the rate constraint $d_t : t \in [T]$ surprisingly does not appear in the worst-case analysis of \ROROcmax and \ROROcmin.  In this section, we give intuitive justification to explain this dynamic for completeness.

Specifically, we show that when a threshold-based algorithm (i.e., \ROROcmax or \ROROcmin) achieves a certain competitive ratio for \OCSUmin when $d_t = c \ \ \forall t \in [j]$ (i.e., when the rate allows completing the entire job in a single time slot), the worst-case competitive ratio will stay constant if $d_t < c \ \forall t \in [j]$.

\begin{lemma} \label{lem:rate-const}
Let $\texttt{ALG}$ denote a threshold-based algorithm for \OCSUmin which uses a threshold function $\psi(w)$.  Suppose $\texttt{ALG}$ is $\eta$-competitive when $d_t = c \ \forall t \in [j]$. 
If $d_t < c \ \forall t \in [j]$, the competitive ratio of $\texttt{ALG}$ is still upper bounded by $\eta$.
\end{lemma}
\begin{proof}
We consider the case where a rate constraint $d_t < c$ causes $\texttt{ALG}$ to make a decision which violates its worst-case competitive ratio of $\eta$.  At any arbitrary time slot $t$, the disconnect between the setting where $x_t \in [0, c]$ and the setting with rate constraints $< c$, where $x_t \in [0, d_t]$, is that $x_t$ cannot be~$> d_t$.  Intuitively, a challenging situation for $\texttt{ALG}$ under such a constraint is the case where $\texttt{ALG}$ would otherwise run more than $> d_t$ of the job during a period of ``good'' carbon intensity (before the compulsory execution), but it is now constrained from doing so.

We now show that such a situation implies that $\texttt{ALG}$ achieves a worst-case competitive ratio which is equal to or better than $\eta$.  Recall that $w^{(t)}$ denotes the progress of $\texttt{ALG}$ after time step $t$.

For an instance $\mathcal{I} \in \Omega$ and an arbitrary time step $m$, let $w^{(m)} = w^{(m-1)} + d_m$, implying that $x_m = d_m$.  For the sake of comparison, we first consider this time step with a cost function $g_m(\cdot)$ such that $g_m(x_m) + \beta \lvert x_m - x_{m-1} \rvert = \psi(w^{(m)}) x_m$, implying that even if the rate constraint was not present, $\texttt{ALG}$ would set $x_m = d_m$.  If no more of the job is completed by $\texttt{ALG}$ after time step $m$ (ignoring the compulsory execution), we know that $\texttt{ALG}$ is $\eta$-competitive (e.g., for \ROROcmax, $\eta = \alpha_1$, and for \ROROcmin, $\eta = \alpha_2$).

Consider the exact same setting, except with a substituted cost function $g_m'(\cdot)$, such that $g_m'(x_m) + \beta \lvert x_m - x_{m-1} \rvert < \psi(w^{(m)}) \cdot x_m$.  We denote this new instance by $\mathcal{I}'$.  This implies that without the presence of a rate constraint, $\texttt{ALG}$ would set $x_m > d_m$.  In other words, $g_m'( \cdot)$ has a ``good carbon intensity'', but $\texttt{ALG}$ cannot run as much of the job as it otherwise should due to the rate constraint.  

Note that $\OPT$ is also subject to the same rate constraint $d_m$.  Thus, we know that $\OPT(\mathcal{I}')$ is lower bounded by $[ \psi(w^{(m)}) - \beta ] (1 - d_m) + g_m'(d_m)$ -- the rest of the optimal solution is bounded by the final threshold value, since we assume that no more of the job is completed by $\texttt{ALG}$ after time step $m$.  

The worst-case carbon emission of $\texttt{ALG}$ is upper bounded by $\texttt{ALG}(\mathcal{I}') \le \texttt{ALG}(\mathcal{I}) - \int^{w^{(m)}}_{w^{(m-1)}} \psi(u) du + g_m'(d_m)$, which follows since we substitute the last portion of the threshold function (of ``width'' $d_m$) with the new cost function $g_m'(d_m)$.

Compared to the previous setting of $\mathcal{I}$, the $\OPT$ and $\texttt{ALG}$ solutions have both decreased -- $\OPT(\mathcal{I}')$ has decreased by a factor of $g_m'(d_m) - [ \psi(w^{(m)}) - \beta ] d_m$, while $\texttt{ALG}(\mathcal{I}')$ has decreased by a factor of $g_m'(d_m) - \int^{w^{(m)}}_{w^{(m-1)}} \psi(u) du$.

However, note that since $\psi$ is monotonically decreasing in $w$, by definition, $[ \psi(w^{(m)}) - \beta ] d_m < \int^{w^{(m)}}_{w^{(m-1)}} \psi(u) du$.  Thus, the cost of $\texttt{ALG}$ has improved \textit{more} than the cost of $\OPT$.  This then implies the following:
\begin{align*}
\frac{\texttt{ALG}(\mathcal{I})}{\OPT(\mathcal{I})} \leq \frac{\texttt{ALG}(\mathcal{I}) - \int^{w^{(m)}}_{w^{(m-1)}} \psi(u) du + g_m'(d_m)}{[ \psi(w^{(m)}) - \beta ] (c - d_m) + g_t'(d_m)} < \eta.
\end{align*}
At a high-level, this result shows that even if there is a rate constraint which prevents $\texttt{ALG}$ from accepting a good carbon intensity, the worst-case competitive ratio does not change.
\end{proof}

\subsection{\aug Consistency and Robustness for \OCSUmin} \label{apx:aug-con-rob}
\input{Z-aug}

%% file: Z-cmax.tex
 In this section, we provide the proof of \autoref{th:alg1-alpha} that states \ROROcmax for \OCSUmin is $\alpha_1$-competitive, where $\alpha_1$ is defined in \autoref{eq:alg1-alpha}.

\begin{proof}[Proof of \autoref{th:alg1-alpha}]
Let $\cali \in \Omega$ denote any valid sequence, and let $w^{(j)}$ denote \ROROcmax's final progress before the compulsory execution that begins at time $j \leq T$. Note that $w^{(j)} \in [0, $\cmax$]$. Assuming a job arrives with length $c$, the worst performance of \ROROcmax happens when $c=\cmin$, and \ROROcmax is $\alpha_1$-competitive. In this scenario, we have two cases:

\noindent\textbf{Case (i).} When the job is finished during the compulsory execution.

In this case, $w^{(j)} < \cmin$, and according to \textit{lemma B.2} in \cite{Lechowicz:24}, the offline optimum is lower-bounded by:
\begin{align}
    \texttt{OPT}(\cali) \geq \cmin (\phi_1(w^{(j)}) - \beta) \label{case2-low-opt}
\end{align}

Additionally, according to \textit{lemma B.3} in \cite{Lechowicz:24}, the online carbon emissions is upper-bounded by:
\begin{align}
    \ROROcmax(\cali) \leq \int_{0}^{w^{(j)}} \phi_1(u) du + w^{(j)} \beta + (\cmin - w^{(j)})U \label{cmax-up-bound} 
\end{align}

By combining inequalities \eqref{case2-low-opt} and \eqref{cmax-up-bound}, we have:
\begin{align}
    \frac{\ROROcmax (\cali)} {\texttt{OPT}(\cali)} \leq \frac{\int_{0}^{w^{(j)}} \phi_1(u) du + w^{(j)} \beta + (\cmin - w^{(j)})U}{\cmin (\phi_1(w^{(j)}) - \beta)} = \alpha
\end{align}

\noindent\textbf{Case (ii).} When the job is done without compulsory execution.
\begin{lemma}
    The carbon emissions of \ROROcmax is upper bounded by: 
    \begin{align}
        \ROROcmax(\cali) \leq \cmin \phi_1(0) + \cmin \beta  = \cmin \frac{U}{\alpha} + \cmin \beta
    \end{align}
\begin{proof}
    According to \textit{lemma B.3} in \cite{Lechowicz:24}, when the job length = c, and the job is done with no compulsory execution, \ROROcmax $\leq \int_{0}^{c} \phi_1(u) du + \beta c$. In the worst-case scenario, the job length is \cmin; hence, we have:
    \begin{align}
       \ROROcmax (\cali) \leq \int_{0}^{\cmin} \phi_1(u) du + \beta \cmin &\leq \cmin \phi_1(0) + \cmin \beta \label{case1-up-bound-alg1} = \cmin \frac{U}{\alpha}+ 2 \beta \cmin 
    \end{align}
\end{proof}
\end{lemma}

The right hand side of inequality \eqref{case1-up-bound-alg1} stands correct because $\phi_1(u)$ is monotonically decreasing, and $\phi_1(u) \leq \phi_1(0)=\frac{U}{\alpha}+\beta$. 

The lower bound of $\texttt{OPT}(\cali)$ is:
\begin{align}
    \texttt{OPT}(\cali) \geq \cmin L \label{case1-low-bound-opt1}
\end{align}

Combining inequalities \eqref{case1-up-bound-alg1} and \eqref{case1-low-bound-opt1}, we have:
\begin{align}
    \frac{\ROROcmax (\cali)}{\texttt{OPT}(\cali)} \leq \frac{\cmin \frac{U}{\alpha}+ 2 \beta \cmin}{\cmin L} = \frac{U}{\alpha L} + \frac{2\beta}{L} = \alpha_{1b}
\end{align}
Comparing $\alpha_{1b}$ against $\alpha$, we note that when \ROROcmax completes the job without compulsory execution, it has \textit{``overspent''} compared to the \RORO algorithm which knows the actual job length.  This follows by recalling that the threshold functions of both algorithms are monotonically decreasing, and observing their respective threshold values at $w^{(j)} = \cmin$, which can be expressed as follows: $\phi_{\OCSUmin}(\cmin) = L + \beta < \phi_1(\cmin)$.  Thus, \ROROcmax allows a strictly higher carbon intensity to run the job with length \cmin, and the competitive ratio $\alpha_{1b}$ is worse.

Therefore, since $\alpha_1 \coloneqq \max \{\alpha_{1b}, \alpha\}$, we obtain the following final competitive bound:
\begin{align*}
    \alpha_1 = \frac{U}{\alpha L} + \frac{2\beta}{L}
\end{align*}
\end{proof}

%% file: Z-cmin.tex
In this section, we provide the proof of \autoref{th:alg2-alpha} that states the instantiation of \ROROcmin for \OCSUmin is $\alpha_2$-competitive, where $\alpha_2 = \alpha'$ and $\alpha'$ is defined in \autoref{eq:alphaprime}.

\begin{proof}[Proof of \autoref{th:alg2-alpha}]

\noindent\textbf{Roadmap of the proof:} We initially consider an intermediate algorithm, named \tmpalgg. This algorithm operates under the assumption that the length of each job is \cmin  and schedules the job based on a threshold function $\hat{\phi}(\hat{w}), \hat{w} \in [0, \cmin]$ (\sref{Definition}{def:alg2-thresh-hat}). We first compute this threshold function and the competitive ratio, preparing for the worst-case scenario where the actual job length is \cmax (\textbf{Step 1}).

\noindent Subsequently, we adapt the threshold function $\hat{\phi}(\hat{w})$, scaling it up by a factor of \cmax/\cmin to obtain a new threshold function $\phi_2(w), w \in [0,\cmax]$ (\autoref{eq:alg2-thresh}), which defines the operation of \ROROcmin -- we proceed to determine the competitive ratio ($\alpha_2$) for \ROROcmin, finding that it is less than that of \tmpalgg. This outcome reinforces the rationale behind choosing \ROROcmin over \tmpalgg, as \ROROcmin proves to be a more strategically sound choice under the given conditions.(\textbf{Step 2})

\noindent\textbf{Step 1}: Let $\cali \in \Omega$ denote any valid \OCSUmin sequence, and let $\hat{w}^j$ be the \tmpalgg's final progress before the compulsory execution, which begins at time step $j \leq T$, and $\hat{w}^{(j)} \in [0, \cmin]$.

\begin{definition}\label{def:alg2-thresh-hat}
Threshold function $\hat{\phi}$ for \OCSUmin solved by \ROROcmin for any progress $\hat{w} \in [0, \cmin]$:
    \begin{align}
        \hat{\phi}(\hat{w}) = U-\beta + \left( \frac{U}{ \alpha'}-U+2\beta \right) \exp \left( \frac{\hat{w}}{\cmin\alpha'} \right) \label{eq:alg2-thresh-hat} 
    \end{align}
\end{definition}

\begin{lemma}\label{lem:alg2-opt-low-bound}
    The offline optimum is lower bounded by $\texttt{OPT}(\cali) \ge (\hat{\phi}(\hat{w}^{(j)})-\beta)\cmin$.
\end{lemma}
\begin{proof}
    According to \textit{lemma B.2} in \cite{Lechowicz:24}, 
    the optimal offline strategy, setting aside any additional switching emissions, involves completing the job at the most favorable cost function within the sequence $\{g_t(\cdot)_{t \in [T]}\}$. Suppose the best cost function occurs at an arbitrary step $m$ ($m \in [T], m \leq j$), denoted by $g_m(.)$. 
    
    \textit{Lemma B.2} in \cite{Lechowicz:24} further states that for any job length c, $\OPT(\cali) = c(\diffp{g_m}{x}) \geq c(\hat{\phi}(\hat{w}^{(j)})-\beta),\quad \hat{w}^{(j)} \in [0, c]$, where $\hat{w}^{(j)}$ is the progress before the compulsory execution. Given that $c \in [\cmin, \cmax]$ in \OCSUmin problem; the lower bound of \OPT is when $c = \cmin$ , mirroring the assumption made in the \tmpalgg algorithm; therefore, $\OPT(\cali) \geq (\hat{\phi}(\hat{w}^{(j)})-\beta)\cmin,\quad \hat{w}^{(j)} \in [0, \cmin]$   
\end{proof}

\begin{lemma} \label{lem:alg2-up-bound}
    The carbon emissions of $\tmpalgg(\cali)$ is upper-bounded by:
    \begin{align}
        \tmpalgg(\cali) \leq \int_{0}^{\hat{w}^{(j)}} \hat{\phi} (u) du + \beta \hat{w}^{(j)} + (\cmax-\hat{w}^{(j)})U 
    \end{align}
\end{lemma}
\begin{proof}
    According to \textit{lemma B.3} in \cite{Lechowicz:24}, $\RORO(\cali)$ incurred carbon emissions for a job with a length of $c$ is upper-bounded by:
    \begin{align}
       \RORO(\cali) \leq \int_{0}^{{w}^{(j)}} \phi_{\texttt{OCS-min}} (u) du + \beta {w}^{(j)} + (c-{w}^{(j)})U, {w}^{(j)} \in [0, c]
    \end{align}
    where ${w}^{(j)}$ is the progress at time $j$ ($j \leq T$) before the compulsory execution. The term $(c-{w}^{(j)})U$ denotes the maximum carbon emissions while doing the compulsory execution to satisfy constraint \eqref{align:deadline}. Since \tmpalgg is not aware of the actual job length ($c$), the maximum carbon emitted during the compulsory execution is $(\cmax-\hat{w}^{(j)})U$. Therefore $\tmpalgg(\cali)$ is upper-bounded by $\int_{0}^{\hat{w}^{(j)}} \hat{\phi} (u) du + \beta \hat{w}^{(j)} + (\cmax-\hat{w}^{(j)})U, \ : \ \hat{w}^{(j)} \in [0, \cmin]$
\end{proof}

Combining \sref{Lemma}{lem:alg2-opt-low-bound} and \sref{Lemma}{lem:alg2-up-bound} gives:
\begin{align}
    \frac{\tmpalgg(\cali)}{\OPT(\cali)} &\leq \frac{\int_{0}^{\hat{w}^{(j)}} \hat{\phi} (u) du + \beta \hat{w}^{(j)} + (\cmax-\hat{w}^{(j)})U}{(\hat{\phi}(\hat{w}^{(j)})-\beta)\cmin} \leq \alpha' + \frac{U(\cmax-\cmin)}{\cmin(\hat{\phi}(\hat{w}^{(j)})-\beta)} = \alpha_{2a}
\end{align}

where the last inequality holds since for any $\hat{w} \in [0,\cmin]$:
\begingroup
\allowdisplaybreaks
\begin{align}
   &\int_{0}^{\hat{w}^{(j)}} \hat{\phi} (u) du + \beta \hat{w}^{(j)} + (\cmax-\hat{w}^{(j)})U \\
   &= \int_{0}^{\hat{w}^{(j)}} \left[ U-\beta+ \left(\frac{U}{\alpha'}-U+2\beta \right) \exp \left(\frac{\hat{w}^{(j)}}{\alpha' \cmin} \right) \right] + \beta \hat{w}^{(j)} + (\cmax - \hat{w}^{(j)})U  \\
   &= \alpha' \cmin \left( \frac{U}{\alpha'}-U+2\beta \right) \left[ \exp \left(\frac{w}{\alpha' \cmin} \right)-1 \right] + \beta \hat{w}^{(j)} + (\cmax - \hat{w}^{(j)})U + (U - \beta)\hat{w}^{(j)} \\
   &= \alpha' \cmin \left(\frac{U}{\alpha'}-U+2\beta \right) \left[ \exp \left(\frac{w}{\alpha' \cmin} \right)-1 \right] + U \cmax \\
   &= \alpha' \cmin \left[ U - 2\beta + \left( \frac{U}{\alpha'}+2\beta-U \right) \exp \left( \frac{\hat{w}^{(j)}}{\alpha' \cmin} \right) \right] + U (\cmax - \cmin) \\ 
   &= \alpha' \cmin (\hat{\phi}(\hat{w}^{(j)}) - \beta) + U (\cmax - \cmin)
\end{align}
\endgroup

In what follows, we will calculate the optimal $\alpha'$ (as in~\cite[Theorem 3.3]{Lechowicz:24}) to define the threshold function based on the assumptions established for \tmpalgg (namely, that \tmpalgg assumes the job has length \cmin and that it must complete the job during the compulsory execution if it has length $>$ \cmin).  
It is known that worst-case instances for online search problems such as \OCSUmin occur when inputs arrive in decreasing order of cost (i.e., carbon intensity)~\cite{SunZeynali:20, SunLee:21, Lechowicz:24}.  We formalize these \textit{$x$-instances} below.

\begin{definition}[$x$-instance for \OCSUmin]
For sufficiently large $m, n \in \mathbb{Z}$, we let $\delta \coloneqq \nicefrac{(U-L)}{m}$.  Given $x \in [L,U]$, $\cali_x \in \Omega$ is an $x$-instance for \OCSUmin which consists of $m_x \coloneqq 2 \lceil \nicefrac{(x-L)}{\delta} \rceil + 1$ alternating blocks of cost functions.  For $i \in [m_x -2]$, the $i$th block contains $n$ linear cost functions with coefficient $U$ if $i$ is odd, or a single linear cost function with coefficient $U - \lceil \nicefrac{i}{2} \rceil \delta $ when $i$ is even.  The last $2$ blocks consist of $n$ linear cost function with coefficients $(x+ \epsilon)$, followed by $n$ cost functions with coefficients $U$.
\end{definition}

As $m \rightarrow \infty$, the alternating blocks of single ``good'' cost functions continuously decrease down to $x$, and each of these blocks is interrupted by a long block of worst-case $U$ functions.  Note that $\cali_U$ is a simple stream of $n$ cost functions, all with coefficient $U$, and that the last cost function for any $\cali_x$ are always $U$ (i.e., the marginal emission is maximized during the compulsory execution).  

\begin{lemma}
Any deterministic online algorithm $\texttt{ALG}$ for \OCSUmin which assumes the job has length $\cmin$ (and is thus forced to complete longer jobs during the compulsory execution) has a competitive ratio of at least $\alpha'$ (where $\alpha'$ is as defined in \eqref{alg2-alpha-prime} ).
\end{lemma}
\begin{proof}
On any $x$-instance $\cali_x$, we may fully describe the actions of any deterministic algorithm $\texttt{ALG}$ via a conversion function $h(x) : [L, U] \rightarrow [0, \cmin]$.  Note that this function is unidirectional (irrevocable), and non-increasing in $[L,U]$ such that $h(x-\delta) \geq h(x)$, since processing $\cali_{x-\delta}$ is equivalent to first processing $\cali_x$ (besides the final two blocks) and then processing blocks with marginal emissions of $x - \delta$ and $U$.
The total emission of $\texttt{ALG}$ 
described by the conversion function $h(x)$ on instance $\cali_x$ is expressed as follows: 
\begin{align}
    \texttt{ALG}(\mathcal{I}_x) &= h(U/\alpha')U/\alpha' - \int_{U/\alpha'}^{x}udh(u)+(c-h(x))U \leq h(U/\alpha')U/\alpha' - \int_{U/\alpha'}^{x}udh(u)+(\cmax-h(x))U \label{eq7}
\end{align}

We note that on an instance $\cali_x$, $\texttt{OPT}(I_x) \rightarrow \cmax x$ as $\epsilon \to 0$ and $n$ is sufficiently large. Letting \texttt{ALG} be $\alpha'$-competitive, we then have the following necessary condition on the conversion function when considering equation \eqref{eq7}:
\begin{align}
    h(U/\alpha')U/\alpha' - \int_{U/\alpha'}^{x}udh(u)+(\cmax-h(x))U \leq \alpha' \cmax x \label{eq:integ-by-parts}
\end{align}
By integral by parts, \eqref{eq:integ-by-parts} implies that $h(x)$ must satisfy:
\begin{align}
    h(x) \geq \frac{\cmax U-\alpha' \cmax x}{U-x-2\beta} + \frac{1}{U-x-2\beta} \int_{U/\alpha'}^{x} h(u)du
\end{align}

By Gronwall's Inequality~\cite[p. 356, Theorem 1]{Mitrinovic:91}, we have:
\begin{align}
    h(x) &\geq \frac{U \cmax-\alpha'\cmax x}{U-x-2\beta} +  \Bigg[ \frac{U \alpha' \cmax-U \cmax-2\beta \cmax}{u+2\beta-U} - \alpha' \cmax \ln(u+2\beta-U) \Bigg ]_{U/\alpha'}^{x}\\
    &h(x) \geq \alpha' \cmax \ln(U/\alpha'+2\beta-U) - \alpha' \cmax \ln(x+2\beta-U)
\end{align}

By the problem definition, the job with length \cmin should be completed upon observing the best carbon intensity $L$, i.e., $h(x) \leq h(L) \leq \cmin$, giving the following:
\begin{align}
    \cmin/\cmax &\geq \alpha' \ln(U/\alpha'+2\beta-U) - \alpha' \ln(L+2\beta-U).
\end{align}
The optimal $\alpha'$ is obtained when the above inequality is binding, which gives the following:
{\small\begin{align}
\alpha' &= \left[ \frac{\cmax}{\cmin}  W \left[ \frac{\cmin}{\cmax} \left( \frac{2\beta}{U} + \frac{L}{U} - 1 \right) \exp \left( \frac{\cmin}{\cmax} \left(\frac{2\beta}{U} - 1 \right) \right) \right] - \frac{2\beta}{U} + 1 \right]^{-1} \label{alg2-alpha-prime}.
\end{align}}
\end{proof}

\noindent\textbf{Step 2}: By scaling up $\hat{\phi(w)}, w \in [0, \cmin]$ to the factor of \cmax/\cmin, we will have $\phi_2(w), w \in [0, \cmax]$ in \eqref{eq:alg2-thresh}. The competitive ratio of $\ROROcmin$ that utilizes $\phi_2(w)$ can be derived using the following two cases:

\noindent $\bullet$ \textbf{If $\ROROcmin$ completes any amount of the job before the compulsory execution.}  In this case, the analysis from \autoref{subsec:alg1-comp-amalysis} exactly translates to the $\ROROcmin$ setting. Substituting $\phi_2$ for $\phi_1$ gives the following competitive bound for this case:
\begin{align}
   \frac{\ROROcmin(\cali)}{\OPT(\cali)} = \frac{U}{\alpha' L} + \frac{2 \beta}{L}.
\end{align}

\noindent $\bullet$ \textbf{If $\ROROcmin$ completes none of the job before the compulsory execution.} 
In this case, we know that $\OPT(\cali)$ is lower-bounded by $\phi_2(0) - \beta$, because if a cost function better than $\phi_2(0) - \beta$ arrived during the instance $\cali$, $\ROROcmin$ would have completed a non-zero amount of the job before the compulsory execution.  This gives the following competitive bound:
\begin{align}
   \frac{\ROROcmin(\cali)}{\OPT(\cali)} = \frac{U \cmin}{\left[ \phi_2(0) - \beta \right] \cmin} = \frac{U}{\nicefrac{U}{\alpha'}} = \alpha'.
\end{align}

Because $\alpha'$ approaches $\nicefrac{U}{L}$ as $\nicefrac{\cmax}{\cmin}$ grows, the competitive ratio in the second case is the worst-case bound, yielding the following competitive bound for $\ROROcmin$:
\begin{align}
   \alpha_2 = \alpha'.
\end{align}

Here we note that $\alpha_{2}$ does not contain an extra linear dependence on $\nicefrac{\cmax}{\cmin}$ which \textit{is} present in $\alpha_{2a}$, implying that $\alpha_{2} \leq \alpha_{2a}$. This is intuitive, since even if \tmpalgg completes some fraction of the job before the compulsory execution, it must complete $(\cmax - \cmin)$ of the job during the compulsory execution, whereas the scaled threshold in \ROROcmin allows it to be more flexible. In the rest of the paper, we use the design of \ROROcmin as our baseline based on its improved theoretical bounds and its superior performance in experiments.

\end{proof}

%% file: Z-aug.tex
In the following, we prove \autoref{th:aug-con-rob}, which states that the instantiation of \aug for \OCSUmin is $(\alpha + \gamma)$-consistent and $\left[ \left( 1-\frac{\gamma}{\alpha_1 - \texttt{sign}(\alpha_1 - \alpha_2)\epsilon - \alpha} \right) \alpha^{\max}_{\ROROpred} + \left( \frac{\gamma(\alpha_1 - \texttt{sign}(\alpha_1 - \alpha_2)\epsilon)}{\alpha_1 - \texttt{sign}(\alpha_1 - \alpha_2)\epsilon - \alpha} \right) \right]$-robust.
We note that consistency and robustness in learning-augmented algorithm design describe an algorithm's performance when the predictions are exactly accurate and entirely incorrect, respectively.  See \sref{Definition}{def:const-rob} for the formal definitions of both consistency and robustness.

\begin{proof}[Proof of \autoref{th:aug-con-rob}]
Initially, we start by noting that the online solutions given by \combalg and \aug are always feasible considering the constraint in \autoref{align:deadline}. Let $\cali \in \Omega$ be an arbitrary valid \OCSUmin sequence, and for a job with length $c$:
\begin{align*}
   \combalg(\cali):& \sum_{t=1}^{T}\Tilde{x}_t = \sum_{t=1}^{T} [kx_{1t} + (1-k)x_{2t}] > kc + (1-k)c > c  \\
   \aug(\cali):& \sum_{t=1}^{T}{x}_t = \sum_{t=1}^{T} [\lambda\hat{x}_t + (1-\lambda)\Tilde{x}_t] > \lambda c + (1-\lambda)c > c
\end{align*}

\begin{lemma}\label{lem:bounded-cost}
    The carbon emissions of \combalg is bounded by:
    \begin{align}
        \combalg(\cali) \leq k \ROROcmax(\cali) + (1-k) \ROROcmin(\cali)
    \end{align}
\end{lemma}
\begin{proof}
\begin{align*}
    \combalg(\cali) &= \sum_{t=1}^{T} g_t(\Tilde{x}_t) + \sum_{t=1}^{T+1} \beta |\Tilde{x}_t - \Tilde{x}_{t-1}| \\
    &= \sum_{t=1}^{T} g_t(k x_{1t} + (1 - k) x_{2t}) + \sum_{t=1}^{T+1} \beta |k x_{1t} + (1 - k) x_{2t} - k x_{2(t-1)} - (1 - k) x_{2(t-1)}| \\
    &\leq k \sum_{t=1}^{T} g_t(x_{1t}) + (1 - k) \sum_{t=1}^{T} g_t(x_{2t}) + \sum_{t=1}^{T+1} \beta |k x_{1t} - k x_{1(t-1)}| + \sum_{t=1}^{T+1} \beta |(1 - k) x_{2t}  - (1 - k) x_{2(t-1)}| \\
    &\leq k(\sum_{t=1}^{T} g_t(x_{1t}) + \beta |x_{1t} - x_{1(t-1)}|) + (1-k)(\sum_{t=1}^{T} g_t(x_{2t}) + \beta |x_{2t} - x_{2(t-1)}|)\\
    &\leq k \ROROcmax(I) + (1 - k) \ROROcmin(I)
\end{align*}    
\end{proof}
Since $\ROROcmax(\cali) \leq \alpha_1 \OPT(\cali)$ and $\ROROcmin(\cali) \leq \alpha_2 \OPT(\cali)$ by definition, we have:
\begin{align}
    \combalg(\cali) &\leq (k \alpha_1 + (1-k)\alpha_2) \OPT(\cali) \\
    \combalg(\cali) &\leq (k(\alpha_1 - \alpha_2) + \alpha_2) \OPT(\cali)
\end{align}
We  denote $\epsilon \in [0, |\alpha_1-\alpha_2|]$, and we set $k = 1- \frac{\epsilon}{|\alpha_1 - \alpha_2|}$; therefore, we have:
\begin{align}
    \combalg(\cali) \leq (\alpha_1 - \texttt{sign}(\alpha_1 - \alpha_2) \epsilon) \OPT(\cali) = \alpha_{\combalg} \OPT(\cali) \label{eq:alg-rob-alpha}    
\end{align}
where \texttt{sign}$(x)$ is the sign function.

By using the same proof in \sref{Lemma}{lem:bounded-cost}, we can show that:
\begin{align}
    \aug(\cali) \leq \lambda \ROROpred(\cali) + (1-\lambda) \combalg(\cali) \label{eq:alg-aug-up-bound}
\end{align}

\begin{lemma} \label{lem:const}
\aug is $(\alpha + \gamma)$-consistent with accurate predictions.
\end{lemma}
\begin{proof}
We assume $\ROROpred(\cali)$ has the perfect prediction of the job length ($\hat{c} = c$), and by leveraging the perfect prediction (excepting minor differences in the compulsory execution), we have that $\ROROpred(\cali) = \RORO(\cali)$.  Therefore, by definition, $\ROROpred(\cali) \leq \alpha \OPT(\cali)$

Considering \autoref{eq:alg-rob-alpha} and \autoref{eq:alg-aug-up-bound}, we have:
\begin{align}
    \aug(\cali) &\leq \lambda \ROROpred(\cali) + (1-\lambda) \combalg(\cali) \\
    \aug(\cali) &\leq \lambda \alpha \OPT(\cali) + (1-\lambda) \alpha_{\combalg} \OPT(\cali) \\
    \aug(\cali) &\leq (\lambda \alpha + (1-\lambda)\alpha_{\combalg}) \OPT(\cali) \\
    \aug(\cali) &\leq (\lambda \alpha + (1-\lambda)(\alpha_1 - \texttt{sign}(\alpha_1 - \alpha_2)\epsilon)) \OPT(\cali)
\end{align}

Since $\alpha \leq \alpha_1 - \texttt{sign}(\alpha_1 - \alpha_2)\epsilon \quad \forall \epsilon \in [0, |\alpha_1-\alpha_2|]$, and we have $\gamma \in [0, \alpha_1 - \texttt{sign}(\alpha_1 - \alpha_2)\epsilon - \alpha ]$; therefore we set $\lambda = 1 - \frac{\gamma}{\alpha_1 - \texttt{sign}(\alpha_1 - \alpha_2)\epsilon - \alpha}$, and we have:
\begin{align}
    \aug(\cali) &\leq (\alpha + \gamma) \OPT(\cali)
\end{align}
\end{proof}

\begin{lemma} \label{lem:rob}
\aug is $\bigg( \left( 1-\frac{\gamma}{\alpha_1 - \texttt{sign}(\alpha_1 - \alpha_2)\epsilon - \alpha} \right) \alpha^{\max}_{\ROROpred}$ $ + \left( \frac{\gamma(\alpha_1 - \texttt{sign}(\alpha_1 - \alpha_2)\epsilon)}{\alpha_1 - \texttt{sign}(\alpha_1 - \alpha_2)\epsilon - \alpha} \right) \bigg)$-robust for any prediction.
\end{lemma}
\begin{proof}
To calculate the competitive ratio of $\ROROpred(\cali)$ when the job length prediction error is maximized, we consider two cases:

\noindent\textbf{Case (i)} [$\hat{c}=\cmin, c=\cmax$]: Since \ROROpred assumes the job length is \cmin it will utilize the threshold function below:
\begin{align}
    \phi(w) = U-\beta + \left( \frac{U}{\alpha}-U+2\beta \right) \exp \left(\frac{w}{\cmin\alpha} \right) \label{eq:pred-rob-case1-thresh}
\end{align}

Let $w^{(j)}$ be the final progress before the compulsory execution at time $j \leq T$, and let $\cali \in \Omega$ be a \OCSUmin sequence that the minimum carbon intensity $L$ is revealed at time $m \leq j \leq T$. By disregarding extra switching emissions, $\OPT(\cali) \rightarrow \cmax L$. In \cite{Lechowicz:24}, by the definition of the threshold function for any job length, we use the threshold function for \cmin in \autoref{eq:pred-rob-case1-thresh}, when the minimum carbon intensity $L$ arrives at time $m$, the remained amount of the job until $w^{(j)}$ would be scheduled; hence, $w^m = w^{(j)}$ and $m=j$, and according to \textit{Lemma B.2} in \cite{Lechowicz:24}, $L = \phi(w^{(j)}) - \beta$ which means no other carbon intensities are accepted and the rest of the job ($\cmin-w^{(j)})$ should be done during compulsory execution.
By \textit{Lemma B.3} in \cite{Lechowicz:24} and observing that the rest of the job must be completed in the compulsory execution, $\ROROpred(\cali)$ is upper-bounded by:
\begin{align}
    \ROROpred(\cali) &\leq \left[ \int_{0}^{w^{(j)}} \phi(w) + \beta w^{(j)} + (\cmin - w^{(j)})U \right] + (\cmax - \cmin)U  \label{eq:pred-rob-case1-up-bound-eq1} \\ 
    &\leq [\alpha (\phi(w^{(j)}) - \beta] + (\cmax - \cmin)U \label{eq:pred-rob-case1-up-bound-eq2} \\
    & \leq \alpha L + (\cmax - \cmin)U \label{eq:pred-rob-case1-up-bound-eq3}
\end{align}

The term $(\cmax - \cmin)U$ in \eqref{eq:pred-rob-case1-up-bound-eq1} is the amount of remaining job that must be done during compulsory execution since $c = \cmax$.

Considering \eqref{eq:pred-rob-case1-up-bound-eq3} and the lower bound of $\OPT(\cali)$, we have:
\begin{align}
    \frac{\ROROpred(\cali)}{\OPT(\cali)} &\leq \frac{\alpha L + (\cmax - \cmin)U}{\cmax L} \leq \frac{\alpha}{\cmax} + \frac{\cmax - \cmin}{\cmax} \frac{U}{L} = \alpha^{1}_{\ROROpred} \label{eq:pred-rob-case1-alpha}
\end{align}

\noindent\textbf{Case (ii)} [$\hat{c}=\cmax, c=\cmin$]: In this case, 
\ROROpred uses the exact same threshold function as \ROROcmax. Thus, $\ROROcmax(\cali) = \ROROpred(\cali)$, and we inherit the following competitive bound:
\begin{align}
    \frac{\ROROpred(\cali)}{\OPT(\cali)} &\leq \frac{U}{\alpha L} + \frac{2 \beta}{L} = \alpha^{2}_{\ROROpred} \label{eq:pred-rob-case2-alpha}
\end{align}

Considering both \textbf{Case (i)} and \textbf{Case (ii)}, we let $\alpha^{\max}_{\ROROpred} = \max \{\alpha^{1}_{\ROROpred}, \alpha^{2}_{\ROROpred}\}$ to reflect the worst-case in either of these cases. By \autoref{eq:alg-aug-up-bound}, we have the following robustness bound:
\begin{align}
    \aug(\cali) &\leq \lambda \ROROpred(\cali)+ (1-\lambda) \combalg(\cali) \\
    \aug(\cali) &\leq \lambda \alpha^{\max}_{\ROROpred} \OPT(\cali) + (1-\lambda) \alpha_{\combalg} \OPT(\cali) \\
    \aug(\cali) &\leq (\lambda \alpha^{\max}_{\ROROpred} + (1-\lambda) (\alpha_1 - \epsilon)) \OPT(\cali) \\
    \aug(\cali) &\leq \Bigg( \left( 1-\frac{\gamma}{\alpha_1 - \texttt{sign}(\alpha_1 - \alpha_2)\epsilon - \alpha} \right) \alpha^{\max}_{\ROROpred} + \left( \frac{\gamma(\alpha_1 - \texttt{sign}(\alpha_1 - \alpha_2)\epsilon)}{\alpha_1 - \texttt{sign}(\alpha_1 - \alpha_2)\epsilon - \alpha} \right) \Bigg) \OPT(\cali)
\end{align}
\end{proof}
By combining the results of \sref{Lemma}{lem:const} and \sref{Lemma}{lem:rob}, the statement of \autoref{th:aug-con-rob} follows, and we conclude that \aug for \OCSUmin is $(\alpha + \gamma)$-consistent and $\left[ \left( 1-\frac{\gamma}{\alpha_1 - \texttt{sign}(\alpha_1 - \alpha_2)\epsilon - \alpha} \right) \alpha^{\max}_{\ROROpred} + \left( \frac{\gamma(\alpha_1 - \texttt{sign}(\alpha_1 - \alpha_2)\epsilon)}{\alpha_1 - \texttt{sign}(\alpha_1 - \alpha_2)\epsilon - \alpha} \right) \right]$-robust.
\end{proof}

%% file: paper.bbl

\begin{thebibliography}{59}


\ifx \showCODEN    \undefined \def \showCODEN     #1{\unskip}     \fi
\ifx \showDOI      \undefined \def \showDOI       #1{#1}\fi
\ifx \showISBNx    \undefined \def \showISBNx     #1{\unskip}     \fi
\ifx \showISBNxiii \undefined \def \showISBNxiii  #1{\unskip}     \fi
\ifx \showISSN     \undefined \def \showISSN      #1{\unskip}     \fi
\ifx \showLCCN     \undefined \def \showLCCN      #1{\unskip}     \fi
\ifx \shownote     \undefined \def \shownote      #1{#1}          \fi
\ifx \showarticletitle \undefined \def \showarticletitle #1{#1}   \fi
\ifx \showURL      \undefined \def \showURL       {\relax}        \fi
\providecommand\bibfield[2]{#2}
\providecommand\bibinfo[2]{#2}
\providecommand\natexlab[1]{#1}
\providecommand\showeprint[2][]{arXiv:#2}

\bibitem[Ambati et~al\mbox{.}(2021)]%
        {Ambati:2021}
\bibfield{author}{\bibinfo{person}{Pradeep Ambati}, \bibinfo{person}{Noman Bashir}, \bibinfo{person}{David Irwin}, {and} \bibinfo{person}{Prashant Shenoy}.} \bibinfo{year}{2021}\natexlab{}.
\newblock \showarticletitle{{Good Things Come to Those Who Wait: Optimizing Job Waiting in the Cloud}}. In \bibinfo{booktitle}{\emph{Proceedings of the ACM Symposium on Cloud Computing}} (Seattle, WA, USA) \emph{(\bibinfo{series}{SoCC '21})}. \bibinfo{publisher}{Association for Computing Machinery}, \bibinfo{address}{New York, NY, USA}, \bibinfo{pages}{229--242}.
\newblock
\showISBNx{9781450386388}
\urldef\tempurl%
\url{https://doi.org/10.1145/3472883.3487007}
\showDOI{\tempurl}


\bibitem[Arent et~al\mbox{.}(2022)]%
        {Arent:2022}
\bibfield{author}{\bibinfo{person}{Douglas~J Arent}, \bibinfo{person}{Peter Green}, \bibinfo{person}{Zia Abdullah}, \bibinfo{person}{Teresa Barnes}, \bibinfo{person}{Sage Bauer}, \bibinfo{person}{Andrey Bernstein}, \bibinfo{person}{Derek Berry}, \bibinfo{person}{Joe Berry}, \bibinfo{person}{Tony Burrell}, \bibinfo{person}{Birdie Carpenter}, {et~al\mbox{.}}} \bibinfo{year}{2022}\natexlab{}.
\newblock \showarticletitle{{Challenges and opportunities in decarbonizing the US energy system}}.
\newblock \bibinfo{journal}{\emph{Renewable and Sustainable Energy Reviews}}  \bibinfo{volume}{169} (\bibinfo{year}{2022}), \bibinfo{pages}{112939}.
\newblock


\bibitem[Arora et~al\mbox{.}(2023)]%
        {Arora:2023}
\bibfield{author}{\bibinfo{person}{Rohan Arora}, \bibinfo{person}{Umamaheswari Devi}, \bibinfo{person}{Tamar Eilam}, \bibinfo{person}{Aanchal Goyal}, \bibinfo{person}{Chandra Narayanaswami}, {and} \bibinfo{person}{Pritish Parida}.} \bibinfo{year}{2023}\natexlab{}.
\newblock \showarticletitle{{Towards Carbon Footprint Management in Hybrid Multicloud}}. In \bibinfo{booktitle}{\emph{Proceedings of the 2nd Workshop on Sustainable Computer Systems}}. \bibinfo{pages}{1--7}.
\newblock


\bibitem[Bashir et~al\mbox{.}(2023)]%
        {Bashir:2022}
\bibfield{author}{\bibinfo{person}{Noman Bashir}, \bibinfo{person}{David Irwin}, \bibinfo{person}{Prashant Shenoy}, {and} \bibinfo{person}{Abel Souza}.} \bibinfo{year}{2023}\natexlab{}.
\newblock \showarticletitle{{Sustainable Computing -- Without the Hot Air}}.
\newblock \bibinfo{journal}{\emph{ACM SIGENERGY Energy Informatics Review}} \bibinfo{volume}{3}, \bibinfo{number}{3} (\bibinfo{year}{2023}), \bibinfo{pages}{47--52}.
\newblock


\bibitem[B\"{o}ckenhauer et~al\mbox{.}(2014)]%
        {Bockenhauer:14}
\bibfield{author}{\bibinfo{person}{Hans-Joachim B\"{o}ckenhauer}, \bibinfo{person}{Dennis Komm}, \bibinfo{person}{Richard Kr\'{a}lovi\v{c}}, {and} \bibinfo{person}{Peter Rossmanith}.} \bibinfo{year}{2014}\natexlab{}.
\newblock \showarticletitle{{The online knapsack problem: Advice and randomization}}.
\newblock \bibinfo{journal}{\emph{Theoretical Computer Science}}  \bibinfo{volume}{527} (\bibinfo{year}{2014}), \bibinfo{pages}{61--72}.
\newblock
\showISSN{0304-3975}
\urldef\tempurl%
\url{https://doi.org/10.1016/j.tcs.2014.01.027}
\showDOI{\tempurl}


\bibitem[Borodin and El-Yaniv(2005)]%
        {Borodin:2005}
\bibfield{author}{\bibinfo{person}{Allan Borodin} {and} \bibinfo{person}{Ran El-Yaniv}.} \bibinfo{year}{2005}\natexlab{}.
\newblock \bibinfo{booktitle}{\emph{Online computation and competitive analysis}}.
\newblock \bibinfo{publisher}{cambridge university press}.
\newblock


\bibitem[Borodin et~al\mbox{.}(1992)]%
        {Borodin:92}
\bibfield{author}{\bibinfo{person}{Allan Borodin}, \bibinfo{person}{Nathan Linial}, {and} \bibinfo{person}{Michael~E. Saks}.} \bibinfo{year}{1992}\natexlab{}.
\newblock \showarticletitle{{An Optimal On-Line Algorithm for Metrical Task System}}.
\newblock \bibinfo{journal}{\emph{J. ACM}} \bibinfo{volume}{39}, \bibinfo{number}{4} (\bibinfo{date}{Oct} \bibinfo{year}{1992}), \bibinfo{pages}{745--763}.
\newblock
\showISSN{0004--5411}
\urldef\tempurl%
\url{https://doi.org/10.1145/146585.146588}
\showDOI{\tempurl}


\bibitem[Bostandoost et~al\mbox{.}(2023)]%
        {Bostandoost:23}
\bibfield{author}{\bibinfo{person}{Roozbeh Bostandoost}, \bibinfo{person}{Bo Sun}, \bibinfo{person}{Carlee Joe-Wong}, {and} \bibinfo{person}{Mohammad Hajiesmaili}.} \bibinfo{year}{2023}\natexlab{}.
\newblock \showarticletitle{{Near-Optimal Online Algorithms for Joint Pricing and Scheduling in EV Charging Networks}}. In \bibinfo{booktitle}{\emph{Proceedings of the 14th ACM International Conference on Future Energy Systems}} (Orlando, FL, USA) \emph{(\bibinfo{series}{e-Energy '23})}. \bibinfo{publisher}{Association for Computing Machinery}, \bibinfo{address}{New York, NY, USA}, \bibinfo{pages}{72--83}.
\newblock
\showISBNx{9798400700323}
\urldef\tempurl%
\url{https://doi.org/10.1145/3575813.3576878}
\showDOI{\tempurl}


\bibitem[Cole et~al\mbox{.}(2021)]%
        {Cole:2021}
\bibfield{author}{\bibinfo{person}{Wesley~J Cole}, \bibinfo{person}{Danny Greer}, \bibinfo{person}{Paul Denholm}, \bibinfo{person}{A~Will Frazier}, \bibinfo{person}{Scott Machen}, \bibinfo{person}{Trieu Mai}, \bibinfo{person}{Nina Vincent}, {and} \bibinfo{person}{Samuel~F Baldwin}.} \bibinfo{year}{2021}\natexlab{}.
\newblock \showarticletitle{{Quantifying the challenge of reaching a 100\% renewable energy power system for the United States}}.
\newblock \bibinfo{journal}{\emph{Joule}} \bibinfo{volume}{5}, \bibinfo{number}{7} (\bibinfo{year}{2021}), \bibinfo{pages}{1732--1748}.
\newblock


\bibitem[Commission(2024)]%
        {EU:2024:CBAM}
\bibfield{author}{\bibinfo{person}{European Commission}.} \bibinfo{year}{2024}\natexlab{}.
\newblock \bibinfo{title}{{Carbon Border Adjustment Mechanism}}.
\newblock \bibinfo{howpublished}{\url{https://taxation-customs.ec.europa.eu/carbon-border-adjustment-mechanism\_en}}.
\newblock
\newblock
\shownote{Accessed January}.


\bibitem[Corless et~al\mbox{.}(1996)]%
        {Corless:96LambertW}
\bibfield{author}{\bibinfo{person}{Robert~M Corless}, \bibinfo{person}{Gaston~H Gonnet}, \bibinfo{person}{David~EG Hare}, \bibinfo{person}{David~J Jeffrey}, {and} \bibinfo{person}{Donald~E Knuth}.} \bibinfo{year}{1996}\natexlab{}.
\newblock \showarticletitle{{On the Lambert W function}}.
\newblock \bibinfo{journal}{\emph{Advances in Computational mathematics}}  \bibinfo{volume}{5} (\bibinfo{year}{1996}), \bibinfo{pages}{329--359}.
\newblock


\bibitem[Cortez et~al\mbox{.}(2017)]%
        {Coretz:2017}
\bibfield{author}{\bibinfo{person}{Eli Cortez}, \bibinfo{person}{Anand Bonde}, \bibinfo{person}{Alexandre Muzio}, \bibinfo{person}{Mark Russinovich}, \bibinfo{person}{Marcus Fontoura}, {and} \bibinfo{person}{Ricardo Bianchini}.} \bibinfo{year}{2017}\natexlab{}.
\newblock \showarticletitle{{Resource Central: Understanding and Predicting Workloads for Improved Resource Management in Large Cloud Platforms}}. In \bibinfo{booktitle}{\emph{Proceedings of the 26th Symposium on Operating Systems Principles}} (Shanghai, China) \emph{(\bibinfo{series}{SOSP '17})}. \bibinfo{publisher}{Association for Computing Machinery}, \bibinfo{address}{New York, NY, USA}, \bibinfo{pages}{153--167}.
\newblock
\showISBNx{9781450350853}
\urldef\tempurl%
\url{https://doi.org/10.1145/3132747.3132772}
\showDOI{\tempurl}


\bibitem[Cygan et~al\mbox{.}(2016)]%
        {Cygan:16}
\bibfield{author}{\bibinfo{person}{Marek Cygan}, \bibinfo{person}{{\L}ukasz Je{\.z}}, {and} \bibinfo{person}{Ji{\v{r}}{\'\i} Sgall}.} \bibinfo{year}{2016}\natexlab{}.
\newblock \showarticletitle{{Online knapsack revisited}}.
\newblock \bibinfo{journal}{\emph{Theory of Computing Systems}}  \bibinfo{volume}{58} (\bibinfo{year}{2016}), \bibinfo{pages}{153--190}.
\newblock


\bibitem[Damaschke et~al\mbox{.}(2007)]%
        {Damaschke:07}
\bibfield{author}{\bibinfo{person}{Peter Damaschke}, \bibinfo{person}{Phuong~Hoai Ha}, {and} \bibinfo{person}{Philippas Tsigas}.} \bibinfo{year}{2007}\natexlab{}.
\newblock \showarticletitle{{Online Search with Time-Varying Price Bounds}}.
\newblock \bibinfo{journal}{\emph{Algorithmica}} \bibinfo{volume}{55}, \bibinfo{number}{4} (\bibinfo{date}{Dec.} \bibinfo{year}{2007}), \bibinfo{pages}{619--642}.
\newblock
\urldef\tempurl%
\url{https://doi.org/10.1007/s00453-007-9156-9}
\showDOI{\tempurl}


\bibitem[Dodge et~al\mbox{.}(2022)]%
        {Dodge:2022}
\bibfield{author}{\bibinfo{person}{Jesse Dodge}, \bibinfo{person}{Taylor Prewitt}, \bibinfo{person}{Remi Tachet~des Combes}, \bibinfo{person}{Erika Odmark}, \bibinfo{person}{Roy Schwartz}, \bibinfo{person}{Emma Strubell}, \bibinfo{person}{Alexandra~Sasha Luccioni}, \bibinfo{person}{Noah~A Smith}, \bibinfo{person}{Nicole DeCario}, {and} \bibinfo{person}{Will Buchanan}.} \bibinfo{year}{2022}\natexlab{}.
\newblock \showarticletitle{{Measuring the Carbon Intensity of AI in Cloud Instances}}. In \bibinfo{booktitle}{\emph{FAccT}}. \bibinfo{publisher}{ACM}, \bibinfo{address}{New York, NY, USA}, \bibinfo{pages}{1877--1894}.
\newblock


\bibitem[El-Yaniv et~al\mbox{.}(2001)]%
        {ElYaniv:01}
\bibfield{author}{\bibinfo{person}{Ran El-Yaniv}, \bibinfo{person}{Amos Fiat}, \bibinfo{person}{Richard~M. Karp}, {and} \bibinfo{person}{G. Turpin}.} \bibinfo{year}{2001}\natexlab{}.
\newblock \showarticletitle{{Optimal Search and One-Way Trading Online Algorithms}}.
\newblock \bibinfo{journal}{\emph{Algorithmica}} \bibinfo{volume}{30}, \bibinfo{number}{1} (\bibinfo{date}{May} \bibinfo{year}{2001}), \bibinfo{pages}{101--139}.
\newblock
\urldef\tempurl%
\url{https://doi.org/10.1007/s00453-001-0003-0}
\showDOI{\tempurl}


\bibitem[Gupta et~al\mbox{.}(2022a)]%
        {Gupta:2022:Act}
\bibfield{author}{\bibinfo{person}{Udit Gupta}, \bibinfo{person}{Mariam Elgamal}, \bibinfo{person}{Gage Hills}, \bibinfo{person}{Gu-Yeon Wei}, \bibinfo{person}{Hsien-Hsin~S Lee}, \bibinfo{person}{David Brooks}, {and} \bibinfo{person}{Carole-Jean Wu}.} \bibinfo{year}{2022}\natexlab{a}.
\newblock \showarticletitle{{ACT: Designing Sustainable Computer Systems with an Architectural Carbon Modeling Tool}}. In \bibinfo{booktitle}{\emph{Proceedings of the 49th Annual International Symposium on Computer Architecture}}. \bibinfo{publisher}{ACM}, \bibinfo{address}{New York, NY, USA}, \bibinfo{pages}{784--799}.
\newblock


\bibitem[Gupta et~al\mbox{.}(2022b)]%
        {Gupta:2022:Chasing}
\bibfield{author}{\bibinfo{person}{Udit Gupta}, \bibinfo{person}{Young~Geun Kim}, \bibinfo{person}{Sylvia Lee}, \bibinfo{person}{Jordan Tse}, \bibinfo{person}{Hsien-Hsin~S. Lee}, \bibinfo{person}{Gu-Yeon Wei}, \bibinfo{person}{David Brooks}, {and} \bibinfo{person}{Carole-Jean Wu}.} \bibinfo{year}{2022}\natexlab{b}.
\newblock \showarticletitle{{Chasing Carbon: The Elusive Environmental Footprint of Computing}}.
\newblock \bibinfo{journal}{\emph{IEEE Micro}} \bibinfo{volume}{42}, \bibinfo{number}{4} (\bibinfo{date}{jul} \bibinfo{year}{2022}), \bibinfo{pages}{37--47}.
\newblock
\showISSN{0272-1732}
\urldef\tempurl%
\url{https://doi.org/10.1109/MM.2022.3163226}
\showDOI{\tempurl}


\bibitem[Hanafy et~al\mbox{.}(2023a)]%
        {Hanafy:23:War}
\bibfield{author}{\bibinfo{person}{Walid~A. Hanafy}, \bibinfo{person}{Roozbeh Bostandoost}, \bibinfo{person}{Noman Bashir}, \bibinfo{person}{David Irwin}, \bibinfo{person}{Mohammad Hajiesmaili}, {and} \bibinfo{person}{Prashant Shenoy}.} \bibinfo{year}{2023}\natexlab{a}.
\newblock \showarticletitle{{The War of the Efficiencies: Understanding the Tension between Carbon and Energy Optimization}}. In \bibinfo{booktitle}{\emph{Proceedings of the 2nd Workshop on Sustainable Computer Systems}}. \bibinfo{publisher}{{ACM}}, \bibinfo{address}{New York, NY, USA}, \bibinfo{numpages}{7}~pages.
\newblock
\urldef\tempurl%
\url{https://doi.org/10.1145/3604930.3605709}
\showDOI{\tempurl}


\bibitem[Hanafy et~al\mbox{.}(2023b)]%
        {Hanafy:23:CarbonScaler}
\bibfield{author}{\bibinfo{person}{Walid~A. Hanafy}, \bibinfo{person}{Qianlin Liang}, \bibinfo{person}{Noman Bashir}, \bibinfo{person}{David Irwin}, {and} \bibinfo{person}{Prashant Shenoy}.} \bibinfo{year}{2023}\natexlab{b}.
\newblock \showarticletitle{{CarbonScaler: Leveraging Cloud Workload Elasticity for Optimizing Carbon-Efficiency}}.
\newblock \bibinfo{journal}{\emph{Proceedings of the ACM on Measurement and Analysis of Computing Systems}} \bibinfo{volume}{7}, \bibinfo{number}{3} (\bibinfo{date}{Dec} \bibinfo{year}{2023}), \bibinfo{numpages}{28}~pages.
\newblock
\showeprint[arxiv]{2302.08681}~[cs.DC]


\bibitem[Holland et~al\mbox{.}(2022)]%
        {Holland:2022}
\bibfield{author}{\bibinfo{person}{Stephen~P. Holland}, \bibinfo{person}{Matthew~J. Kotchen}, \bibinfo{person}{Erin~T. Mansur}, {and} \bibinfo{person}{Andrew~J. Yates}.} \bibinfo{year}{2022}\natexlab{}.
\newblock \showarticletitle{{Why Marginal CO2 Emissions Are Not Decreasing for U.S. Electricity: Estimates and Implications for Climate Policy}}.
\newblock \bibinfo{journal}{\emph{Proceedings of the National Academy of Sciences}} \bibinfo{volume}{119}, \bibinfo{number}{8} (\bibinfo{year}{2022}), \bibinfo{pages}{e2116632119}.
\newblock


\bibitem[Hoorfar and Hassani(2008)]%
        {HoorfarHassani:08}
\bibfield{author}{\bibinfo{person}{Abdolhossein Hoorfar} {and} \bibinfo{person}{Mehdi Hassani}.} \bibinfo{year}{2008}\natexlab{}.
\newblock \showarticletitle{{Inequalities on the Lambert W function and hyperpower function}}.
\newblock \bibinfo{journal}{\emph{Journal of Inequalities in Pure and Applied Mathematics}} \bibinfo{volume}{9}, \bibinfo{number}{51} (\bibinfo{date}{Jan.} \bibinfo{year}{2008}).
\newblock
Issue 2.


\bibitem[Kim et~al\mbox{.}(2023)]%
        {Kim:2023}
\bibfield{author}{\bibinfo{person}{Young~Geun Kim}, \bibinfo{person}{Udit Gupta}, \bibinfo{person}{Andrew McCrabb}, \bibinfo{person}{Yonglak Son}, \bibinfo{person}{Valeria Bertacco}, \bibinfo{person}{David Brooks}, {and} \bibinfo{person}{Carole-Jean Wu}.} \bibinfo{year}{2023}\natexlab{}.
\newblock \showarticletitle{{GreenScale: Carbon-Aware Systems for Edge Computing}}.
\newblock \bibinfo{journal}{\emph{arXiv preprint arXiv:2304.00404}} (\bibinfo{year}{2023}).
\newblock


\bibitem[Kuchnik et~al\mbox{.}(2019)]%
        {Kuchnik:2019}
\bibfield{author}{\bibinfo{person}{Michael Kuchnik}, \bibinfo{person}{J. Park}, \bibinfo{person}{C. Cranor}, \bibinfo{person}{Elisabeth Moore}, \bibinfo{person}{Nathan DeBardeleben}, {and} \bibinfo{person}{George Amvrosiadis}.} \bibinfo{year}{2019}\natexlab{}.
\newblock \bibinfo{booktitle}{\emph{{This is Why ML-driven Cluster Scheduling Remains Widely Impractical}}}.
\newblock \bibinfo{type}{{T}echnical {R}eport} CMU-PDL-19-103.
\newblock


\bibitem[Lechowicz et~al\mbox{.}(2024)]%
        {Lechowicz:24}
\bibfield{author}{\bibinfo{person}{Adam Lechowicz}, \bibinfo{person}{Nicolas Christianson}, \bibinfo{person}{Bo Sun}, \bibinfo{person}{Noman Bashir}, \bibinfo{person}{Mohammad Hajiesmaili}, \bibinfo{person}{Adam Wierman}, {and} \bibinfo{person}{Prashant Shenoy}.} \bibinfo{year}{2024}\natexlab{}.
\newblock \showarticletitle{{Online Conversion with Switching Costs: Robust and Learning-augmented Algorithms}}. In \bibinfo{booktitle}{\emph{Proceedings of the 2024 SIGMETRICS/Performance Joint International Conference on Measurement and Modeling of Computer Systems}} (Venice, Italy) \emph{(\bibinfo{series}{SIGMETRICS / Performance '24})}. \bibinfo{publisher}{Association for Computing Machinery}, \bibinfo{address}{New York, NY, USA}.
\newblock


\bibitem[Lechowicz et~al\mbox{.}(2023)]%
        {Lechowicz:23}
\bibfield{author}{\bibinfo{person}{Adam Lechowicz}, \bibinfo{person}{Nicolas Christianson}, \bibinfo{person}{Jinhang Zuo}, \bibinfo{person}{Noman Bashir}, \bibinfo{person}{Mohammad Hajiesmaili}, \bibinfo{person}{Adam Wierman}, {and} \bibinfo{person}{Prashant Shenoy}.} \bibinfo{year}{2023}\natexlab{}.
\newblock \showarticletitle{{The Online Pause and Resume Problem: Optimal Algorithms and An Application to Carbon-Aware Load Shifting}}.
\newblock \bibinfo{journal}{\emph{Proceedings of the ACM on Measurement and Analysis of Computing Systems}} \bibinfo{volume}{7}, \bibinfo{number}{3}, Article \bibinfo{articleno}{53} (\bibinfo{date}{Dec} \bibinfo{year}{2023}), \bibinfo{numpages}{36}~pages.
\newblock
\showeprint[arxiv]{2303.17551}~[cs.DS]


\bibitem[Lee et~al\mbox{.}(2024)]%
        {Lee:24}
\bibfield{author}{\bibinfo{person}{Russell Lee}, \bibinfo{person}{Bo Sun}, \bibinfo{person}{Mohammad Hajiesmaili}, {and} \bibinfo{person}{John C.~S. Lui}.} \bibinfo{year}{2024}\natexlab{}.
\newblock \showarticletitle{{Online Search with Predictions: Pareto-optimal Algorithm and its Applications in Energy Markets}}. In \bibinfo{booktitle}{\emph{Proceedings of the 15th ACM International Conference on Future Energy Systems}} (Singapore, Singapore) \emph{(\bibinfo{series}{e-Energy '24})}. \bibinfo{publisher}{Association for Computing Machinery}, \bibinfo{address}{New York, NY, USA}.
\newblock


\bibitem[Lorenz et~al\mbox{.}(2008)]%
        {Lorenz:08}
\bibfield{author}{\bibinfo{person}{Julian Lorenz}, \bibinfo{person}{Konstantinos Panagiotou}, {and} \bibinfo{person}{Angelika Steger}.} \bibinfo{year}{2008}\natexlab{}.
\newblock \showarticletitle{{Optimal Algorithms for k-Search with Application in~Option Pricing}}.
\newblock \bibinfo{journal}{\emph{Algorithmica}} \bibinfo{volume}{55}, \bibinfo{number}{2} (\bibinfo{date}{Aug.} \bibinfo{year}{2008}), \bibinfo{pages}{311--328}.
\newblock
\urldef\tempurl%
\url{https://doi.org/10.1007/s00453-008-9217-8}
\showDOI{\tempurl}


\bibitem[Luccioni et~al\mbox{.}(2023)]%
        {Luccioni:2023}
\bibfield{author}{\bibinfo{person}{Alexandra~Sasha Luccioni}, \bibinfo{person}{Yacine Jernite}, {and} \bibinfo{person}{Emma Strubell}.} \bibinfo{year}{2023}\natexlab{}.
\newblock \bibinfo{title}{{Power Hungry Processing: Watts Driving the Cost of AI Deployment?}}
\newblock
\newblock
\showeprint[arxiv]{2311.16863}~[cs.LG]


\bibitem[Lykouris and Vassilvtiskii(2018)]%
        {Lykouris:18}
\bibfield{author}{\bibinfo{person}{Thodoris Lykouris} {and} \bibinfo{person}{Sergei Vassilvtiskii}.} \bibinfo{year}{2018}\natexlab{}.
\newblock \showarticletitle{{Competitive Caching with Machine Learned Advice}}. In \bibinfo{booktitle}{\emph{Proceedings of the 35th International Conference on Machine Learning}} \emph{(\bibinfo{series}{Proceedings of Machine Learning Research}, Vol.~\bibinfo{volume}{80})}, \bibfield{editor}{\bibinfo{person}{Jennifer Dy} {and} \bibinfo{person}{Andreas Krause}} (Eds.). \bibinfo{publisher}{PMLR}, \bibinfo{pages}{3296--3305}.
\newblock
\urldef\tempurl%
\url{https://proceedings.mlr.press/v80/lykouris18a.html}
\showURL{%
\tempurl}


\bibitem[Maji et~al\mbox{.}(2022)]%
        {Maji:22:CC}
\bibfield{author}{\bibinfo{person}{Diptyaroop Maji}, \bibinfo{person}{Prashant Shenoy}, {and} \bibinfo{person}{Ramesh~K. Sitaraman}.} \bibinfo{year}{2022}\natexlab{}.
\newblock \showarticletitle{{CarbonCast: Multi-Day Forecasting of Grid Carbon Intensity}}. In \bibinfo{booktitle}{\emph{Proceedings of the 9th ACM International Conference on Systems for Energy-Efficient Buildings, Cities, and Transportation}} (Boston, Massachusetts) \emph{(\bibinfo{series}{BuildSys '22})}. \bibinfo{publisher}{Association for Computing Machinery}, \bibinfo{address}{New York, NY, USA}, \bibinfo{pages}{198--207}.
\newblock
\showISBNx{9781450398909}
\urldef\tempurl%
\url{https://doi.org/10.1145/3563357.3564079}
\showDOI{\tempurl}


\bibitem[{M}aps(2023)]%
        {electricity-map}
\bibfield{author}{\bibinfo{person}{Electricity {M}aps}.} \bibinfo{year}{2023}\natexlab{}.
\newblock \bibinfo{title}{{Electricity Map}}.
\newblock \bibinfo{howpublished}{\url{https://www.electricitymap.org/map}}.
\newblock


\bibitem[Marchetti-Spaccamela and Vercellis(1995)]%
        {Marchetti:95}
\bibfield{author}{\bibinfo{person}{A. Marchetti-Spaccamela} {and} \bibinfo{person}{C. Vercellis}.} \bibinfo{year}{1995}\natexlab{}.
\newblock \showarticletitle{{Stochastic on-line knapsack problems}}.
\newblock \bibinfo{journal}{\emph{Mathematical Programming}} \bibinfo{volume}{68}, \bibinfo{number}{1-3} (\bibinfo{date}{Jan.} \bibinfo{year}{1995}), \bibinfo{pages}{73--104}.
\newblock
\urldef\tempurl%
\url{https://doi.org/10.1007/bf01585758}
\showDOI{\tempurl}


\bibitem[Mitrinovic et~al\mbox{.}(1991)]%
        {Mitrinovic:91}
\bibfield{author}{\bibinfo{person}{Dragoslav~S. Mitrinovic}, \bibinfo{person}{Josip~E. Pe{\v{c}}ari{\'c}}, {and} \bibinfo{person}{A.~M. Fink}.} \bibinfo{year}{1991}\natexlab{}.
\newblock \bibinfo{booktitle}{\emph{{Inequalities Involving Functions and Their Integrals and Derivatives}}}. Vol.~\bibinfo{volume}{53}.
\newblock \bibinfo{publisher}{Springer Science \& Business Media}.
\newblock


\bibitem[Mohr et~al\mbox{.}(2014)]%
        {mohr2014online}
\bibfield{author}{\bibinfo{person}{Esther Mohr}, \bibinfo{person}{Iftikhar Ahmad}, {and} \bibinfo{person}{G{\"u}nter Schmidt}.} \bibinfo{year}{2014}\natexlab{}.
\newblock \showarticletitle{{Online algorithms for conversion problems: a survey}}.
\newblock \bibinfo{journal}{\emph{Surveys in Operations Research and Management Science}} \bibinfo{volume}{19}, \bibinfo{number}{2} (\bibinfo{year}{2014}), \bibinfo{pages}{87--104}.
\newblock


\bibitem[Monserrate(2022)]%
        {Monserrate:2022}
\bibfield{author}{\bibinfo{person}{Steven~Gonzalez Monserrate}.} \bibinfo{year}{2022}\natexlab{}.
\newblock \bibinfo{title}{{The Staggering Ecological Impacts of Computation and the Cloud}}.
\newblock \bibinfo{howpublished}{\url{https://thereader.mitpress.mit.edu/the-staggering-ecological-impacts-of-computation-and-the-cloud/}}.
\newblock


\bibitem[(NSF)(2022)]%
        {NSF:2022:DESC}
\bibfield{author}{\bibinfo{person}{National Science~Foundation (NSF)}.} \bibinfo{year}{2022}\natexlab{}.
\newblock \bibinfo{title}{{Design for Environmental Sustainability in Computing (DESC)}}.
\newblock \bibinfo{howpublished}{\url{https://www.nsf.gov/pubs/2023/nsf23532/nsf23532.htm}}.
\newblock


\bibitem[Purohit et~al\mbox{.}(2018)]%
        {Purohit:18}
\bibfield{author}{\bibinfo{person}{Manish Purohit}, \bibinfo{person}{Zoya Svitkina}, {and} \bibinfo{person}{Ravi Kumar}.} \bibinfo{year}{2018}\natexlab{}.
\newblock \showarticletitle{{Improving Online Algorithms via ML Predictions}}. In \bibinfo{booktitle}{\emph{Advances in Neural Information Processing Systems}}, \bibfield{editor}{\bibinfo{person}{S.~Bengio}, \bibinfo{person}{H.~Wallach}, \bibinfo{person}{H.~Larochelle}, \bibinfo{person}{K.~Grauman}, \bibinfo{person}{N.~Cesa-Bianchi}, {and} \bibinfo{person}{R.~Garnett}} (Eds.), Vol.~\bibinfo{volume}{31}. \bibinfo{publisher}{Curran Associates, Inc.}
\newblock


\bibitem[Rzadca et~al\mbox{.}(2020)]%
        {Rzadca:2020}
\bibfield{author}{\bibinfo{person}{Krzysztof Rzadca}, \bibinfo{person}{Pawel Findeisen}, \bibinfo{person}{Jacek Swiderski}, \bibinfo{person}{Przemyslaw Zych}, \bibinfo{person}{Przemyslaw Broniek}, \bibinfo{person}{Jarek Kusmierek}, \bibinfo{person}{Pawel Nowak}, \bibinfo{person}{Beata Strack}, \bibinfo{person}{Piotr Witusowski}, \bibinfo{person}{Steven Hand}, {and} \bibinfo{person}{John Wilkes}.} \bibinfo{year}{2020}\natexlab{}.
\newblock \showarticletitle{{Autopilot: workload autoscaling at Google}}. In \bibinfo{booktitle}{\emph{Proceedings of the Fifteenth European Conference on Computer Systems}} (Heraklion, Greece) \emph{(\bibinfo{series}{EuroSys '20})}. \bibinfo{publisher}{Association for Computing Machinery}, \bibinfo{address}{New York, NY, USA}, Article \bibinfo{articleno}{16}, \bibinfo{numpages}{16}~pages.
\newblock
\showISBNx{9781450368827}
\urldef\tempurl%
\url{https://doi.org/10.1145/3342195.3387524}
\showDOI{\tempurl}


\bibitem[Sharma et~al\mbox{.}(2016)]%
        {Sharma:2016:Flint}
\bibfield{author}{\bibinfo{person}{Prateek Sharma}, \bibinfo{person}{Tian Guo}, \bibinfo{person}{Xin He}, \bibinfo{person}{David Irwin}, {and} \bibinfo{person}{Prashant Shenoy}.} \bibinfo{year}{2016}\natexlab{}.
\newblock \showarticletitle{{Flint: batch-interactive data-intensive processing on transient servers}}. In \bibinfo{booktitle}{\emph{Proceedings of the Eleventh European Conference on Computer Systems}} (London, United Kingdom) \emph{(\bibinfo{series}{EuroSys '16})}. \bibinfo{publisher}{Association for Computing Machinery}, \bibinfo{address}{New York, NY, USA}, Article \bibinfo{articleno}{6}, \bibinfo{numpages}{15}~pages.
\newblock
\showISBNx{9781450342407}
\urldef\tempurl%
\url{https://doi.org/10.1145/2901318.2901319}
\showDOI{\tempurl}


\bibitem[Souza et~al\mbox{.}(2023a)]%
        {Souza:23}
\bibfield{author}{\bibinfo{person}{Abel Souza}, \bibinfo{person}{Noman Bashir}, \bibinfo{person}{Jorge Murillo}, \bibinfo{person}{Walid Hanafy}, \bibinfo{person}{Qianlin Liang}, \bibinfo{person}{David Irwin}, {and} \bibinfo{person}{Prashant Shenoy}.} \bibinfo{year}{2023}\natexlab{a}.
\newblock \showarticletitle{{Ecovisor: A Virtual Energy System for Carbon-Efficient Applications}}. In \bibinfo{booktitle}{\emph{Proceedings of the 28th ACM International Conference on Architectural Support for Programming Languages and Operating Systems, Volume 2}} (Vancouver, BC, Canada) \emph{(\bibinfo{series}{ASPLOS 2023})}. \bibinfo{publisher}{Association for Computing Machinery}, \bibinfo{address}{New York, NY, USA}, \bibinfo{pages}{252--265}.
\newblock
\showISBNx{9781450399166}
\urldef\tempurl%
\url{https://doi.org/10.1145/3575693.3575709}
\showDOI{\tempurl}


\bibitem[Souza et~al\mbox{.}(2023b)]%
        {Souza:2023:Casper}
\bibfield{author}{\bibinfo{person}{Abel Souza}, \bibinfo{person}{Shruti Jasoria}, \bibinfo{person}{Basundhara Chakrabarty}, \bibinfo{person}{Alexander Bridgwater}, \bibinfo{person}{Axel Lundberg}, \bibinfo{person}{Filip Skogh}, \bibinfo{person}{Ahmed Ali-Eldin}, \bibinfo{person}{David Irwin}, {and} \bibinfo{person}{Prashant Shenoy}.} \bibinfo{year}{2023}\natexlab{b}.
\newblock \showarticletitle{{CASPER: Carbon-Aware Scheduling and Provisioning for Distributed Web Services}}. In \bibinfo{booktitle}{\emph{Proceedings of the 14th International Green and Sustainable Computing Conference (IGSC), Toronto, ON, Canada}}. \bibinfo{publisher}{ACM}.
\newblock


\bibitem[Strubell et~al\mbox{.}(2020)]%
        {Strubell:2020}
\bibfield{author}{\bibinfo{person}{Emma Strubell}, \bibinfo{person}{Ananya Ganesh}, {and} \bibinfo{person}{Andrew Mc{C}allum}.} \bibinfo{year}{2020}\natexlab{}.
\newblock \showarticletitle{{Energy and Policy Considerations for Modern Deep Learning Research}}. In \bibinfo{booktitle}{\emph{Proceedings of the AAAI Conference on Artificial Intelligence}}, Vol.~\bibinfo{volume}{34}. \bibinfo{publisher}{AAAI}, \bibinfo{address}{New York, NY}, \bibinfo{pages}{13693--13696}.
\newblock


\bibitem[Sukprasert et~al\mbox{.}(2023)]%
        {sukprasert2023quantifying}
\bibfield{author}{\bibinfo{person}{Thanathorn Sukprasert}, \bibinfo{person}{Abel Souza}, \bibinfo{person}{Noman Bashir}, \bibinfo{person}{David Irwin}, {and} \bibinfo{person}{Prashant Shenoy}.} \bibinfo{year}{2023}\natexlab{}.
\newblock \bibinfo{title}{{Quantifying the Benefits of Carbon-Aware Temporal and Spatial Workload Shifting in the Cloud}}.
\newblock
\newblock
\showeprint[arxiv]{2306.06502}~[cs.DC]


\bibitem[Sun et~al\mbox{.}(2021a)]%
        {SunLee:21}
\bibfield{author}{\bibinfo{person}{Bo Sun}, \bibinfo{person}{Russell Lee}, \bibinfo{person}{Mohammad Hajiesmaili}, \bibinfo{person}{Adam Wierman}, {and} \bibinfo{person}{Danny Tsang}.} \bibinfo{year}{2021}\natexlab{a}.
\newblock \showarticletitle{{Pareto-Optimal Learning-Augmented Algorithms for Online Conversion Problems}}. In \bibinfo{booktitle}{\emph{Advances in Neural Information Processing Systems}}, \bibfield{editor}{\bibinfo{person}{M.~Ranzato}, \bibinfo{person}{A.~Beygelzimer}, \bibinfo{person}{Y.~Dauphin}, \bibinfo{person}{P.S. Liang}, {and} \bibinfo{person}{J.~Wortman Vaughan}} (Eds.), Vol.~\bibinfo{volume}{34}. \bibinfo{publisher}{Curran Associates, Inc.}, \bibinfo{pages}{10339--10350}.
\newblock


\bibitem[Sun et~al\mbox{.}(2022)]%
        {sun2022online}
\bibfield{author}{\bibinfo{person}{Bo Sun}, \bibinfo{person}{Lin Yang}, \bibinfo{person}{Mohammad Hajiesmaili}, \bibinfo{person}{Adam Wierman}, \bibinfo{person}{John~CS Lui}, \bibinfo{person}{Don Towsley}, {and} \bibinfo{person}{Danny~HK Tsang}.} \bibinfo{year}{2022}\natexlab{}.
\newblock \showarticletitle{{The Online Knapsack Problem with Departures}}.
\newblock \bibinfo{journal}{\emph{Proceedings of the ACM on Measurement and Analysis of Computing Systems}} \bibinfo{volume}{6}, \bibinfo{number}{3} (\bibinfo{year}{2022}), \bibinfo{pages}{1--32}.
\newblock


\bibitem[Sun et~al\mbox{.}(2021b)]%
        {SunZeynali:20}
\bibfield{author}{\bibinfo{person}{Bo Sun}, \bibinfo{person}{Ali Zeynali}, \bibinfo{person}{Tongxin Li}, \bibinfo{person}{Mohammad Hajiesmaili}, \bibinfo{person}{Adam Wierman}, {and} \bibinfo{person}{Danny~H.K. Tsang}.} \bibinfo{year}{2021}\natexlab{b}.
\newblock \showarticletitle{{Competitive Algorithms for the Online Multiple Knapsack Problem with Application to Electric Vehicle Charging}}.
\newblock \bibinfo{journal}{\emph{Proceedings of the ACM on Measurement and Analysis of Computing Systems}} \bibinfo{volume}{4}, \bibinfo{number}{3}, Article \bibinfo{articleno}{51} (\bibinfo{date}{June} \bibinfo{year}{2021}), \bibinfo{numpages}{32}~pages.
\newblock
\urldef\tempurl%
\url{https://doi.org/10.1145/3428336}
\showDOI{\tempurl}


\bibitem[Switzer et~al\mbox{.}(2023)]%
        {Switzer:2023}
\bibfield{author}{\bibinfo{person}{Jennifer Switzer}, \bibinfo{person}{Gabriel Marcano}, \bibinfo{person}{Ryan Kastner}, {and} \bibinfo{person}{Pat Pannuto}.} \bibinfo{year}{2023}\natexlab{}.
\newblock \showarticletitle{{Junkyard Computing: Repurposing Discarded Smartphones to Minimize Carbon}}. In \bibinfo{booktitle}{\emph{Proceedings of the 28th ACM International Conference on Architectural Support for Programming Languages and Operating Systems, Volume 2}}. \bibinfo{publisher}{ACM}, \bibinfo{address}{New York, NY, USA}, \bibinfo{pages}{400--412}.
\newblock


\bibitem[Tannu and Nair(2023)]%
        {Tannu:2023}
\bibfield{author}{\bibinfo{person}{Swamit Tannu} {and} \bibinfo{person}{Prashant~J Nair}.} \bibinfo{year}{2023}\natexlab{}.
\newblock \showarticletitle{{The Dirty Secret of SSDs: Embodied Carbon}}.
\newblock \bibinfo{journal}{\emph{ACM SIGENERGY Energy Informatics Review}} \bibinfo{volume}{3}, \bibinfo{number}{3} (\bibinfo{year}{2023}), \bibinfo{pages}{4--9}.
\newblock


\bibitem[Thiede et~al\mbox{.}(2023)]%
        {Thiede:2023:Containers}
\bibfield{author}{\bibinfo{person}{John Thiede}, \bibinfo{person}{Noman Bashir}, \bibinfo{person}{David Irwin}, {and} \bibinfo{person}{Prashant Shenoy}.} \bibinfo{year}{2023}\natexlab{}.
\newblock \showarticletitle{{Carbon Containers: A System-Level Facility for Managing Application-Level Carbon Emissions}}. In \bibinfo{booktitle}{\emph{Proceedings of the 2023 ACM Symposium on Cloud Computing}} \emph{(\bibinfo{series}{SoCC '23})}. \bibinfo{publisher}{Association for Computing Machinery}, \bibinfo{address}{New York, NY, USA}, \bibinfo{pages}{17--31}.
\newblock
\showISBNx{9798400703874}
\urldef\tempurl%
\url{https://doi.org/10.1145/3620678.3624644}
\showDOI{\tempurl}


\bibitem[Verma et~al\mbox{.}(2015)]%
        {Verma:2015}
\bibfield{author}{\bibinfo{person}{Abhishek Verma}, \bibinfo{person}{Luis Pedrosa}, \bibinfo{person}{Madhukar Korupolu}, \bibinfo{person}{David Oppenheimer}, \bibinfo{person}{Eric Tune}, {and} \bibinfo{person}{John Wilkes}.} \bibinfo{year}{2015}\natexlab{}.
\newblock \showarticletitle{{Large-scale cluster management at Google with Borg}}. In \bibinfo{booktitle}{\emph{Proceedings of the Tenth European Conference on Computer Systems}} (Bordeaux, France) \emph{(\bibinfo{series}{EuroSys '15})}. \bibinfo{publisher}{Association for Computing Machinery}, \bibinfo{address}{New York, NY, USA}, Article \bibinfo{articleno}{18}, \bibinfo{numpages}{17}~pages.
\newblock
\showISBNx{9781450332385}
\urldef\tempurl%
\url{https://doi.org/10.1145/2741948.2741964}
\showDOI{\tempurl}


\bibitem[Virtanen et~al\mbox{.}(2020)]%
        {SciPy}
\bibfield{author}{\bibinfo{person}{Pauli Virtanen}, \bibinfo{person}{Ralf Gommers}, \bibinfo{person}{Travis~E. Oliphant}, \bibinfo{person}{Matt Haberland}, \bibinfo{person}{Tyler Reddy}, \bibinfo{person}{David Cournapeau}, \bibinfo{person}{Evgeni Burovski}, \bibinfo{person}{Pearu Peterson}, \bibinfo{person}{Warren Weckesser}, \bibinfo{person}{Jonathan Bright}, \bibinfo{person}{St{\'e}fan~J. {van der Walt}}, \bibinfo{person}{Matthew Brett}, \bibinfo{person}{Joshua Wilson}, \bibinfo{person}{K.~Jarrod Millman}, \bibinfo{person}{Nikolay Mayorov}, \bibinfo{person}{Andrew R.~J. Nelson}, \bibinfo{person}{Eric Jones}, \bibinfo{person}{Robert Kern}, \bibinfo{person}{Eric Larson}, \bibinfo{person}{C~J Carey}, \bibinfo{person}{{\.I}lhan Polat}, \bibinfo{person}{Yu Feng}, \bibinfo{person}{Eric~W. Moore}, \bibinfo{person}{Jake {VanderPlas}}, \bibinfo{person}{Denis Laxalde}, \bibinfo{person}{Josef Perktold}, \bibinfo{person}{Robert Cimrman}, \bibinfo{person}{Ian Henriksen}, \bibinfo{person}{E.~A. Quintero},
  \bibinfo{person}{Charles~R. Harris}, \bibinfo{person}{Anne~M. Archibald}, \bibinfo{person}{Ant{\^o}nio~H. Ribeiro}, \bibinfo{person}{Fabian Pedregosa}, \bibinfo{person}{Paul {van Mulbregt}}, {and} \bibinfo{person}{{SciPy 1.0 Contributors}}.} \bibinfo{year}{2020}\natexlab{}.
\newblock \showarticletitle{{SciPy 1.0: Fundamental Algorithms for Scientific Computing in Python}}.
\newblock \bibinfo{journal}{\emph{Nature Methods}}  \bibinfo{volume}{17} (\bibinfo{year}{2020}), \bibinfo{pages}{261--272}.
\newblock
\urldef\tempurl%
\url{https://doi.org/10.1038/s41592-019-0686-2}
\showDOI{\tempurl}


\bibitem[Wiesner et~al\mbox{.}(2021)]%
        {Wiesner:21}
\bibfield{author}{\bibinfo{person}{Philipp Wiesner}, \bibinfo{person}{Ilja Behnke}, \bibinfo{person}{Dominik Scheinert}, \bibinfo{person}{Kordian Gontarska}, {and} \bibinfo{person}{Lauritz Thamsen}.} \bibinfo{year}{2021}\natexlab{}.
\newblock \showarticletitle{{Let's Wait AWhile: How Temporal Workload Shifting Can Reduce Carbon Emissions in the Cloud}}. In \bibinfo{booktitle}{\emph{Proceedings of the 22nd International Middleware Conference}}. \bibinfo{publisher}{Association for Computing Machinery}, \bibinfo{address}{New York, NY, USA}, \bibinfo{pages}{260--272}.
\newblock
\urldef\tempurl%
\url{https://doi.org/10.1145/3464298.3493399}
\showDOI{\tempurl}


\bibitem[Yang et~al\mbox{.}(2021)]%
        {Yang2021Competitive}
\bibfield{author}{\bibinfo{person}{Lin Yang}, \bibinfo{person}{Ali Zeynali}, \bibinfo{person}{Mohammad~H. Hajiesmaili}, \bibinfo{person}{Ramesh~K. Sitaraman}, {and} \bibinfo{person}{Don Towsley}.} \bibinfo{year}{2021}\natexlab{}.
\newblock \showarticletitle{{Competitive Algorithms for Online Multidimensional Knapsack Problems}}.
\newblock \bibinfo{journal}{\emph{Proceedings of the ACM on Measurement and Analysis of Computing Systems}} \bibinfo{volume}{5}, \bibinfo{number}{3}, Article \bibinfo{articleno}{30} (\bibinfo{date}{Dec} \bibinfo{year}{2021}), \bibinfo{numpages}{30}~pages.
\newblock


\bibitem[Zeynali et~al\mbox{.}(2021)]%
        {Zeynali:21}
\bibfield{author}{\bibinfo{person}{Ali Zeynali}, \bibinfo{person}{Bo Sun}, \bibinfo{person}{Mohammad Hajiesmaili}, {and} \bibinfo{person}{Adam Wierman}.} \bibinfo{year}{2021}\natexlab{}.
\newblock \showarticletitle{{Data-driven Competitive Algorithms for Online Knapsack and Set Cover}}.
\newblock \bibinfo{journal}{\emph{Proceedings of the AAAI Conference on Artificial Intelligence}} \bibinfo{volume}{35}, \bibinfo{number}{12} (\bibinfo{date}{May} \bibinfo{year}{2021}), \bibinfo{pages}{10833--10841}.
\newblock
\urldef\tempurl%
\url{https://doi.org/10.1609/aaai.v35i12.17294}
\showDOI{\tempurl}


\bibitem[Zhang et~al\mbox{.}(2017)]%
        {zhang2017optimal}
\bibfield{author}{\bibinfo{person}{ZiJun Zhang}, \bibinfo{person}{Zongpeng Li}, {and} \bibinfo{person}{Chuan Wu}.} \bibinfo{year}{2017}\natexlab{}.
\newblock \showarticletitle{{Optimal posted prices for online cloud resource allocation}}.
\newblock \bibinfo{journal}{\emph{Proceedings of the ACM on Measurement and Analysis of Computing Systems}} \bibinfo{volume}{1}, \bibinfo{number}{1} (\bibinfo{year}{2017}), \bibinfo{pages}{1--26}.
\newblock


\bibitem[Zheng et~al\mbox{.}(2020)]%
        {Zheng:2020:Curtailment}
\bibfield{author}{\bibinfo{person}{Jiajia Zheng}, \bibinfo{person}{Andrew~A. Chien}, {and} \bibinfo{person}{Sangwon Suh}.} \bibinfo{year}{2020}\natexlab{}.
\newblock \showarticletitle{{Mitigating Curtailment and Carbon Emissions through Load Migration between Data Centers}}.
\newblock \bibinfo{journal}{\emph{Joule}} \bibinfo{volume}{4}, \bibinfo{number}{10} (\bibinfo{year}{2020}), \bibinfo{pages}{2208--2222}.
\newblock
\showISSN{2542-4351}
\urldef\tempurl%
\url{https://doi.org/10.1016/j.joule.2020.08.001}
\showDOI{\tempurl}


\bibitem[Zhou et~al\mbox{.}(2008)]%
        {Zhou:08}
\bibfield{author}{\bibinfo{person}{Yunhong Zhou}, \bibinfo{person}{Deeparnab Chakrabarty}, {and} \bibinfo{person}{Rajan Lukose}.} \bibinfo{year}{2008}\natexlab{}.
\newblock \showarticletitle{{Budget Constrained Bidding in Keyword Auctions and Online Knapsack Problems}}.
\newblock In \bibinfo{booktitle}{\emph{Lecture Notes in Computer Science}}. \bibinfo{publisher}{Springer Berlin Heidelberg}, \bibinfo{pages}{566--576}.
\newblock


\bibitem[Zrigui et~al\mbox{.}(2022)]%
        {Zrigui:2022}
\bibfield{author}{\bibinfo{person}{Salah Zrigui}, \bibinfo{person}{Raphael~Y de Camargo}, \bibinfo{person}{Arnaud Legrand}, {and} \bibinfo{person}{Denis Trystram}.} \bibinfo{year}{2022}\natexlab{}.
\newblock \showarticletitle{{Improving the Performance of Batch Schedulers Using Online Job Runtime Classification}}.
\newblock \bibinfo{journal}{\emph{J. Parallel and Distrib. Comput.}}  \bibinfo{volume}{164} (\bibinfo{year}{2022}), \bibinfo{pages}{83--95}.
\newblock


\end{thebibliography}
